\documentclass[
aps,
physrev,
superscriptaddress,
amsmath,
amssymb,
nofootinbib,
floatfix,
twocolumn,
reprint
]
{revtex4-2}

\usepackage[utf8]{inputenc}
\usepackage{hyperref}
\usepackage{graphicx,color}
\usepackage{url}
\usepackage{bm}
\usepackage{bbm}
\usepackage{braket}
\usepackage{qcircuit}
\usepackage{amsthm}
\usepackage{enumitem}
\usepackage{algorithm}
\usepackage{multirow}
\usepackage{mhchem}

\setlist[description]{leftmargin=\parindent,labelindent=\parindent}

\theoremstyle{plain}
\newtheorem{thm}{Theorem}
\newtheorem{prop}{Proposition}

\usepackage{hyperref}
\hypersetup{
colorlinks=true, 
linkcolor=blue, 
citecolor=blue 
}

\begin{document}

\title{Subspace-Based Local Compilation of Variational Quantum Circuits for Large-Scale Quantum Many-Body Simulation}
\author{Shota Kanasugi}
\affiliation{Quantum Laboratory, Fujitsu Research, Fujitsu Limited., 4-1-1 Kamikodanaka, Nakahara, Kawasaki, Kanagawa 211-8588, Japan}
\email{kanasugi.shota@fujitsu.com}

\author{Yuichiro Hidaka}
\affiliation{QunaSys Inc., Aqua Hakusan Building 9F, 1-13-7 Hakusan, Bunkyo, Tokyo 113-0001, Japan}

\author{Yuya O. Nakagawa}
\affiliation{QunaSys Inc., Aqua Hakusan Building 9F, 1-13-7 Hakusan, Bunkyo, Tokyo 113-0001, Japan}

\author{Shoichiro Tsutsui}
\affiliation{QunaSys Inc., Aqua Hakusan Building 9F, 1-13-7 Hakusan, Bunkyo, Tokyo 113-0001, Japan}

\author{Norifumi Matsumoto}
\affiliation{Quantum Laboratory, Fujitsu Research, Fujitsu Limited., 4-1-1 Kamikodanaka, Nakahara, Kawasaki, Kanagawa 211-8588, Japan}

\author{Kazunori Maruyama}
\affiliation{Quantum Laboratory, Fujitsu Research, Fujitsu Limited., 4-1-1 Kamikodanaka, Nakahara, Kawasaki, Kanagawa 211-8588, Japan}

\author{Hirotaka Oshima}
\affiliation{Quantum Laboratory, Fujitsu Research, Fujitsu Limited., 4-1-1 Kamikodanaka, Nakahara, Kawasaki, Kanagawa 211-8588, Japan}

\author{Shintaro Sato}
\affiliation{Quantum Laboratory, Fujitsu Research, Fujitsu Limited., 4-1-1 Kamikodanaka, Nakahara, Kawasaki, Kanagawa 211-8588, Japan}

\date{\today}

\begin{abstract}
    Simulation of quantum many-body systems is one of the most promising applications of quantum computers.  
    It is crucial to efficiently implement the time-evolution operator as a quantum circuit to execute such simulations on near-term quantum computing devices with limited computational resources. 
    However, standard approaches such as Trotterization sometimes require a deep quantum circuit, which is hard to implement on near-term quantum computers.  
    Here, we propose a hybrid quantum-classical algorithm, called Local Subspace Variational Quantum Compilation (LSVQC), for compiling the time-evolution operator of quantum many-body systems. 
    The LSVQC performs a variational optimization to reproduce the action of the target time-evolution operator within a physically reasonable subspace. 
    The optimization is performed for small local subsystems based on the Lieb-Robinson bound, which allows us to execute the cost function evaluation using small-scale quantum devices and/or classical computers. 
    We demonstrate the validity of the LSVQC algorithm through numerical simulations of a simple spin-lattice model and an effective model of a parent compound of cuprate superconductors, \ce{Sr2CuO3}, constructed by the \textit{ab initio} downfolding method. 
    It is shown that the LSVQC achieves a 95\% reduction of the circuit depth for simulating quantum many-body dynamics compared to the Trotterization at best while maintaining the same computational accuracy.
    We also demonstrate that the restriction to a subspace leads to a substantial reduction of required resources and improved accuracy compared to the case of considering the entire Hilbert space. 
    Furthermore, we estimate the gate count needed to execute the quantum simulations using the LSVQC on near-term quantum computing architectures in the noisy intermediate-scale or early fault-tolerant quantum computing era. 
    Our estimation suggests that the acceptable physical gate error rate for the LSVQC can be about one order of magnitude larger than that for the Trotterization.
\end{abstract}

\maketitle

\section{Introduction}\label{sec:Intro}
Simulating quantum many-body systems is important in various fields of science such as condensed matter physics, nuclear physics, quantum chemistry, and materials science. 
For practical applications such as materials design and drug discovery, however, large-scale quantum simulations are necessary, which are hard tasks since the required resources for classical algorithms increase exponentially with increasing the system size.   
On the other hand, quantum computers generally realize such quantum simulations with exponentially fewer computational resources than classical computers.  
Hence, quantum simulation is expected to be one of the most promising applications to realize practical quantum advantage or quantum utility~\cite{kim2023evidence}. 

A crucial task in quantum simulation is implementing the time-evolution operator as a quantum circuit. 
This is not only a technique to simulate the dynamics of quantum many-body systems according to the Schr\"{o}dinger equation but also is an essential building block of various quantum algorithms such as quantum phase estimation~\cite{kitaev1995quantum,cleve1998quantum,nielsen-chuang}. 
A variety of algorithms have been developed to implement the time-evolution operator on quantum computers efficiently.  
The most common algorithm is the Trotterization~\cite{Lloyd1996universal}, which allows us to implement the time-evolution operator most simply with analytically guaranteed precision. 
However, it often requires many quantum gates for practical applications beyond the capability of current or near-term quantum computers such as noisy intermediate-scale quantum (NISQ) devices~\cite{preskill2018quantum}. 
Increased gate count in Trotterization is also problematic for ideal fault-tolerant quantum computers (FTQCs) since it causes a long computational time that hinders the practical use of quantum computers. 
More advanced algorithms such as the qubitization-based techniques~\cite{low2019hamiltonian} exhibit better asymptotic scaling of gate complexity than Trotterization. 
However, such techniques require many ancilla qubits and gate operations and are hence suitable only for long-term FTQCs. 

For near-term NISQ devices, quantum-classical hybrid variational algorithms have been proposed to realize quantum simulation with shallow-depth quantum circuits. 
Some algorithms perform a variational optimization of parametrized quantum circuits (i.e., ansatz) to represent a time-evolved quantum state based on McLachlan's variational principle~\cite{VQS1_PhysRevX.7.021050,VQS2_yuan2019theory,VQS3_PhysRevLett.125.010501,VQS4_PRXQuantum.2.030307,VQS5_barison2021efficient}. 
Another approach is variational quantum compiling algorithms~\cite{Khatri-QAQC,sharma2020noise,bilek2022recursive,cirstoiu2020variational,gibbs2022long} that aim to compile a Trotterized time-evolution circuit into a more easily implementable form. 
Although these variational algorithms potentially outperform the Trotterization-based approaches in small-scale quantum simulations, it is hard to scale up these algorithms to practically important large-scale quantum simulations due to their variational nature, which induces obstacles such as a barren plateau~\cite{BP_cerezo2021cost,BP_mcclean2018barren}.
Considering the above difficulty of existing algorithms, it is necessary to develop a more sophisticated algorithm to implement the large-scale time-evolution operator on near-term quantum computers such as in the NISQ era. 
It is also important for early fault-tolerant quantum computers (early-FTQCs)~\cite{EFTQC0_PRXQuantum.3.010345,EFTQC1_campbell2021early,EFTQC2_PRXQuantum.3.010318,EFTQC3_kshirsagar2022proving,EFTQC4_ding2022even,EFTQC5_kuroiwa2023clifford+,STAR_PRXQuantum.5.010337}, in which small-scale or partial quantum error correction is implemented but the number of logical qubits and circuit depth are still limited. 

One possible route of designing algorithms to reduce implementation costs of quantum simulations is considering \textit{subspace} relevant to the problem of interest. Indeed, various subspace-based algorithms have been devised to reduce the computational costs needed to perform quantum simulations on both classical and quantum computers~\cite{SM_motta2024subspace}. In the electronic structure problems, the ground- and excited-state properties of quantum many-body systems can be approximately obtained by projecting the Schr\"{o}dinger equation onto a subspace spanned by some physically relevant many-electron states. Diagonalization of such subspace-restricted Schr\"{o}dinger equation is the basic idea of various near-term quantum algorithms to calculate eigenpairs of quantum many-body Hamiltonian such as quantum subspace expansion (QSE)~\cite{QSE_PhysRevA.95.042308,QSE_PhysRevX.8.011021}, quantum equation of motion algorithm~\cite{qEOM_PhysRevResearch.2.043140}, quantum filter diagonalization (QFD)~\cite{QFD_parrish2019quantum}, multireference selected quantum Krylov (MRSQK) algorithm~\cite{MRSQK_stair2020multireference}, quantum Lanczos algorithm~\cite{QLanczos_motta2020determining}, quantum power method~\cite{QuantumPower_PRXQuantum.2.010333}, quantum selected configuration interaction (QSCI)~\cite{QSCI_kanno2023quantum,nakagawa2023adapt-qsci}, and so on. These methods mainly differ in the way of generating the basis states spanning the subspace. 
Subspace consideration is also important for simulating the time evolution of quantum systems. For example, it is shown that the analytical error bound of Trotterization is improved by considering the simulation of quantum states within a low-energy subspace~\cite{TrotSubspace_csahinouglu2021hamiltonian}. Furthermore, diagonalizing time-evolution operators within a physically relevant subspace is a basic concept of some algorithms that aim to realize fast-forwarding time evolution on near-term quantum computers, e.g., subspace variational quantum simulator~\cite{SVQS_PhysRevResearch.5.023078}, fixed-state variational fast-forwarding (fsVFF)~\cite{gibbs2022long}, and classical-quantum fast-forwarding~\cite{CQFF_lim2021fast}. 

In this paper, we propose a quantum-classical hybrid algorithm called Local Subspace Variational Quantum Compilation (LSVQC). 
The LSVQC is a variant of the variational quantum compiling designed to compile the time-evolution operators of large-scale quantum many-body systems using small-scale near-term quantum devices and/or classical computers. 
The main advantages of the LSVQC over other variational algorithms are the following: 
\begin{enumerate}
\renewcommand{\labelenumi}{(\roman{enumi})}
    \item Reduction of requirement on the expressive power of the ansatz circuit to reproduce the time evolution.  
    \item Scalability to practically important large-scale quantum simulations. 
\end{enumerate} 
These advantages are owing to two features of the LSVQC, i.e., the \textit{subspace-based compilation} and the \textit{local compilation}. 
First, the LSVQC is performed to reproduce the action of the target time-evolution operator only within a physically relevant subspace (i.e., subspace-based compilation), which is designed based on the properties of physical quantities we wish to simulate. 
The subspace-based compilation is motivated by the fact that the time evolution of a quantum many-body state often stays within a small subspace in the entire Hilbert space such as low-energy subspace or specific symmetry sectors. In comparison with the case of full Hilbert space compiling, the subspace-based compiling leads to the reduction of the depth and number of parameters (i.e., expressive power) of the ansatz circuit required to accurately approximate the time evolution. 
It enhances the feasibility of quantum simulations on near-term quantum computers with a limited number of quantum gates and circuit depth. 
Second, the LSVQC performs the variational optimization for smaller-size local subsystems of the size $\mathcal{O}(L^0)$ or $\mathcal{O}(\log{L})$ instead of the large-scale entire size-$L$ system. 
The physics behind such local compilation protocol is the Lieb-Robinson (LR) bound~\cite{lieb1972finite}, which dictates the universal causality of quantum many-body systems with local interactions. 
We formally derive the theoretical guarantee for the local compilation in a similar way as the local variational quantum compilation (LVQC), which was originally proposed in Ref.~\cite{LVQC_PRXQuantum.3.040302} and later extended to fermionic simulation in our previous work~\cite{LVQC_GF_PhysRevResearch.5.033070}. 
Note that the validity of such a local compilation protocol is intuitively nontrivial. 
Since the computationally hard large-scale optimization is replaced with a smaller-scale optimization requiring an ansatz circuit with low expressive power, the LSVQC is less susceptible to barren plateaus and is scalable to large-scale quantum simulations in contrast to other variational algorithms. 

We demonstrate the validity of the LSVQC by performing numerical simulations of the one-dimensional (1D) Heisenberg model and an \textit{ab initio} effective model of \ce{Sr2CuO3}~\cite{SCO_AFM_PhysRevB.51.5994,SCO_AFM_PhysRevLett.76.3212,SCO_separation_PhysRevLett.81.657,SCO_separation_PhysRevB.59.7358,SCO_SC_liu2014new}, which is a quasi-one-dimensional strongly correlated material. 
It is demonstrated that the LSVQC allows us to accurately simulate quantum many-body dynamics using circuits shallower by about 90-95\% than the Trotterization while maintaining accuracy. 
We also estimate the gate count needed to perform the quantum simulation of strongly correlated materials. 
The results suggest that the LSVQC significantly reduces the hardware requirement to execute the quantum simulation in the NISQ and early-FTQC era~\cite{EFTQC0_PRXQuantum.3.010345,EFTQC1_campbell2021early,EFTQC2_PRXQuantum.3.010318,EFTQC3_kshirsagar2022proving,EFTQC4_ding2022even,EFTQC5_kuroiwa2023clifford+,STAR_PRXQuantum.5.010337}. 

The rest of the paper is organized as follows. 
In Sec.~\ref{sec:overview}, we give an overview of the LSVQC algorithm. 
In Sec.~\ref{sec:formulation}, we provide a theoretical formulation of the subspace-based compilation and some guidelines for subspace design. 
In Sec.~\ref{sec:LSVQC_derivation}, we derive a theorem that guarantees the local compilation protocol of the LSVQC. 
In Sec.~\ref{sec:numerical_test}, we demonstrate the basic properties of the LSVQC by performing numerical simulations for the 1D Heisenberg model. 
In Sec.~\ref{sec:DF}, we apply the LSVQC to quantum many-body dynamics simulations of the \textit{ab initio} effective model of \ce{Sr2CuO3}. 
In Sec.~\ref{sec:resource}, we perform resource estimation of the gate count needed to simulate strongly correlated materials. 
Finally, we summarize our work in Sec.~\ref{sec:Summary}.

\section{Overview of the LSVQC algorithm}\label{sec:overview}
\begin{figure*}[tbp]
    \includegraphics[width=17.8cm]{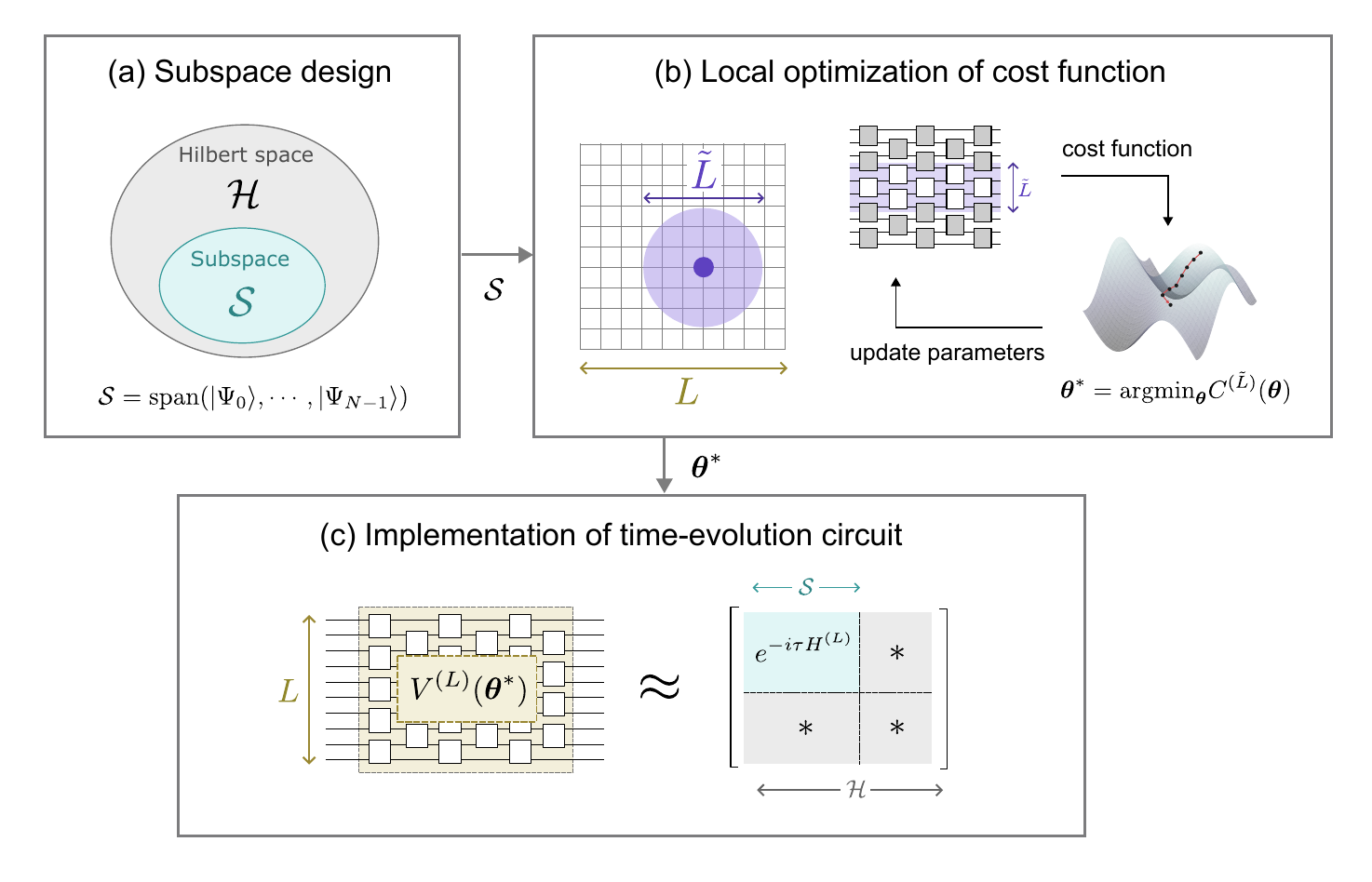}
    \caption{Overview of the LSVQC algorithm. 
    (a) First, we design a subspace $\mathcal{S}$ based on the properties of the physical quantities we aim to simulate. 
    (b) Next, we perform variational optimization of a cost function $C^{(\tilde{L})}(\bm{\theta})$ defined on size-$\tilde{L}$ local subsystems, where $\tilde{L}\leq L$ with $L$ being the whole system size. 
    The cost function is evaluated by using size-$\tilde{L}$ localized quantum circuits (represented as white two-qubit gates in the purple region).  
    The optimization loop is executed on $\tilde{L}$-qubit quantum computing devices and/or classical computers. 
    Then, we obtain an optimal parameter set $\bm{\theta}^*$ that minimizes the subsystem cost function $C^{(\tilde{L})}(\bm{\theta})$. 
    (c) Finally, we implement an optimized ansatz $V^{(L)}(\bm{\theta}^*)$ on a $L$-qubit quantum computing device. 
    When the optimization is successfully converged, $V^{(L)}(\bm{\theta}^*)$ approximates the action of the target time-evolution unitary $e^{-i\tau H^{(L)}}$ within the subspace $\mathcal{S}$. } 
    \label{fig:LSVQC_overview}
\end{figure*}
In this section, we provide an overview of our LSVQC algorithm. 
The LSVQC aims to implement the time-evolution operator $e^{-i\tau H^{(L)}}$ for a $L$-qubit Hamiltonian $H^{(L)}$ with a fixed time $\tau$. 
Specifically, the LSVQC compiles the target unitary $e^{-i\tau H^{(L)}}$ using an ansatz $V^{(L)}(\bm{\theta})$ as follows (see Fig.~\ref{fig:LSVQC_overview}):
\begin{enumerate}
\renewcommand{\labelenumi}{(\roman{enumi})}
    \item Design a subspace $\mathcal{S}$ appropriate to physical quantities we aim to simulate. Here, we also specify a set of quantum circuits that prepare the basis states spanning the subspace.
    \item Construct localized versions of the target time-evolution operator, ansatz, and state preparation circuits for local subsystems of size $\tilde{L}(\leq L)$. Then, variationally optimize a cost function defined on the local subsystems and find an optimal parameter set $\bm{\theta}^*$ minimizing the cost function.  
    \item Construct an optimized ansatz circuit $V^{(L)}(\bm{\theta}^*)$ on a $L$-qubit quantum computer using the optimal parameter set $\bm{\theta}^*$. The optimized circuit $V^{(L)}(\bm{\theta}^*)$ approximates the time-evolution operator $e^{-i\tau H^{(L)}}$. 
\end{enumerate} 
In step (i), the structure of the subspace is determined by details of the target physical quantities we aim to simulate. 
A guideline of the subspace design for some typical situations encountered in quantum many-body simulations is given in Sec.~\ref{sec:subspace}.  
In step (ii), the compilation size $\tilde{L}$, which typically scales as $\mathcal{O}(L^0)$ or $\mathcal{O}(\log{L})$, is estimated based on the LR bound as shown in Sec.~\ref{sec:LSVQC_derivation}. 
The optimization is executed by performing the Loschmidt echo test (LET)~\cite{sharma2020noise} using $\tilde{L}$-qubit quantum devices (see Fig.~\ref{fig:Loschmidt_echo} for the quantum circuit) and/or classical computers. 
The theoretical guarantee for the local compilation protocol is given in Sec.~\ref{sec:LSVQC_derivation}. 
Note that we can perform this local compilation protocol using only classical computers when the compilation size $\tilde{L}$ is sufficiently small such that the cost function can be easily calculated on classical computers. In this case, we can execute ideal optimization without physical and statistical errors inherent to quantum computers. 
When the optimization is successfully converged, the output quantum circuit $V^{(L)}(\bm{\theta}^*)$ obtained in step (iii) accurately reproduces the action of the target time-evolution operator only within the subspace $\mathcal{S}$. 
In other words, the optimized circuit $V^{(L)}(\bm{\theta}^*)$ returns a correct output state only when it is applied to a state belonging to the target subspace. 

Here, we provide some comments on related works. 
The idea of subspace-based compiling is also employed in the fsVFF algorithm~\cite{gibbs2022long}, which is a variant of the variational quantum compiling for fast-forwarding quantum simulation. 
However, the target cost function of the fsVFF, which is designed to realize subspace diagonalization of time-evolution unitary for various time steps, is different from that of the LSVQC. In addition, the fsVFF is valid only for fast-forwardable Hamiltonians, which is a very restricted class of quantum many-body Hamiltonians~\cite{FF_atia2017fast,FF_berry2007efficient,FF_gu2021fast,FF_loke2017efficient}. In contrast, the LSVQC is a more general algorithm applicable to a broader class of Hamiltonians with the LR bound. Specifically, the LR bound exists in broad quantum many-body systems with finite-ranged,  short-ranged,  and long-ranged interactions in generic dimensions~\cite{LR_nachtergaele2006lieb,LR_nachtergaele2006propagation,LR_hastings2006spectral,LR_PhysRevLett.114.157201,LR_matsuta2017improving,LRbound_PhysRevA.101.022333,LR_PhysRevX.10.031010,LR_PhysRevLett.127.160401}.  

The local compilation protocol similar to the LSVQC is also adopted in the LVQC~\cite{LVQC_PRXQuantum.3.040302,LVQC_GF_PhysRevResearch.5.033070}. However, the LVQC is designed to compile the many-body time-evolution operators on the entire Hilbert space. Such full Hilbert space compiling is performed by optimizing a cost function related to the average gate fidelity, which is measured by the Hilbert–Schmidt test~\cite{Khatri-QAQC}. 
Since the Hilbert–Schmidt test utilizes the maximally entangled Bell states, it requires qubits of twice the compilation size (i.e., $2\tilde{L}$ qubits) and many nonlocal CNOT gates connecting the Bell pairs. This is problematic for near-term quantum devices, in which the number of qubits and qubit connectivity is highly restricted, and possibly causes large errors in evaluating the cost functions. On the other hand, the LET utilized in the LSVQC does not require doubling the number of qubits and highly nonlocal two-qubit gates connecting the Bell pairs (see Fig.~\ref{fig:Loschmidt_echo}). 
Thus, compared to the LVQC, the LSVQC is more suitable for near-term quantum computers like NISQ devices or early-FTQCs, where the available number of qubits and gate operations are restricted. 
Furthermore, compared to the full Hilbert space compiling of the LVQC, we can expect that the subspace-based compiling of the LSVQC leads to a reduction of the expressive power of the ansatz needed to reproduce the time evolution of a specific physical quantity. 
In other words, the circuit depth and the number of gates of the ansatz required for successful optimization of the LSVQC are smaller than those of the LVQC. 
We provide numerical support for this intuition in Sec.~\ref{sec:DF}. 

\section{Formulation for subspace-based compilation}\label{sec:formulation}
This section provides several formulations regarding the subspace-based compiling in the LSVQC algorithm. 

\subsection{Cost functions}
We first introduce the cost function that we aim to minimize through the LSVQC algorithm. 
We compile the time evolution operator variationally into a low-depth quantum circuit within a $N$-dimensional subspace $\mathcal{S}$, 
\begin{align}
    \mathcal{S} &= \mathrm{span}(\ket{\Psi_0}, \cdots, \ket{\Psi_{N-1}}),
\end{align}
where $\{\ket{\Psi_k}\}_{k=0}^{N-1}\equiv B_\mathcal{S}$ are some basis states. 
Specifically, we impose that $B_\mathcal{S}=\{\ket{\Psi_k}\}_{k=0}^{N-1}$ is a linearly independent and nonorthogonal basis set for later convenience. 
To accomplish this task, we define the cost function
\begin{align}
    C_{\rm LET}^{B_\mathcal{S}}(U^{(L)}_{\tau},V^{(L)}) &= 1 - \frac{1}{N}\sum_{k=0}^{N-1}\left| 
\bra{\Psi_k}(V^{(L)})^{\dag}U_\tau^{(L)}\ket{\Psi_k} \right|^2 , \label{eq:cost_LET}
\end{align}
where $L$ denotes the number of qubits, $\tau$ is a fixed time, $U^{(L)}_{\tau}=e^{-i\tau H^{(L)}}$ is the time evolution operator for a given Hamiltonian $H^{(L)}$, and $V^{(L)}(\bm{\theta})$ is the ansatz quantum circuit with $\bm{\theta}$ being a set of variational parameters. 
This cost function quantifies the overlap between the exact time-evolved state $U^{(L)}_{\tau}\ket{\Psi_k}$ and the ansatz state $V^{(L)}\ket{\Psi_k}$ averaged over the basis set spanning the subspace $\mathcal{S}$. 
The cost function $C_{\rm LET}^{B_\mathcal{S}}$ becomes zero if and only if $U^{(L)}_{\tau}\ket{\Psi_k}=e^{i\varphi}V^{(L)}\ket{\Psi_k}$ for all $k$. Note that the global phase factor takes a common value independent of $k$ owing to the linear independence and nonorthogonality of the basis states (see Appendix~\ref{append:subspace_cost} for details). Hence, the time evolution of any states belonging to the subspace $\mathcal{S}$, which are described by linear combinations of the basis states, can be reproduced. 
We also assume that the basis states $\{\ket{\Psi_k}\}_{k=0}^{N-1}$ can be efficiently prepared on a quantum computer by using a set of state preparation circuits $\{W_k^{(L)}\}_{k=0}^{N-1}$ as $\ket{\Psi_k}=W_k^{(L)}\ket{0}^{\otimes L}$. 
Then, the cost $C_{\rm LET}^{B_\mathcal{S}}$ can be measured by performing the LET~\cite{sharma2020noise} on $L$-qubit quantum circuits shown in Fig.~\ref{fig:Loschmidt_echo}(a).

Although the cost function $C_{\rm LET}^{B_\mathcal{S}}$ is a natural quantity describing the difference between $U^{(L)}_{\tau}$ and $V^{(L)}(\bm{\theta})$ within the target subspace $\mathcal{S}$, it suffers from the exponentially vanishing gradient (i.e., barren plateaus) for a large system size since $C_{\rm LET}^{B_\mathcal{S}}$ is a global cost defined on the whole system~\cite{BP_mcclean2018barren,BP_cerezo2021cost}.  
To avoid the barren plateau, we can alternatively use a local version of the cost function defined as 
\begin{align}
    &C_{\rm LLET}^{B_\mathcal{S}}(U^{(L)}_{\tau},V^{(L)}) = \frac{1}{L}\sum_{j=1}^{L}C_{\rm LLET}^{(j),B_\mathcal{S}}(U^{(L)}_{\tau},V^{(L)}), \label{eq:cost_LLET} 
\end{align}
where 
\begin{align}
    &C_{\rm LLET}^{(j),B_\mathcal{S}}(U^{(L)}_{\tau},V^{(L)}) = 1 - \frac{1}{N}\sum_{k=0}^{N-1}\mathrm{Tr}[\Pi_j\rho_k(U^{(L)}_{\tau},V^{(L)})], \label{eq:cost_LLET_j} \\
    &\rho_k(U^{(L)}_{\tau},V^{(L)}) = (\tilde{V}_k^{(L)})^{\dag}U_{\tau}^{(L)}\ket{\Psi_k}\bra{\Psi_k}(U_{\tau}^{(L)})^{\dag}\tilde{V}_k^{(L)}. \label{eq:cost_LLET_rho}
\end{align}
In Eqs.~\eqref{eq:cost_LLET_j} and~\eqref{eq:cost_LLET_rho}, $\tilde{V}_k^{(L)} = V^{(L)} W_k^{(L)}$ and $\Pi_j = \mathbbm{1}_1\otimes\cdots\otimes\ket{0}\bra{0}_j \otimes\cdots\otimes \mathbbm{1}_L$ is the projection operator for $j$-th qubit ($j=1,2,\cdots,L$). 
The local cost function $C_{\rm LLET}^{B_\mathcal{S}}$ can be measured on a quantum computer by performing the local Loschmidt echo test (LLET)~\cite{sharma2020noise} using a $L$-qubit quantum circuit shown in Fig.~\ref{fig:Loschmidt_echo}(b). 
Since the LLET only performs a local measurement at $j$-th qubit to obtain $C_{\rm LLET}^{(j),B_\mathcal{S}}$, the local cost function $C_{\rm LLET}^{B_\mathcal{S}}$ avoids the barren plateau under the assumption of using a shallow-depth ansatz~\cite{gibbs2022long}. 
The local cost function $C_{\rm LLET}^{B_\mathcal{S}}$ and global cost function $C_{\rm LET}^{B_\mathcal{S}}$ are connected through the following inequality~\cite{Khatri-QAQC,sharma2020noise}:
\begin{align}
    C_{\rm LLET}^{B_\mathcal{S}} \leq C_{\rm LET}^{B_\mathcal{S}} \leq L \cdot C_{\rm LLET}^{B_\mathcal{S}}, 
    \label{eq:cost_ineq_general}
\end{align}
where we omitted the arguments $(U_{\tau}^{(L)},V^{(L)})$ for brevity. 
From Eq.~\eqref{eq:cost_ineq_general}, we can safely say that $C_{\rm LLET}^{B_\mathcal{S}}=0$ if and only if $C_{\rm LET}^{B_\mathcal{S}}=0$. 
Therefore, we can alternatively minimize the local cost $C_{\rm LLET}^{B_\mathcal{S}}$ instead of minimizing the global cost $C_{\rm LET}^{B_\mathcal{S}}$. 

When the system size $L$ is small enough and the target time evolution $U^{(L)}_\tau$ can be efficiently implemented (e.g., using Trotterization), we can directly optimize the cost functions $C_{\rm LET}^{B_\mathcal{S}}$ or $C_{\rm LLET}^{B_\mathcal{S}}$ by performing the LET or LLET on $L$-qubit quantum devices. 
However, it becomes difficult to directly perform the LET or LLET on $L$-qubit quantum devices when the system size $L$ becomes large. 
To circumvent such issues, we adopt the LSVQC algorithm that performs variational optimization on small local subsystems. 
In the LSVQC, not only measurement but also the size of the quantum circuit is localized in contrast to the LLET. 
In Sec.~\ref{sec:LSVQC_derivation}, we provide a theorem that guarantees such local optimization protocol. 
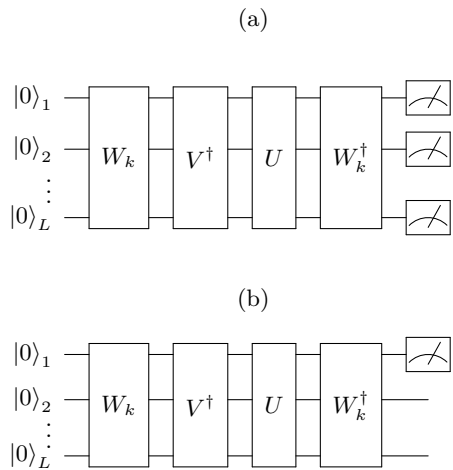
\begin{figure}[htbp]
    \centerline{{(a)} }
    \[
    \Qcircuit @C=1em @R=.7em {
      \lstick{\ket{0}_1} &  \multigate{3}{W_k} & \multigate{3}{V^{\dag}} & \multigate{3}{U} & \multigate{3}{W^{\dag}_k} & \meter \\
      \lstick{\ket{0}_2} & \ghost{W_k} & \ghost{V^{\dag}} & \ghost{U} & \ghost{W^{\dag}_k} & \meter \\
      \lstick{\vdots} & & & & &  \\
      \lstick{\ket{0}_L} & \ghost{W_k} & \ghost{V^{\dag}} & \ghost{U} & \ghost{W^{\dag}_k} & \meter
    }
    \]
    \vspace{0.25cm}
    \centerline{{(b)} }
    \[
    \Qcircuit @C=1em @R=.7em {
      \lstick{\ket{0}_1} &  \multigate{3}{W_k} & \multigate{3}{V^{\dag}} & \multigate{3}{U} & \multigate{3}{W^{\dag}_k} & \meter \\
      \lstick{\ket{0}_2} & \ghost{W_k} & \ghost{V^{\dag}} & \ghost{U} & \ghost{W^{\dag}_k} & \qw  \\
      \lstick{\vdots} & & & & &  \\
      \lstick{\ket{0}_L} & \ghost{W_k} & \ghost{V^{\dag}} & \ghost{U} & \ghost{W^{\dag}_k} & \qw
    }
    \]
    \caption{
    (a) Quantum circuit for the LET. The probability of obtaining the measurement outcome in which all $L$ qubits are in the $\ket{0}$ state is equal to $|\bra{\Psi_k}(V^{(L)})^{\dag}U^{(L)}_\tau\ket{\Psi_k}|^2$.  
    (b) Quantum circuit for the LLET. The probability of obtaining the measurement outcome in the $j$-th qubit ($j=1,2,\cdots,L$) is in the $\ket{0}$ state is equal to $\mathrm{Tr}[\Pi_j\rho_k(U^{(L)}_\tau,V^{(L)})]$ with $\rho_k$ given by Eq.~\eqref{eq:cost_LLET_rho}. Note that the figure is for $j=1$. } 
    \label{fig:Loschmidt_echo}
\end{figure}

\subsection{Guideline for subspace design}\label{sec:subspace}
Here, we provide a guideline for designing the subspace appropriate to typically encountered situations in quantum many-body simulations.  

\subsubsection{Dynamics of a fixed input state}
One of the most common situations is simulating the dynamics of a physical observable $A$ for a fixed input state $\ket{\psi}$ as 
\begin{align}
    \braket{A(t)}_{\psi}=\bra{\psi}e^{itH}Ae^{-itH}\ket{\psi}, 
\end{align}
where $H$ is the Hamiltonian and $t$ is a time.
The input state $\ket{\psi}$ can be generally expanded in the energy eigenbasis as
\begin{align}
    \ket{\psi}=\sum_{m=1}^{N_{\rm eig}}a_m\ket{E_m}, 
\end{align}
where $\{\ket{E_m}\}_{m=1}^{N_{\rm eig}}$ is a set of energy eigenstates and $a_m$ is a nonzero complex coefficient. 
Then, the dynamics $\braket{A(t)}_{\psi}$ can be expressed as 
\begin{align}
    \braket{A(t)}_{\psi} 
    &= \sum_{m,m'}^{N_{\rm eig}}a_m^* a_{m'}e^{i(E_m-E_{m'})t}\bra{E_m}A\ket{E_{m'}}.
    \label{eq:dynamics_A}
\end{align}
Equation~\eqref{eq:dynamics_A} indicates that the time evolution of $\braket{A(t)}_{\psi}$ is enclosed in the $N_{\rm eig}$-dimensional subspace spanned by the energy eigenstates $\{\ket{E_m}\}_{m=1}^{N_{\rm eig}}$. 
This means that $N_{\rm eig}$ linearly-independent states are required to perfectly learn the action of the time evolution $e^{-itH}$ on the input state $\ket{\psi}$~\cite{poland2020no,gibbs2022long}.
From this fact, we naively expect that one appropriate choice of the subspace is the following Krylov subspace generated by the real-time evolution,  
\begin{align}
    \mathcal{S}_{\psi}(N_t) = \mathrm{span}\left( \{U(t_n) \ket{\psi}\}_{n=0}^{N_{t}} \right),
    \label{eq:subspace_A}
\end{align}
where $U(t_n)\approx e^{-it_nH}$ is the approximate time-evolution operator, $t_n=n\Delta_t$ is a time step with some interval $\Delta_t$, and $N_t$ is an integer satisfying $0 \leq N_t \leq N_{\rm eig}$. 
Although we can consider other subspaces (e.g., linear response subspace in Ref.~\cite{QSE_PhysRevA.95.042308,QSE_PhysRevX.8.011021}), we adopt the above Krylov subspace as a natural choice for quantum computers since the time-evolved states $\{U(t_n) \ket{\psi}\}$ can be prepared on quantum circuits with polynomial cost and controllable accuracy. Note that the real-time Krylov subspace is indeed utilized in other subspace-based quantum algorithms such as the QFD~\cite{QFD_parrish2019quantum}, MRSQK~\cite{MRSQK_stair2020multireference}, and fsVFF~\cite{gibbs2022long}.
The time-evolved states $\{U(t_n) \ket{\psi}\}$ in general form a linearly independent and nonorthogonal basis set except for very special parameter settings, and hence appropriate for the LSVQC. 
Although the parameter $N_t$ should be chosen to be equal to $N_{\rm eig}$ ideally, we allow $N_t<N_{\rm eig}$ considering the common situations such that only some eigenstates have dominant amplitudes. 
The value of $N_{\rm eig}$ can be estimated by computing the determinant of the Gramian matrix as proposed in Ref.~\cite{gibbs2022long}. 

Since the required number of quantum circuits for the LET or LLET increases with the dimension of the subspace, exponential growth of $N_{\rm eig}$ against the system size $L$ is undesirable for efficient computation. 
Therefore, we only consider the case $N_{\rm eig}$ scales only polynomially with the system size $L$. 
This is a reasonable assumption for many circumstances since many physically relevant quantum many-body states are confined within an exponentially small subspace of the entire Hilbert space~\cite{QsimSpace_PhysRevLett.106.170501}. 
Even when $N_{\rm eig}$ scales exponentially with the system size $L$, it is possible that the state $\ket{\psi}$ has a dominant overlap with only a small number of the energy eigenstates and we can set $N_t=\mathcal{O}(\mathrm{poly}L)$~\cite{gibbs2022long}. 
Note that such polynomial size subspace does not necessarily lie in a regime efficiently simulatable by classical computers when $L$ is large.

We can implement the approximate time-evolution operator $U(t_n)$ by adopting the first-order Trotterization. 
We here note that, in near-term quantum computing devices, it is required that the depth of $U(t_n)$ is within the coherence time. 
In addition, the LSVQC itself requires that the depth of the state preparation circuits is shallow to ensure the local compilation theorem derived in Sec.~\ref{sec:LSVQC_derivation} (see Eq.~\eqref{eq:L_compilation} for details).  
Considering these restrictions, we implement $U(t_n)$ as a single step Trotterized unitary for all $t_n$ in numerical simulations in Sec.~\ref{sec:numerical_test} and~\ref{sec:DF}. 
Since such approximation generally causes a large Trotter error and leakage from the ideal Krylov subspace when $t_n$ is large, a small value of $t_n$ is desirable. On the other hand, the basis states in the Krylov subspace tend to be linearly dependent for small $t_n$ since the unitary $U(t_n)$ approaches the identity operator. Such a situation is undesirable since the information on the ideal Krylov subspace is lost. 
Therefore, we choose a moderate value of $t_n$ considering the above tradeoff between the Trotter error and linear independence of basis states. 
We demonstrate such a tradeoff in the numerical simulations in Sec.~\ref{sec:numerical_test}. 

\subsubsection{Dynamics of correlation functions}
Another common situation in quantum many-body physics is simulating the dynamical correlation function for a given quantum state $\ket{\psi}$,  
\begin{align}
    \braket{A(t)B}_{\psi} 
    &= \bra{\psi} e^{itH}Ae^{-itH} B \ket{\psi},
    \label{eq:dynamics_AB}
\end{align}
where $A$ and $B$ are some quantum mechanical operators (not necessarily Hermitian). 
For example, when $A$ and $B$ are fermionic annihilation and creation operators, respectively, $\braket{A(t)B}_{\psi}$ becomes the Green's function (GF) which is treated in Sec.~\ref{sec:DF}. 
Equation~\eqref{eq:dynamics_AB} is regarded as a transition amplitude of the physical quantity $A$ for the time-evolved states $e^{-itH}\ket{\psi}$ and $e^{-itH}B\ket{\psi}$. 
Thus, in the same spirit of the subspace~\eqref{eq:subspace_A}, we can choose a real-time Krylov subspace for two states $\ket{\psi}$ and $B\ket{\psi}$ as a subspace appropriate for computing the dynamical correlation function $\braket{A(t)B}_{\psi}$. 
Specifically, such Krylov subspace is defined as 
\begin{align}
    \mathcal{S} &= \mathrm{span}\left( \{U(t_n) \ket{\psi}, U(t_n) B\ket{\psi}\}_{n=0}^{N_{t}} \right), 
    \label{eq:subspace_AB}
\end{align}
where the definition of $U(t_n)$ and $N_t$ are same with Eq.~\eqref{eq:subspace_A}. 
In contrast to the subspace~\eqref{eq:subspace_A}, the subspace~\eqref{eq:subspace_AB} might not satisfy the requirement of liner independence and nonorthogonality depending on the details of the operators $A$ and $B$. 
In such cases, we moderately change the choice of basis states to fulfill the liner independence and nonorthogonality conditions. 
An example of such transformation is shown in Sec.~\ref{subsec:GF} for the GF. 
We note that a subspace similar to Eq.~\eqref{eq:subspace_AB} is indeed used in Ref.~\cite{jamet2022quantum} to perform the QSE for GF computation. 

\section{Local compilation theorem}\label{sec:LSVQC_derivation}
Here, we derive the {\it local compilation theorem} that guarantees the local optimization protocol in the LSVQC algorithm (procedure (b) in Fig.~\ref{fig:LSVQC_overview}). 
Similarly to the discussion in Ref.~\cite{LVQC_PRXQuantum.3.040302}, the local compilation theorem is derived in the following two steps: 
\begin{enumerate}
\renewcommand{\labelenumi}{(\roman{enumi})}
    \item Local restriction of the target time-evolution operator by the LR bound.
    \item Local restriction of the ansatz and state preparation circuits based on the causal cone. 
\end{enumerate}
We consider the local restriction of not only the target time-evolution unitary and ansatz but also the state preparation circuit generating the subspace basis states. This is the main difference from the discussion of Ref.~\cite{LVQC_PRXQuantum.3.040302}, in which only the local restriction of the target time-evolution unitary and ansatz is discussed. 

First, we clarify the setup and notation.  
We consider quantum many-body Hamiltonian $H^{(L)}$ defined on a size-$L$ lattice $\Lambda$ (i.e., $L=|\Lambda|$), 
\begin{align}
    H^{(L)} = \sum_{X \subseteq \Lambda} H_X, \label{eq:H_local}
\end{align} 
where $H_X$ denotes a term that nontrivially acts on a domain $X \subseteq \Lambda$. 
For simplicity of discussion, we only consider spin systems in this section although the extension to the fermionic Hamiltonian is straightforward as shown in Ref.~\cite{LVQC_GF_PhysRevResearch.5.033070}. 
We assume that $H^{(L)}$ consists of only local terms and finite-range interactions. 
This leads to the LR bound, which is formally expressed for any local observables $O_{X,Y}$ acting on domain $X,Y\subseteq\Lambda$ as follows~\cite{LR_hastings2006spectral}: 
\begin{align}
    \left\|\left[(U^{(L)}_{\tau})^{\dag}O_XU^{(L)}_{\tau}, O_Y\right] \right\| \leq C e^{-[\mathrm{dist}(X,Y)-v|\tau|]/\xi}, \label{eq:LR_bound}
\end{align}
where $U^{(L)}_{\tau}=e^{-i\tau H^{(L)}}$ for a fixed time $\tau$, $\mathrm{dist}(X,Y)$ is the distance between the domains, and $\|\cdot\|$ denotes the operator norm. 
The velocity $v$, length $\xi$, and constant $C$ are determined by the properties of $H^{(L)}$ and irrespective of the system size $L$. 
The constant $C$ typically linearly increases as a function of $\tau$.  
For simplicity, we consider an ansatz $V^{(L)}(\bm{\theta})$ with a brick-wall structure (see Fig.~\ref{fig:LSVQC_circuit}),
\begin{align}
    &V^{(L)}(\bm{\theta}) \nonumber\\
    &= \prod_{l=1}^{d_V}\left[ \left(\prod_{i=1}^{L/2}v^{(l)}_{2i-1,2i}(\theta_{2i-1}^{(l)}) \right)\left(\prod_{i=1}^{L/2-1}v^{(l)}_{2i,2i+1}(\theta_{2i}^{(l)}) \right)  \right], \label{eq:V_L}
\end{align}
where $d_V$ denotes the depth of the ansatz and $\bm{\theta}=\{\theta_{i}^{(l)}\}_{i,l}$ is a set of variational parameters. 
The two-qubit gates $v_{i,i+1}^{(l)}$ acts nontrivially only on the sites $i$ and $i+1$. 
We also impose the brick-wall structure on the state preparation circuits $W^{(L)}_{k}$ as  
\begin{align}
    W^{(L)}_{k} &= \prod_{l=1}^{d_{W_k}}\left[  \left(\prod_{i=1}^{L/2}w_{2i-1,2i}^{(k,l)} \right) \left(\prod_{i=1}^{L/2-1}w_{2i,2i+1}^{(k,l)} \right) \right], \label{eq:W_L}
\end{align}
where $d_{W_k}$ denotes the depth and $w_{i,i+1}^{(k,l)}$ acts nontrivially only on the sites $i$ and $i+1$. 
Although we assume brick-wall structures to $V^{(L)}(\bm{\theta})$ and $\{W^{(L)}_{k}\}_{k=0}^{N-1}$ for simplicity, the local compilation theorem can be derived to other cases as long as these circuits have local structure. 
\begin{figure}[tbp]
    \includegraphics[width=8.5cm]{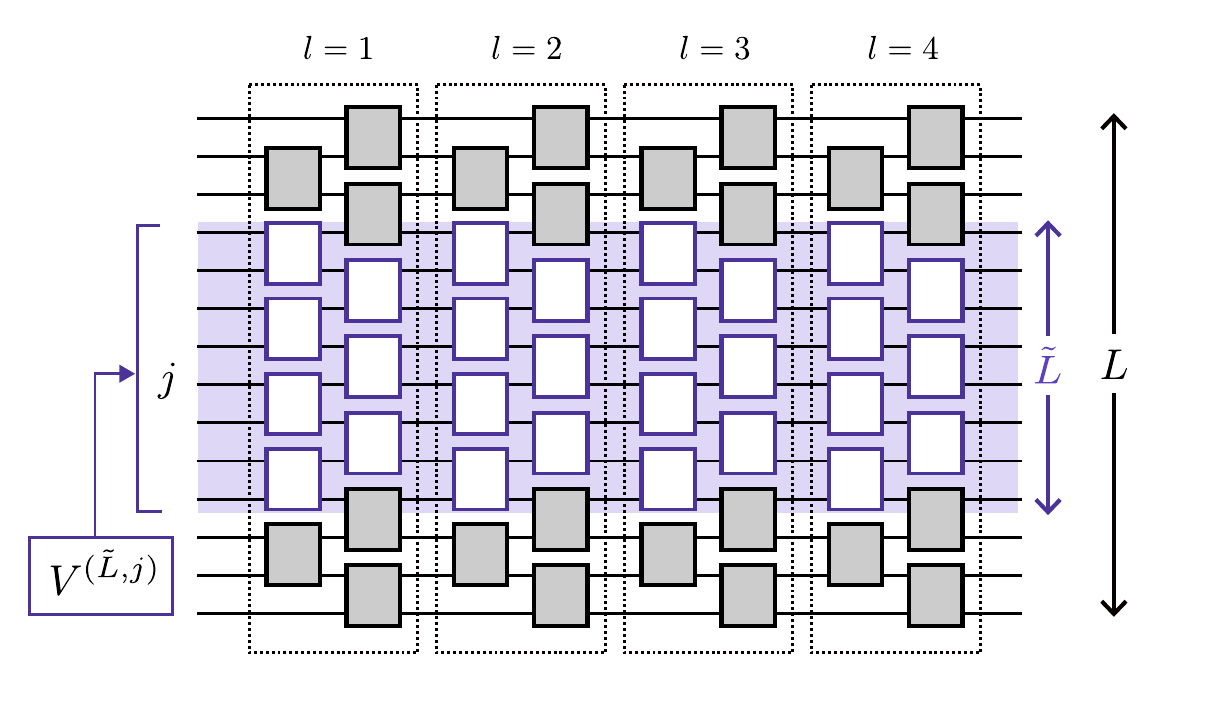}
    \caption{Schematic illustration of the brickwork-structured ansatz $V^{(L)}(\bm{\theta})$ and its local restriction $V^{(\tilde{L},j)}(\bm{\theta})$. The purple area represents the $j$-centered $\tilde{L}$-size local domain $\Lambda^{(\tilde{L},j)}$. Only white two-qubit gates contribute to the restricted ansatz $V^{(\tilde{L},j)}(\bm{\theta})$.  } 
    \label{fig:LSVQC_circuit}
\end{figure}

Now, we proceed to the step (i). 
We define a local restriction of the Hamiltonian $H^{(L)}$ as 
\begin{align}
    H^{(L',j)} &= \sum_{X; X \subseteq \Lambda^{(L',j)}} H_X, \label{eq:H_localized} 
\end{align}
where $\Lambda^{(L',j)}$ is a $j$-centered $L'$-size lattice defined as 
\begin{align}
    \Lambda^{(L',j)} &= \left\{ j'\in \Lambda | \mathrm{dist}(j,j')\leq L'/2 \right\}. \label{eq:lambda_localized}
\end{align}
For the restricted Hamiltonian $H^{(L',j)}$, we define the corresponding time-evolution operator as
\begin{align}
    U^{(L',j)}_{\tau} &= e^{-iH^{(L',j)}\tau}\otimes \mathbbm{1}_{\Lambda \setminus \Lambda^{(L',j)}}. \label{eq:U_localized}
\end{align} 
The restricted time-evolution operator $U^{(L',j)}_{\tau}$ and the target time-evolution operator $U_{\tau}^{(L)}=e^{-i\tau H^{(L)}}$ is related by the following proposition. 
\begin{prop}[Local restriction of the target time-evolution operator] \label{prop1}
    Let us define the restriction size $L'$ as 
    \begin{align}
        L' = 2\left(l_0 + r_H + v\tau + 2\left(d_V + \max_{k}{d_{W_k}}\right) \right), \label{eq:L_restrict}
    \end{align}
    where $r_H$ denotes the range of interaction, $v$ is the LR velocity in Eq.~\eqref{eq:LR_bound}, and $l_0$ is a tunable parameter. 
    Here, we impose that the parameters $\tau$, $d_V$, and $d_{W_k}$ are chosen such that $L' < L$. 
    Then, the local cost function $C_{\rm LLET}^{(j),\mathcal{S}}$ satisfies the following inequality: 
    \begin{align}
        C_{\rm LLET}^{(j),B_\mathcal{S}}(U^{(L)}_{\tau},V^{(L)}) \leq C_{\rm LLET}^{(j),B_\mathcal{S}}(U^{(L',j)}_{\tau},V^{(L)}) + \frac{1}{2}\epsilon_{\rm LR}. \label{eq:LSVQC_proposition_1}
    \end{align}
    Here, the error term $\epsilon_{\rm LR}$ scales as $\mathcal{O}(e^{-(l_0+R/2)/\xi})$ with $R=4(d_V+\max_k{d_{W_k}})$. 
    Thus, the approximation error $\epsilon_{\rm LR}$ can be exponentially small by increasing the parameter $l_0$ (i.e., increasing the restriction size $L'$). 
\end{prop}
Proposition 1 is justified by the LR bound (see Appendix~\ref{append:LSVQC_detail} for the proof of Proposition 1). 
Based on Proposition 1, we can evaluate the local cost functions $C_{\rm LLET}^{(j),\mathcal{S}}(U^{(L)}_{\tau},V^{(L)})$ in Eq.~\eqref{eq:cost_LLET_j} by using the restricted time-evolution operator $U^{(L',j)}_{\tau}$ instead of the original time-evolution operator $U_{\tau}^{(L)}=e^{-i\tau H^{(L)}}$. 

Next, we proceed to the step (ii). 
We define the locally restricted version of the ansatz $V^{(L)}(\bm{\theta})$ as  
\begin{align}
    &V^{(\tilde{L},j)}(\bm{\theta}) \nonumber\\
    &= \prod_{l=1}^{d_V}\left[  \left(\sideset{}{^{'}}\prod_{i\in\Lambda^{(\tilde{L},j)}}v_{2i-1,2i}^{(l)}(\theta_{2i-1}^{(l)}) \right) \left(\sideset{}{^{'}}\prod_{i\in\Lambda^{(\tilde{L},j)}} v_{2i,2i+1}^{(l)}(\theta_{2i}^{(l)}) \right)\right], \label{eq:V_L_localized}
\end{align}
where $\sideset{}{^{'}}\prod_{i\in\Lambda^{(\tilde{L},j})}$ denotes the product over $i$ such that the support of $v^{(l)}_{2i,2i+1}$ or $v^{(l)}_{2i-1,2i}$ is included in the domain $\Lambda^{(\tilde{L},j)}$ (see Fig.~\ref{fig:LSVQC_circuit}). 
Similarly, we define the local restriction of the state preparation circuits $\{W^{(L)}_k\}$ as 
\begin{align}
    W_k^{(\tilde{L},j)} 
    &= \prod_{l=1}^{d_{W_k}}\left[  \left(\sideset{}{^{'}}\prod_{i\in\Lambda^{(\tilde{L},j)}}w_{2i-1,2i}^{(l)} \right) \left(\sideset{}{^{'}}\prod_{i\in\Lambda^{(\tilde{L},j)}}w_{2i,2i+1}^{(l)} \right) \right]. \label{eq:W_L_localized}
\end{align}
Then, $V^{(L)}(\bm{\theta})$ and $\{W^{(L)}_k\}$ appeared in Eq.~\eqref{eq:LSVQC_proposition_1} can be replaced by their locally restricted counterparts using the following proposition. 
\begin{prop}[Local restriction of the ansatz and state preparation circuits]\label{prop2}
    Let us set the compilation size $\tilde{L}(>L')$ such that 
    \begin{align}
        \tilde{L} \geq \mathrm{max}\left( \frac{L'}{2}+2(d_V+\max_k d_{W_k})+1, L'+4\max_k d_{W_k}\right),
        \label{eq:L_compilation}
    \end{align}
    where $L'$ is given by Eq.~\eqref{eq:L_restrict}. 
    Then, the following equality holds for the right-hand side of Eq.~\eqref{eq:LSVQC_proposition_1}: 
    \begin{align}
        C_{\rm LLET}^{(j),B_\mathcal{S}}(U_{\tau}^{(L',j)},V^{(L)}) = C_{\rm LLET}^{(j),B_{\tilde{\mathcal{S}}_j}}(\tilde{U}_{\tau}^{(L',j)},V^{(\tilde{L},j)}). \label{eq:LSVQC_proposition_2}
    \end{align}
    Here, $\tilde{U}_{\tau}^{(L',j)}$ is the extension of $U_{\tau}^{(L',j)}$ to the $\tilde{L}$-size domain $\Lambda^{(\tilde{L},j)}$, 
    \begin{align}
        \tilde{U}_{\tau}^{(L',j)} &= e^{-i\tau H^{(L',j)}}\otimes \mathbbm{1}_{\Lambda^{(\tilde{L},j)}\setminus\Lambda^{(L',j)}}. 
    \end{align}
    $B_{\tilde{\mathcal{S}}_j}$ denotes a basis set of a subspace $\tilde{\mathcal{S}}_j$ spanned by the states generated by the locally restricted state preparation circuits $\{W_k^{(\tilde{L},j)}\}$, i.e.,  
    \begin{align}
        B_{\tilde{\mathcal{S}}_j} &= \left\{W_k^{(\tilde{L},j)}\ket{0}^{\otimes \tilde{L}} \right\}_{k=0}^{N-1}, \\
        \tilde{\mathcal{S}}_j &= \mathrm{span}( B_{\tilde{\mathcal{S}}_j} ). 
    \end{align}
\end{prop}
Proposition 2 is derived based on the causal cone arising from the locality of the ansatz and state preparation circuits. 
The proof of Proposition 2 is given in Appendix~\ref{append:LSVQC_detail}. 

Combining Proposition~\ref{prop1} with Proposition~\ref{prop2}, we arrive at the following local compilation theorem:
\begin{thm}[Local compilation theorem]\label{thm1}
    We define the local subsystem cost function as
    \begin{align}
        C^{(\tilde{L}),B_\mathcal{S}}(\bm{\theta}) &= \frac{1}{L}\sum_{j=1}^{L}C_{\rm LLET}^{(j),B_{\tilde{\mathcal{S}}_j}}(U_{\tau}^{(\tilde{L},j)},V^{(\tilde{L},j)}(\bm{\theta})). \label{eq:subsystem_cost_fn}
    \end{align}
    When we achieve $C^{(\tilde{L}),B_\mathcal{S}}(\bm{\theta}^*)\leq\epsilon_{\rm opt}$ for some optimal parameter set $\bm{\theta}^*$, the cost functions for the original $L$-size system are bounded as follows:
    \begin{align}
        C_{\rm LLET}^{B_\mathcal{S}}(U_{\tau}^{(L)},V^{(L)}(\bm{\theta}^*)) &\leq \epsilon_{\rm opt} + \epsilon_{\rm LR}, \label{eq:LSVQC_ineq_LLET} \\
        C_{\rm LET}^{B_\mathcal{S}}(U_{\tau}^{(L)},V^{(L)}(\bm{\theta}^*)) &\leq L\left(\epsilon_{\rm opt} + \epsilon_{\rm LR}\right). \label{eq:LSVQC_ineq_LET}
    \end{align}
\end{thm}
The proof of Theorem 1 is provided in Appendix~\ref{append:LSVQC_detail}. 
Based on Theorem~\ref{thm1}, we can indirectly minimize the cost functions $C_{\rm LLET}^{B_\mathcal{S}}$ or $C_{\rm LET}^{B_\mathcal{S}}$ for the target $L$-size system by minimizing the local subsystem cost function $C^{(\tilde{L}),B_\mathcal{S}}(\bm{\theta})$. 
This fact justifies the local optimization protocol of the LSVQC. 

Theorem~\ref{thm1} is a general form of the local compilation theorem. 
On the other hand, we can modify the local compilation theorem into a form suitable for practical use. 
In Appendix~\ref{append:local_comp_thms}, we provide modified versions of the local compilation theorem for translation invariant systems and long-time dynamics simulation. 

The appropriate choice of the compilation size $\tilde{L}$ depends on the purpose of the whole computation. 
If we aim to simulate the dynamics of global observables (i.e., observables acting on a large number of sites or whole lattice) with the accuracy $\epsilon$, we need to achieve $C_{\rm LET}^{B_\mathcal{S}}=\mathcal{O}(\epsilon)$ at least. 
This indicates that the compilation size should be chosen to realize $L\epsilon_{\rm LR}=\mathcal{O}(Le^{-(l_0+R/2)/\xi})=\mathcal{O}(\epsilon)$. 
On the other hand, when we simulate the dynamics of some local observables (e.g., local particle density) with the accuracy $\epsilon$, we only need to realize $C_{\rm LLET}^{B_\mathcal{S}}=\mathcal{O}(\epsilon)$ since $C_{\rm LLET}^{B_\mathcal{S}}$ quantifies the local accuracy of the time evolution. 
Then, the compilation size is determined to realize $\epsilon_{\rm LR}=\mathcal{O}(e^{-(l_0+R/2)/\xi})=\mathcal{O}(\epsilon)$. 
In summary, to achieve the accuracy $\epsilon$ for the dynamics simulation of global or local observables, the compilation size $\tilde{L}$ should be chosen as 
\begin{align}
    \tilde{L} &= \mathcal{O}(\xi\log(L^\alpha/\epsilon)) + r_H + v\tau + 2\max_{k}d_{W_k} \nonumber\\
    & + \max{\left( 2d_V+1, r_H+v\tau+2\max_{k}d_{W_k} \right)}, \label{eq:L_compilation_explicit}
\end{align}
where $\alpha=1$ ($\alpha=0$) for global (local) observables. 
Note that the LSVQC algorithm might work well even when the compilation size is set to be smaller than the right-hand side of Eq.~\eqref{eq:L_compilation_explicit} owing to the looseness of the expression of the LR bound given by Eq.~\eqref{eq:LR_bound}. 
Indeed, we show such an example in Sec.~\ref{sec:DF}. 

\section{Numerical demonstration for Heisenberg spin chain}\label{sec:numerical_test}
Here, we demonstrate the basic properties of the LSVQC algorithm by performing a numerical simulation of a simple toy model. 
Specifically, we consider the 1D antiferromagnetic Heisenberg model
\begin{align}
    H^{(L)} &= \sum_{j=1}^{L}\left( X_{j}X_{j+1} + Y_{j}Y_{j+1} + Z_{j}Z_{j+1} \right), \label{eq:H_XXX}
\end{align}
where we adopt the periodic boundary condition.  
We approximate the target time-evolution unitary $e^{-i\tau H^{(L)}}$ by applying the first-order Trotterization as
\begin{align}
    U_{\tau}^{(L)} &= \left(e^{-i(\tau/r)H_{\rm even}^{(L)}}e^{-i(\tau/r)H_{\rm odd}^{(L)}}\right)^r, \label{eq:Trotter_AFM}
\end{align}
where the depth parameter $r$ is set to be a large value $r=100$. 
Note that we only need to implement a local restriction of $U_{\tau}^{(L)}$ to a compilation size $\tilde{L}(<L)$ in the optimization protocol of the LSVQC. 
$H_{\rm odd}^{(L)}$ ($H_{\rm even}^{(L)}$) represents the sum of terms connecting odd-even (even-odd) sites in $H^{(L)}$. 
To make clear correspondence with the discussion in Sec.~\ref{sec:LSVQC_derivation}, we adopt a brick-wall structure ansatz $V^{(L)}(\bm{\theta})$ given by Eq.~\eqref{eq:V_L}. 
We parametrize the two-qubit gates $v_{j,j+1}(\bm{\theta}_l)$ by the following symmetry-preserving form~\cite{LVQC_PRXQuantum.3.040302}
\begin{align}
   & v_{j,j+1}(\eta,\zeta,\chi,\gamma,\phi) \nonumber\\
   &=
    \begin{pmatrix}
        1 & 0 & 0 & 0 \\
        0 & e^{-i(\gamma+\zeta)}\cos{\eta} & -ie^{-i(\gamma-\chi)}\sin{\eta} & 0 \\
        0 & -ie^{-i(\gamma+\chi)}\sin{\eta} & e^{-i(\gamma-\zeta)}\cos{\eta} & 0 \\
        0 & 0 & 0 & e^{-i(2\gamma+\phi)}
    \end{pmatrix},
    \label{eq:v2_AFM}
\end{align}
in the basis of $\{\ket{00}, \ket{01}, \ket{10}, \ket{11}\}$. 
Here, $(\eta, \zeta, \chi, \gamma, \phi)$ is a set of variational parameters. 
$v_{j,j+1}(\bm{\theta}_l)$ in the form of Eq.~\eqref{eq:v2_AFM} preserves the total $Z$ spin of the system, which is a symmetry of Eq.~\eqref{eq:H_XXX}~\cite{LVQC_PRXQuantum.3.040302}. 
We also impose the translational symmetry to the ansatz as $v_{j,j+1}=v_{j+2,j+3}$ for any $j\in\{1,2,\cdots,L\}$ under the periodic boundary condition.  
As a simple example, we aim to simulate the dynamics of the following N\'{e}el state
\begin{align}
    \ket{\psi_0}=\prod_{j=1}^{L/2}X_{2j}\ket{0}^{\otimes L}=\ket{1010\cdots10}. \label{eq:state_AFM}
\end{align}
For simplicity, we assume that the target physical quantity is not a correlation function like Eq.~\eqref{eq:dynamics_AB}. 
In this case, we can utilize the Krylov subspace $\mathcal{S}$ given by Eq.~\eqref{eq:subspace_A} as a target subspace of the LSVQC.

Under the above setup, we numerically demonstrate the performance of the LSVQC algorithm. 
Since the Hamiltonian, the ansatz, and the state preparation circuits are all translationally invariant, we can apply the LSVQC based on Theorem~\ref{thm2} in Appendix~\ref{append:local_comp_thms}.
We first estimate the compilation size $\tilde{L}$ based on Eq.~\eqref{eq:L_compilation_explicit}. 
In the following, we fix the depth of the ansatz as $d_V=2$ and execute the LSVQC to compile the time-evolution operator at $\tau=0.1$. 
The range of the interaction of the Hamiltonian~\eqref{eq:H_XXX} is $r_H=1$. 
The parameters $v$ and $\xi$ are in general $\mathcal{O}(1)$ and independent of the system size. 
Thus, $v\tau$ in Eq.~\eqref{eq:L_compilation_explicit} can be negligible when the time $\tau$ is chosen as a small value $\tau=0.1$. 
Since $\ket{\psi_0}$ is just a computational basis state and the time-evolution operator in the Krylov subspace~\eqref{eq:subspace_A} is approximated by the single step Trotterization circuit, the depth of the state preparation circuits is $d_{W_k}=1$. 
We set $\alpha=0$ and use the LLET cost considering the simulation of local observables. 
Using these parameters, we can roughly estimate the proper compilation size as $\tilde{L}\gtrsim8$.
Therefore, we fix the compilation size as $\tilde{L}=8$ and simulate the dynamics at size $L\geq 8$.
The cost functions are minimized by using the Broyden-Fletcher-Goldfarb-Shanno (BFGS) method implemented in SciPy~\cite{SciPy}. 
We start the optimization from the initial value $\bm{\theta}_0$ which is chosen so that the initial ansatz $V^{(L)}(\bm{\theta}_0)$ becomes equivalent to the same-depth Trotterization circuit~\eqref{eq:Trotter_AFM} (i.e., $r=d_V=2$). 
The quantum circuit simulation is performed using Qulacs~\cite{qulacs}. 

\begin{figure*}[htbp]
    \includegraphics[width=17cm]{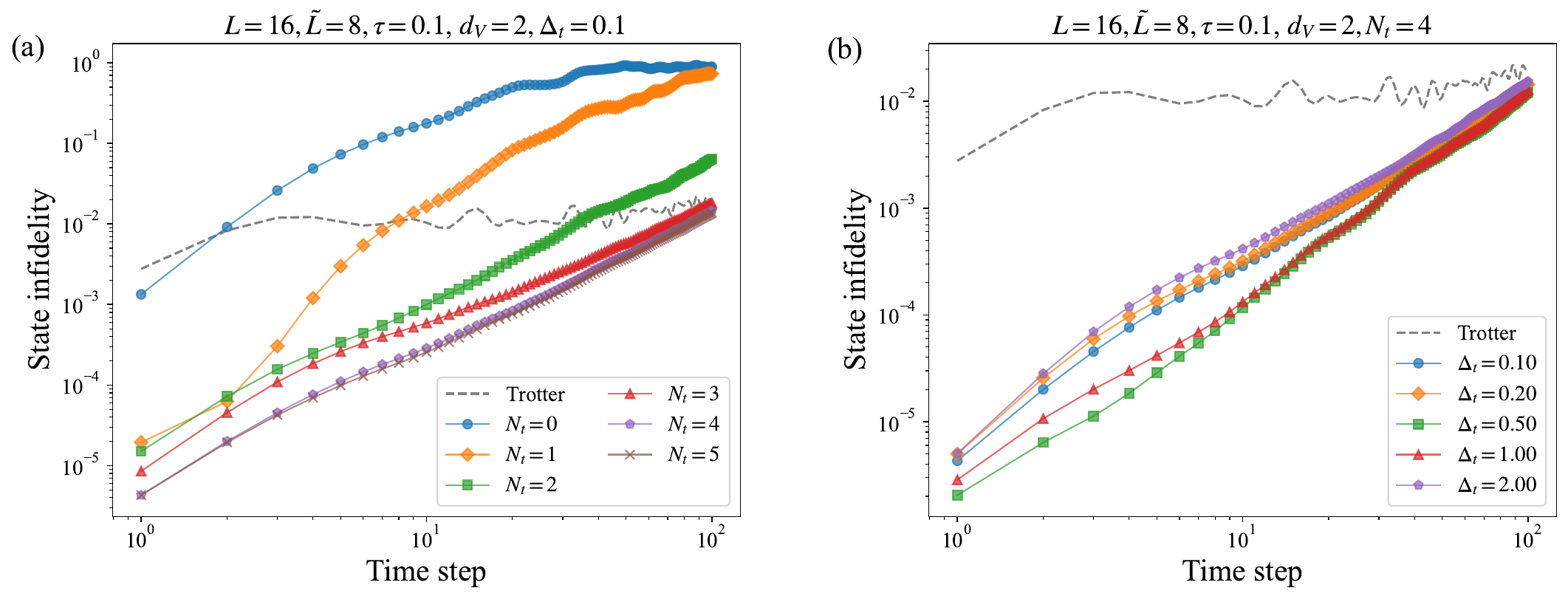}
    \caption{State infidelity for the N\'{e}el state $\ket{\psi_0}$ with $L=16$ as a function of the time step $n$. 
    The colored markers represent the results obtained by the LSVQC with (a) $\Delta_t=0.1$, $N_t=0,1,\cdots,5$ and (b) $N_t=4$, $\Delta_t=0.1, 0.2, 0.5, 1.0, 2.0$. The grey dash-dotted line indicates the result obtained by the Trotterization.
    The LSVQC is performed for $\tilde{L}=8$, $\tau=0.1$, and $d_V=2$. } 
    \label{fig:LSVQC_Heisenberg_fildeity}
\end{figure*}
First, we investigate how the subspace design influences the performance of the LSVQC. 
We quantify the accuracy of optimized ansatz by evaluating the state infidelity as 
\begin{align}
    \bar{F}^{(L)}(n\tau)=1-|\bra{\psi_0}(V^{(L)}(\bm{\theta}^*))^n(U^{(L)}_{\tau})^{n}\ket{\psi_0}|, \label{eq:state_fidelity}
\end{align} 
where $n\in \mathbb{N}$. 
Figure~\ref{fig:LSVQC_Heisenberg_fildeity} shows the time dependence of the state infidelity for $L=16$ calculated by using the optimized ansatz $V^{(L)}(\bm{\theta}^*)$ obtained by the LSVQC (colored markers).
The results of the Trotterization with $r=d_V=2$ (black dashed line) are also shown for comparison.
The accuracy of the LSVQC is strongly affected by the parameters of the subspace~\eqref{eq:subspace_A}.  
Enlarging the dimension of the subspace by increasing $N_t$ from 0 to 5 under $\Delta_t=0.1$, the state infidelity of the LSVQC gradually decreases and saturates at $N_t\simeq4$ (Fig.~\ref{fig:LSVQC_Heisenberg_fildeity} (a)). 
This suggests the time evolution of $\ket{\psi_0}$ can be accurately approximated by the LSVQC using a moderate number of basis states sufficient to describe the ideal Krylov subspace. 
Note that $\Delta_t$ is set to be a small value, $\Delta_t=0.1$, and hence the effect of the Trotter error stated in Sec.~\ref{sec:subspace} is expected to be negligible. 
On the other hand, Fig.~\ref{fig:LSVQC_Heisenberg_fildeity} (b) shows that the accuracy of the LSVQC is maximized at some moderate value of $\Delta_t$, i.e., $\Delta_t=0.5$. 
This behavior can be understood as a consequence of the competition between the linear independence of the input states and the subspace leakage explained in Sec.~\ref{sec:subspace}. 
Increasing $\Delta_t$ increases the linear independence of the input states, which is required to correctly learn the time evolution in the subspace as stated in the No-Free-Lunch theorem~\cite{poland2020no}. 
However, increasing $\Delta_t$ also induces leakage from the ideal subspace due to the large Trotter error. 
As a consequence of these two effects, the best performance of the LSVQC is obtained at some intermediate value of $\Delta_t$ as shown in Fig.~\ref{fig:LSVQC_Heisenberg_fildeity} (b). 
In the following numerical calculations, we choose the subspace parameters as $(N_t,\Delta_t)=(1, 0.5)$ which is an optimal parameter set found by computing the average state infidelity $\frac{1}{n}\sum_{n=1}^{100}\bar{F}(n\tau)$ for $N_t=0,1,\cdots,10$ and $\Delta_t=0.1,0.2,0.5,1.0,2.0$. 

Next, we demonstrate the circuit depth compression realized by the LSVQC.  
For convenience of the following discussion, we here define the depth compression rate as 
\begin{align}
    R^{\rm LSVQC}_{\epsilon} &= \frac{D^{\rm trot}_{\epsilon}}{D^{\rm LSVQC}_{\epsilon}}, \label{eq:d_compress}
\end{align}
where $D^{\rm LSVQC}_{\epsilon}$ and $D^{\rm trot}_{\epsilon}$ are the required circuit depth of the LSVQC and Trotterization to achieve some accuracy $\epsilon$, respectively. 
The ratio $R^{\rm LSVQC}_{\epsilon}$ measures how much the circuit depth can be reduced by employing the LSVQC compared to the Trotterization. 
When $R^{\rm LSVQC}_{\epsilon}>1$, the death compression compared to the Trotterization is realized by using the LSVQC. 

\begin{figure*}[htbp]
    \includegraphics[width=17.5cm]{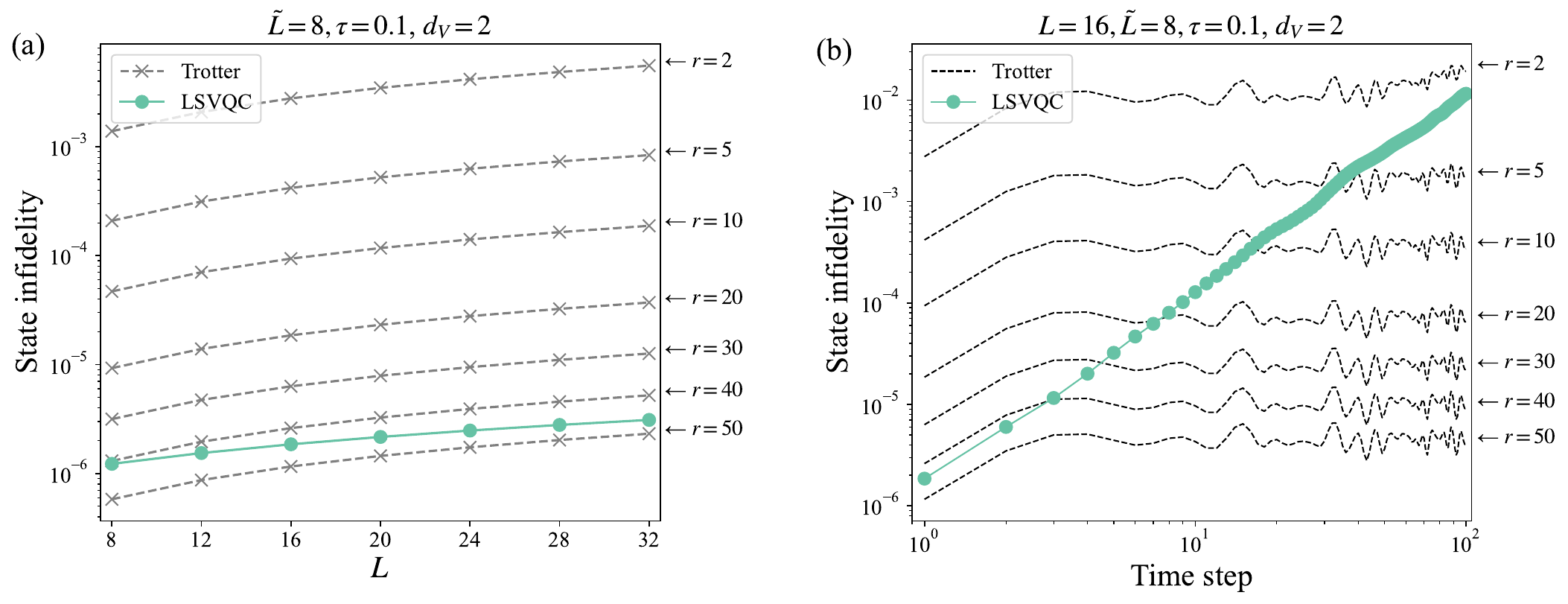}
    \caption{Numerical demonstration of the depth compression in terms of the state infidelity for the N\'{e}el state $\ket{\psi_0}$. 
    (a) The size dependence of the state infidelity for $L\geq8$. 
    (b) The state infidelity as a function of the time step $n$ for $L=16$. 
    The green circles show the state infidelity obtained by the LSVQC. The grey dashed lines or grey cross marks indicate the state infidelity of the Trotterization for various values of the depth parameter $r$. The LSVQC is performed for $\tilde{L}=8$, $\tau=0.1$, and $d_V=2$. } 
    \label{fig:LSVQC_Heisenberg_depth}
\end{figure*}
In Fig.~\ref{fig:LSVQC_Heisenberg_depth}(a), we compare the size dependence of the state infidelity of LSVQC (green circles) and Trotterization (grey cross marks) as a function of the system size $L\geq\tilde{L}=8$. 
To visualize the depth compression rate of the LSVQC against the Trotterization, we show the state infidelity of the Trotterization for various depth parameters $r$ but the depth of the LSVQC is fixed to $d_V=2$.
We see that the state infidelity of the LSVQC is comparable to that of the Trotterization with $r=40$-$50$ for all $L$. 
Hence, the depth compression rate can be estimated as $R^{\rm LSVQC}_{\epsilon}\simeq 40/2 = 20$. 
In other words, the LSVQC realizes the simulation of this model using a circuit shallower by about 95\% than the Trotterization while maintaining the accuracy. 
We also note that this result suggests the robustness of the LSVQC for the extension of the system size under a fixed compilation size. 
The size-insensitive behavior is consistent with the local compilation theorem in Sec.~\ref{sec:LSVQC_derivation}, in which the bound of the compilation size does not directly depend on the system size $L$ (see Eq.~\eqref{eq:L_compilation}). 
In Fig.~\ref{fig:LSVQC_Heisenberg_depth}(b), we show that the state infidelity of the LSVQC gradually increases by increasing the simulation time step $n$. 
This is consistent with the fact that the approximation error of the LSVQC increases at a long-time scale as shown in Eq.~\eqref{eq:L_compilation_explicit_longT}. 
This implies that we need to logarithmically enlarge the compilation size $\tilde{L}$ based on Eq.~\eqref{eq:L_compilation_explicit_longT} to achieve significant depth compression even at a long-time regime. 

\section{Application to \textit{ab initio} materials simulation}\label{sec:DF}
In this section, we apply the LSVQC to \textit{ab initio} materials simulation. 
Specifically, we aim to efficiently simulate a low-energy effective lattice model for strongly correlated electron materials (e.g., high-temperature superconductors) on quantum computers using the LSVQC. 
We use a technique of the classical materials simulation called the {\it ab initio downfolding}~\cite{downfolding_imada2010electronic} to construct the effective low-energy model of strongly correlated electron materials. 
The \textit{ab initio} downfolding is executed in the following three steps. 
The first step is the band structure calculation based on the density functional theory (DFT), which is executed using Quantum ESPRESSO package~\cite{ESPRESSO1_giannozzi2009quantum,ESPRESSO2_giannozzi2017advanced,ESPRESSO3_giannozzi2020quantum}. 
In the DFT calculation, we adopt the generalized
gradient approximation by Perdew-Burke-Ernzerhof to the exchange-correlation functional~\cite{PBEapprox_PhysRevLett.77.3865} and norm-conserving
pseudopotentials~\cite{norm_PhysRevLett.43.1494,norm_PhysRevB.88.085117}. 
Second, we construct the orbital of the effective lattice model using the maximally localized Wannier function~\cite{Wannier_PhysRevB.56.12847}. 
The third step is to calculate the effective Coulomb and effective exchange interactions by adopting the constrained random phase approximation~\cite{cRPA_PhysRevB.70.195104} using the Wannier orbital obtained in the second step.  
The second and third steps are numerically executed using RESPACK package~\cite{RESPACK1_PhysRevB.93.085124,RESPACK2_nakamura2009ab,RESPACK3_nakamura2008ab,REPACK4_PhysRevB.79.195110,RESPACK5_fujiwara2003generalization,RESPACK6_nakamura2021respack}. 
The effective lattice model obtained by the \textit{ab initio} downfolding, which we call the ``\textit{ab initio} downfolding model" in the following, enables us to describe the low-energy physics of strongly correlated electron materials that can not be captured by the DFT.  
More details are summarized in Appendix~\ref{append:DF}. 

We consider \ce{Sr2CuO3} as a benchmark material. 
\ce{Sr2CuO3} is a quasi-one-dimensional cuprate compound with strong electronic correlation mainly from Cu $d$ orbitals~\cite{SCO_AFM_PhysRevB.51.5994,SCO_AFM_PhysRevLett.76.3212,SCO_separation_PhysRevLett.81.657,SCO_separation_PhysRevB.59.7358,SCO_SC_liu2014new}. 
The electronic correlation leads to various intriguing phenomena such as antiferromagnetism~\cite{SCO_AFM_PhysRevB.51.5994,SCO_AFM_PhysRevLett.76.3212}, and spin-charge separation~\cite{SCO_separation_PhysRevLett.81.657,SCO_separation_PhysRevB.59.7358}. 
Superconductivity is also reported in tetragonal hole-doped Sr$_2$CuO$_{3+\delta}$~\cite{SCO_SC_liu2014new}. 
Owing to the quasi-one-dimensional nature of the electronic structure, the \textit{ab into} downfolding model of \ce{Sr2CuO3} is described as a 1D extended Fermi-Hubbard Hamiltonian:
\begin{align}
    H^{(L)} &= -t_1 \sum_{i\in\Lambda}\sum_{\sigma=\uparrow,\downarrow}(c_{i,\sigma}^{\dag}c_{i+1,\sigma}+c_{i+1,\sigma}^{\dag}c_{i,\sigma}) \nonumber\\
    & - t_2 \sum_{i\in\Lambda}\sum_{\sigma}(c_{i,\sigma}^{\dag}c_{i+2,\sigma}+c_{i+2,\sigma}^{\dag}c_{i,\sigma}) \nonumber\\
    &  + U \sum_{i\in\Lambda} n_{i\uparrow}n_{i\downarrow} -\mu \sum_{i\in\Lambda}\sum_{\sigma} n_{i\sigma} , 
    \label{eq:Sr2CuO3_hamil}
\end{align}
where $c_{i,\sigma}$ and $c_{i,\sigma}^{\dag}$ are the annihilation and creation operator of an electron with spin $\sigma(=\uparrow,\downarrow)$ at site $i\in\Lambda=\mathbb{Z}_{/L}=\{0,1,\cdots,L-1\}$, and $n_{i,\sigma}=c_{i,\sigma}^{\dag}c_{i,\sigma}$ is the particle number operator. 
The lattice $\Lambda$ is defined under the periodic boundary condition with $L=|\Lambda|$ being the total number of sites. 
We assume that $L$ is an even integer in the following. 
The parameters $t_1$, $t_2$, $U$, and $\mu$ are the nearest-neighbor hopping integral, next-nearest-neighbor hopping integral, on-site Coulomb interaction, and chemical potential, respectively. 
Note that Eq.~\eqref{eq:Sr2CuO3_hamil} contains the next-nearest-neighbor hopping term that is not included in the conventional Fermi-Hubbard model. 
The values of the parameters calculated by using RESPACK package~\cite{RESPACK1_PhysRevB.93.085124,RESPACK2_nakamura2009ab,RESPACK3_nakamura2008ab,REPACK4_PhysRevB.79.195110,RESPACK5_fujiwara2003generalization,RESPACK6_nakamura2021respack} are summarized in Table~\ref{tab:parameter_Sr2CuO3}. 
Since there is only a single band around the Fermi energy in the DFT band structure of \ce{Sr2CuO3} (see Fig.~\ref{fig:band_Sr2CuO3}(a)), the \textit{ab initio} downloading model~\eqref{eq:Sr2CuO3_hamil} is described by using only a single Wannier orbital shown in Fig.~\ref{fig:band_Sr2CuO3}(b). 
\begin{table}[tbp]
    \centering
    \caption{Parameters of the \textit{ab initio} downfolding model for \ce{Sr2CuO3} obtained by the RESPACK package~\cite{RESPACK1_PhysRevB.93.085124,RESPACK2_nakamura2009ab,RESPACK3_nakamura2008ab,REPACK4_PhysRevB.79.195110,RESPACK5_fujiwara2003generalization,RESPACK6_nakamura2021respack}. The value of the chemical potential $\mu$ is determined so that the origin of the Wannier-interpolated band structure coincides with that of the DFT band structure. }
    \begin{ruledtabular}
    \begin{tabular}{c|cccc}
        Parameters & $t_1$ & $t_2$ & $U$ & $\mu$ \\
        \hline 
        Value (eV) &  0.532 & 0.0403 & 1.054 & 0.159
    \end{tabular}
    \end{ruledtabular}
    \label{tab:parameter_Sr2CuO3}
\end{table}
\begin{figure}[htbp]
    \includegraphics[width=8.5cm]{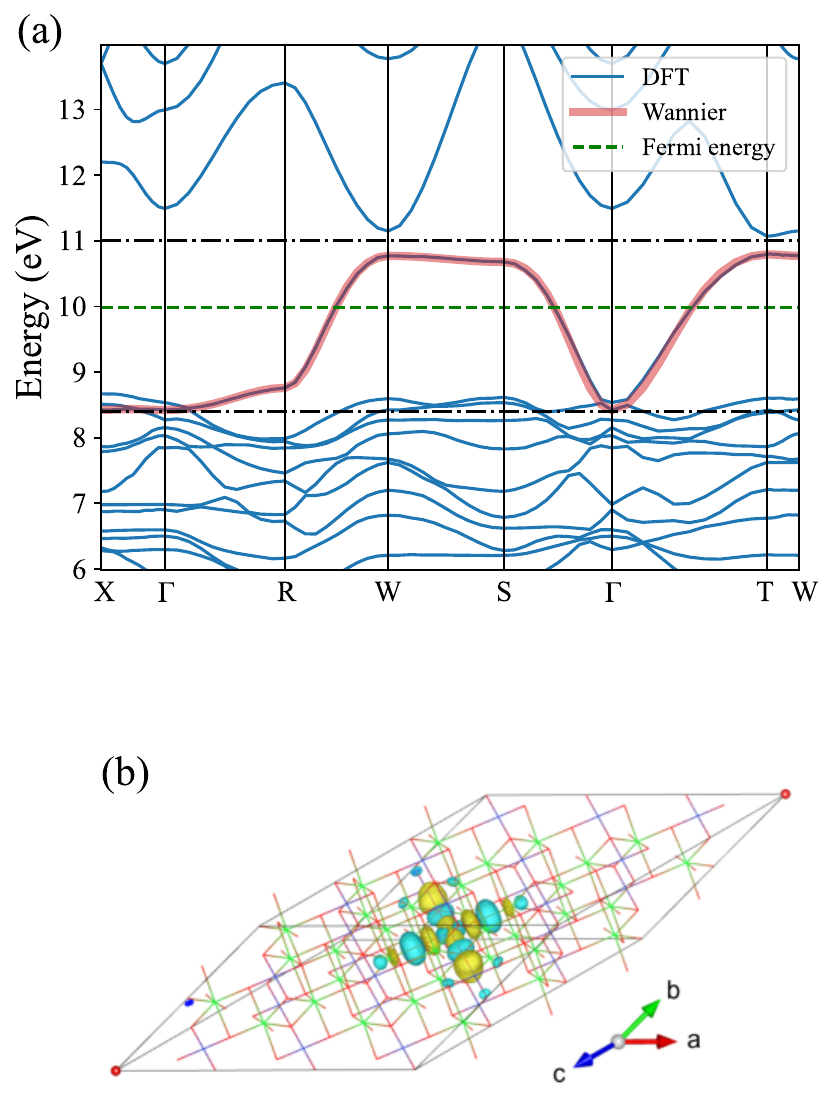}
    \caption{(a) Band structure of \ce{Sr2CuO3}. The blue thin curves represent the DFT band structure calculated by using the Quantum ESPRESSO package, while the red bold curve is the Wannier-interpolated band structure. The green dashed horizontal line indicates the Fermi energy. The Wannier orbitals are constructed using the band structure data in the energy window region, which is described by a region between the black dash-dotted horizontal lines. 
    (b) Calculated Wannier function and the crystal structure of \ce{Sr2CuO3} (drawn by VESTA~\cite{VESTA_momma2011vesta}). The blue and yellow isosurfaces indicate positive- and negative-value regions of the Wannier function. The green, blue, and red vertices express Sr, Cu, and O atoms, respectively.  } 
    \label{fig:band_Sr2CuO3}
\end{figure}

In the quantum circuit simulation, we transform the fermionic Hamiltonian~\eqref{eq:Sr2CuO3_hamil} into the qubit representation using the following Jordan-Wigner transformation~\cite{Jordan1928}: 
\begin{align}
    c_{i,\sigma} &\mapsto \frac{1}{2}(X_{i_\sigma}+iY_{i_\sigma})\prod_{k<i_\sigma}Z_{k_\sigma}, \label{eq:JW_annihilation} \\
    c_{i,\sigma}^{\dag} &\mapsto \frac{1}{2}(X_{i_\sigma}-iY_{i_\sigma})\prod_{k<i_\sigma}Z_{k_\sigma}, \label{eq:JW_creation}
\end{align}
where $(X_{i_\sigma}, Y_{i_\sigma}, Z_{i_\sigma})$ are the Pauli operators for $i_\sigma$-th qubit.  
We allocate the qubits so that their indices satisfy $i_{\uparrow}=i(=0,1,2,\cdots,L-1)$ and $i_{\downarrow}=2L-1-i(=2L-1,2L-2,\cdots,L)$, and the total number of the qubits is $N_q=2L$. 
Then, the qubit representation of Eq.~\eqref{eq:Sr2CuO3_hamil} is obtained as 
\begin{align}
    H^{(L)} &= -\frac{t_1}{2} \sum_{i=1}^{L}\sum_{\sigma=\uparrow,\downarrow}(X_{i+1_\sigma}X_{i_\sigma}+Y_{i+1_\sigma}Y_{i_\sigma}) \nonumber\\
    & - \frac{t_2}{2} \sum_{i=1}^{L}\sum_{\sigma}(X_{i+2_\sigma}Z_{i+1_\sigma}X_{i_\sigma}+Y_{i+2_\sigma}Z_{i+1_\sigma}Y_{i_\sigma}) \nonumber\\
    & + \frac{U}{4} \sum_{i=1}^{L} Z_{i_\uparrow}Z_{i_\downarrow} +\frac{1}{2}\left(\mu-\frac{U}{2}\right)\sum_{i=1}^{L}\sum_{\sigma} Z_{i_\sigma}  . 
    \label{eq:Sr2CuO3_hamil_JW}
\end{align}
We perform the LSVQC using the variational Hamiltonian ansatz (VHA)~\cite{Wecker2015,reiner2019finding}. 
To construct the VHA, we divide the Hamiltonian~\eqref{eq:Sr2CuO3_hamil_JW} as
\begin{align}
    H^{(L)} = \sum_{m=1}^{M} c_m H_m^{(L)}, 
\end{align}
where each $H_m$ is a sum of mutually-commuting terms. 
The Hamiltonian~\eqref{eq:Sr2CuO3_hamil_JW} can be divided into the following six terms (i.e., $M=6$):
\begin{align}
    H_1^{(L)} &= \sum_{i\in\Lambda^{\prime},\sigma}(X_{i+1_\sigma}X_{i_\sigma}+Y_{i+1_\sigma}Y_{i_\sigma}), \\
    H_2^{(L)} &= \sum_{i\in\Lambda^{\prime},\sigma}(X_{i+2_\sigma}X_{i+1_\sigma}+Y_{i+2_\sigma}Y_{i+1_\sigma}), \\
    H_3^{(L)} &= \sum_{i\in\Lambda^{\prime\prime},\sigma}(X_{i+2_\sigma}X_{i_\sigma}+Y_{i+2_\sigma}Y_{i_\sigma})Z_{i+1_\sigma}, \\
    H_4^{(L)} &= \sum_{i\in\Lambda^{\prime\prime},\sigma}(X_{i+4_\sigma}X_{i+2_\sigma}+Y_{i+4_\sigma}Y_{i+2_\sigma})Z_{i+3_\sigma}, \\
    H_5^{(L)} &= \sum_{i\in\Lambda}Z_{i_\uparrow}Z_{i_\downarrow}, \\
    H_6^{(L)} &= \sum_{i\in\Lambda, \sigma} Z_{i_\sigma}, 
\end{align}
where $\Lambda^{\prime}=(2\mathbb{Z})_{/L}=\{0,2,4,\cdots\}$ and $\Lambda^{\prime\prime}=(4\mathbb{Z})_{/L}\cup((4\mathbb{Z})_{/L}+1)=\{0,1,4,5,8,9,\cdots\}$. 
Each term is schematically illustrated in Fig~\ref{fig:H_term_Sr2CuO3}. 
The corresponding coefficients are explicitly given by $c_m=(-t_1/2,-t_1/2,-t_2/2,-t_2/2,U/4,(\mu-U/2)/2)$. 
\begin{figure}[tbp]
    \includegraphics[width=9cm]{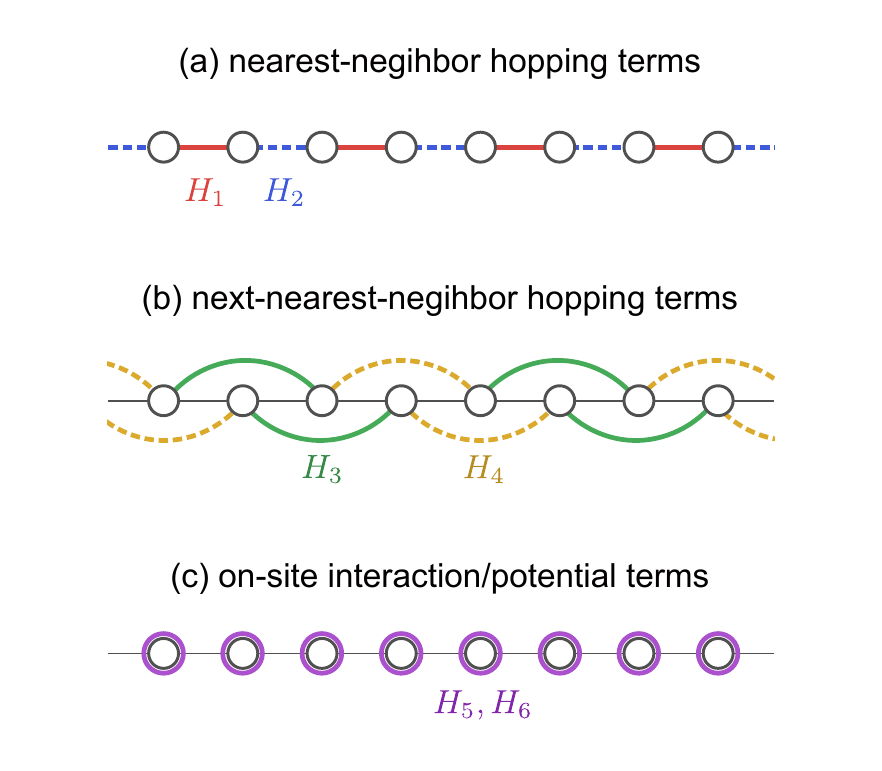}
    \caption{Schematic illustration of the decomposition of the effective Hamiltonian for \ce{Sr2CuO3}. 
    (a) Nearest-neighbor hopping terms $H_1$ (red solid lines) and $H_2$ (blue dashed lines). 
    (b) Next-nearest-neighbor hopping terms $H_3$ (green solid lines) and $H_4$ (yellow dashed lines). 
    (c) On-site interaction terms $H_5$ and on-site one-body potential terms $H_6$. } 
    \label{fig:H_term_Sr2CuO3}
\end{figure}
Under the above decomposition, the VHA is defined as 
\begin{align}
    V^{(L)}(\bm{\theta}) &= \prod_{l=1}^{N_L}\left[\prod_{m=1}^{M}e^{i\theta_{l,m}H_m^{(L)}} \right] ,
    \label{eq:VHA_Sr2CuO3}
\end{align}
where $\bm{\theta}=\{\theta_{l,m}\}$ is a set of $6N_L$ variational parameters. 
$N_L$ denotes the number of layers corresponding to the circuit depth. 
In the same manner, we also define the first-order Trotterization as
\begin{align}
    U^{(L)}_{\rm trot}(t) &= \left(\prod_{m=1}^{M}e^{i(t/r)c_m H_m^{(L)}} \right)^r ,
    \label{eq:Trot_Sr2CuO3}
\end{align}
where $t$ denotes time. 
The target unitary of the LSVQC $e^{-i\tau H^{(L)}}$ is approximated by the above Trotterization with a large depth $r=100$.  
Note that it is not necessary to implement such a deep Trotter circuit on a large-scale $L$-qubit quantum device. Instead, we only need to implement a local restriction of $U^{(L)}_{\rm trot}(t)$ to a small compilation size $\tilde{L}(<L)$ to perform the LSVQC algorithm. 

Our numerical simulations are performed as follows. 
The number of sites is set to $L=8$ ($N_q=16$), which is the largest system size we could perform the exact classical simulation in our computational environment. 
The compilation size is set to $\tilde{L}=4$ ($N_q=8$), which is the smallest system size the next-nearest-neighbor hopping does not vanish. 
Although $\tilde{L}=4$ might be too small compared to the rigorous compilation size given by Eq.~\eqref{eq:L_compilation_explicit}, we expect that the LSVQC is applicable owing to the looseness of the LR bound. 
The optimization of the LSVQC is executed in the same way as that in Sec.~\ref{sec:numerical_test}. We also execute the LVQC in the same setting for comparison.  

\subsection{Double occupation dynamics}\label{subsec:dynamics}
First, we simulate the dynamics of a simple one-body local observable. 
Specifically, we consider the dynamics of the double occupation per site 
\begin{align}
    \braket{D(t)}_{\psi_0} &= \frac{1}{|\Lambda|}\sum_{i\in\Lambda} \bra{\psi_0(t)}n_{i,\uparrow}n_{i,\downarrow}\ket{\psi_0(t)}, \label{eq:double_occ}
\end{align}
where $\ket{\psi_0(t)}\equiv e^{-itH}\ket{\psi_0}$. 
The input state $\ket{\psi_0}$ is chosen as the ground state of the noninteracting Hamiltonian, which is obtained by setting $U=0$ in Eq.~\eqref{eq:Sr2CuO3_hamil}. 
We obtain this state for the particle number $N_e=L=8$ (i.e., half-filling) and total spin $S_z=0$. 
Note that the Hamiltonian~\eqref{eq:Sr2CuO3_hamil} preserves the particle number $N_e$ and total spin $S_z$. 
The noninteracting ground state $\ket{\psi_0}$ can be efficiently prepared on a quantum circuit using the Givens rotation network~\cite{GivensPRA2018}. 
We use OpenFermion package~\cite{openfermion2020} to determine the Givens rotation angles. 
Note that the Givens rotation network circuit has a local circuit structure and satisfies the requirement of the LSVQC. 
At $t>0$, the noninteracting ground state $\ket{\psi_0}$ is evolved under the interacting Hamiltonian~\eqref{eq:Sr2CuO3_hamil}. 
Since the double occupation $D=\sum_i n_{i,\uparrow}n_{i,\downarrow}/|\Lambda|$ does not commute with the interacting Hamiltonian~\eqref{eq:Sr2CuO3_hamil}, nontrivial behavior is expected under the time evolution. 

\begin{figure*}[htbp]
    \includegraphics[width=17.5cm]{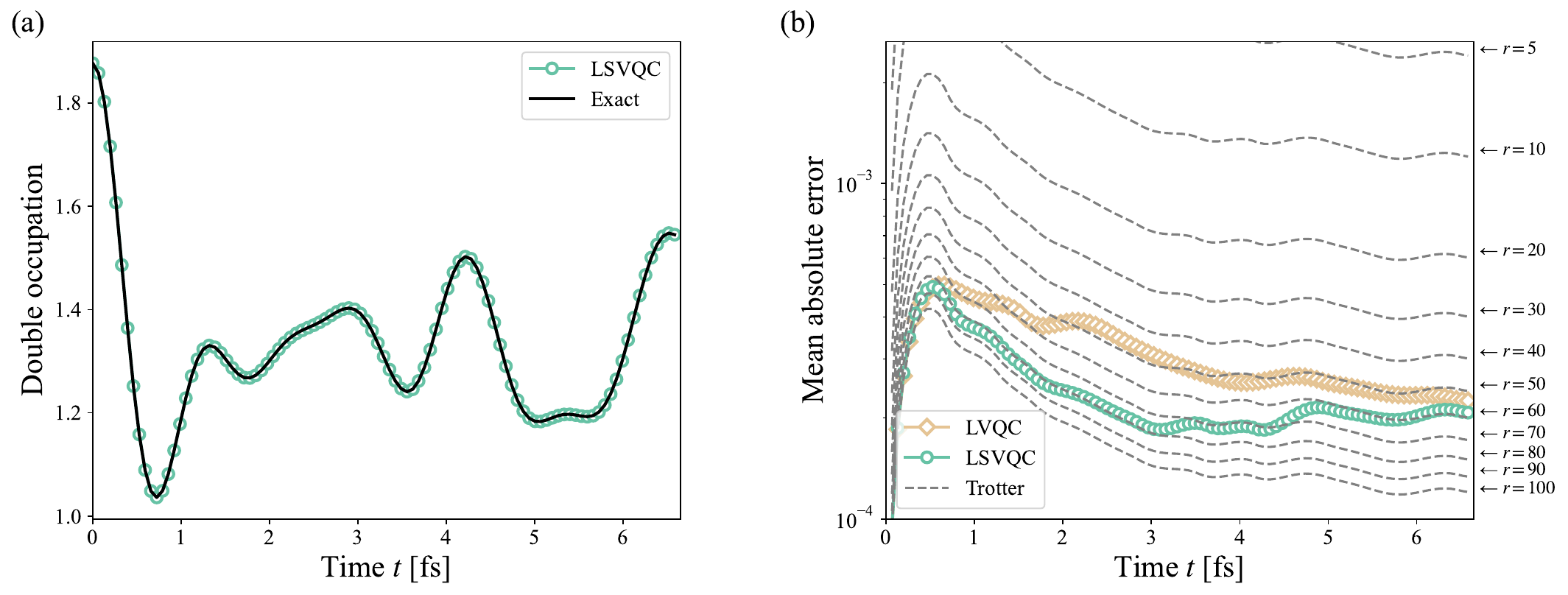}
    \caption{(a) Double occupation dynamics of the \textit{ab initio} downfolding model for \ce{Sr2CuO3} with $L=8$. The green circles and the black solid line indicate the results obtained by the LSVQC and exact calculation, respectively. 
    (b) MAE of the double occupation dynamics. The green circles (brown diamonds) represent the MAE of the LSVQC (LVQC). The grey dashed lines indicate the MAE of the Trotterization for various values of the depth parameter $r$. 
    The LSVQC and LVQC are performed for $\tilde{L}=4$, $\tau=0.1$ (in the unit of $\hbar=1$), and $N_L=5$. The double occupation is calculated in the time domain $t\in[0,10]$ with the time step $\tau=0.1$ in the unit of $\hbar=1$. 
    The time $t$ in the horizontal axis is shown in the femtosecond (fs) unit. } 
    \label{fig:LSVQC_DF_doublon}
\end{figure*}
We adopt the Krylov subspace given by Eq.~\eqref{eq:subspace_A} as a target subspace of the LSVQC. 
The parameters of the subspace are set to be $(N_t,\Delta_t)=(2,0.5)$ based on a parameter search using the state infidelity such as in Sec.~\ref{sec:numerical_test}. 
The depth of the VHA is fixed to $N_L=5$. 
The compilation time is chosen as $\tau=0.1$ in the unit of $\hbar=1$ (i.e., the time unit is 0.658 fs). 
The long-time dynamics at $t=n\tau$ ($n\in\mathbb{N}$) is simulated by repeatedly applying the optimized circuit to the initial state $\ket{\psi_0}$.   

In Fig.~\ref{fig:LSVQC_DF_doublon}(a), we compare the results of the exact calculation (black line) and the LSVQC (green dots). 
The exact result is obtained by applying the exact time-evolution unitary $e^{-itH^{(L)}}$ to the input state $\ket{\psi_0}$ within a subspace with $N_e=8$ and $S_z=0$. 
We see that the LSVQC results nicely reproduce the exact results in a wide range of time. 
To clarify the efficiency of the LSVQC in more detail, we assess the mean absolute error (MAE) of the double occupation as
\begin{align}
    \delta D_{\psi_0}(n\tau) &= \frac{1}{n}\sum_{j=1}^{n}\left| \braket{D(j\tau)}_{\psi_0}^{\rm approx} - \braket{D(j\tau)}_{\psi_0}^{\rm exact} \right|, \label{eq:MAE_D}
\end{align}
where $\braket{D(j\tau)}_{\psi_0}^{\rm exact}$ is the exact value of the double occupation, while $\braket{D(j\tau)}_{\psi_0}^{\rm approx}$ is the approximate value of the double occupation calculated by using the LSVQC, LVQC, or Trotterization. 
In Fig.~\ref{fig:LSVQC_DF_doublon}(b), we compare the MAE $\delta D_{\psi_0}$ for the LSVQC (green circles), LVQC (brown diamonds), and Trotterization (grey dashed line). 
The LSVQC and LVQC are performed for the same compilation size $\tilde{L}=4$ using the VHA with $N_L=5$. 
We see that the MAE of the LSVQC is much smaller than that of the LVQC and comparable to the value of the Trotterization with $r=60\mathchar`-80$.
This means that the LSVQC achieves a significant depth compression against the Trotterization as $R_{\epsilon}^{\rm LSVQC}\simeq 12\mathchar`-16$ (see Eq.~\eqref{eq:d_compress} for the definition of $R_{\epsilon}^{\rm LSVQC}$). 
In addition, the LSVQC realizes a comparable or better depth compression ratio compared to the LVQC. 
This result supports the intuition behind the LSVQC, that is subspace-based compiling achieves high accuracy compared to full Hilbert space compiling (e.g., LVQC) when using the same ansatz with limited expressive power. 

\subsection{Green's function}\label{subsec:GF}
Next, we consider the simulation of real-time GF using LSVQC. 
Although the double occupation dynamics simulation in Sec~\ref{subsec:dynamics} is a slightly artificial setting, the GF is a crucial quantity for practical materials simulations, especially for studying strongly correlated electron materials like \ce{Sr2CuO3}. 

We consider the retarded GF at zero temperature 
\begin{equation}
    G^{\rm R}_{a,b}(t) = -i\Theta(t)\bra{E_0}\{e^{iHt}c_{a}e^{-iHt}, c_{b}^{\dag}\}\ket{E_0},
    \label{eq:def-green}
\end{equation}
where $\{A,B\}=AB+BA$ denote the anticommutator, $\Theta(t)$ is the Heaviside step function, $\ket{E_0}$ is the ground state of the Hamiltonian $H$, and $c_a$ ($c_a^{\dag}$) is annihilation (creation) operator of an electron with some index $a$ specifying the fermionic mode. 
For our \textit{ab initio} downfolding model~\eqref{eq:Sr2CuO3_hamil}, $a=(i,\sigma)$ with spin $\sigma=\uparrow,\downarrow$ and site index $i\in\Lambda=\{1,2,\cdots,L\}$. 
Hereafter, we consider only the spin-diagonal component of the GF because the target Hamiltonian is the spin-preserving. 
The GF in the momentum space is obtained by the Fourier transformation of Eq.~\eqref{eq:def-green} as
\begin{align}
    G^{\rm R}_{\mathbf{k}\sigma}(t) &= \frac{1}{|\Lambda|}\sum_{i,j\in\Lambda}e^{-i\mathbf{k}\cdot(\mathbf{r}_i-\mathbf{r}_j)} G^{\rm R}_{i\sigma,j\sigma}(t), \label{eq:def_green_k}
\end{align}
where $\mathbf{k}$ is the momentum and $\mathbf{r}_i$ denotes the coordinate vector describing the site $i$. 
Then, we can calculate the (single-particle) spectral function as
\begin{align}
    A_{\mathbf{k}\sigma}(\omega) &= -\frac{1}{\pi}\mathrm{Im}\left[ \int_{-\infty}^{\infty}dt e^{i(\omega+i\eta)t}G^{\rm R}_{\mathbf{k}\sigma}(t) \right], 
    \label{eq:A_kw}
\end{align}
where $\eta(\to+0)$ is an infinitesimal positive number ensuring the convergence of the integral. 
From the spectral function, we obtain the density of states (DOS) as  
\begin{align}
    \rho_{\sigma}(\omega) &= \frac{1}{|\Lambda|}\sum_{\mathbf{k}} A_{\mathbf{k}\sigma}(\omega) .
    \label{eq:DOS}
\end{align}

Here, we give some comments regarding the importance of the above quantities in condensed matter physics and materials science. 
The retarded GF plays an important role when studying the nonequilibrium transport phenomena based on the linear response theory~\cite{bonch2015green,abrikosov2012methods,fetter2012quantum}. 
For example, the electronic conductivity in strongly correlated materials can be calculated from the retarded GF using the Kubo formula~\cite{Kubo1957}. 
The spectral function and DOS are crucial quantities describing the electronic structure of the strongly correlated materials.   
In experiments, the spectral function and DOS can be directly measured using spectroscopic techniques such as angle-resolved photoemission spectroscopy~\cite{ARPES-review} and scanning tunneling spectroscopy~\cite{STS-STM-review}. 

Now, we introduce a subspace appropriate for the GF computation. 
Under the Jordan-Wigner transformation of Eqs.~\eqref{eq:JW_annihilation} and~\eqref{eq:JW_creation}, the GF in Eq.~\eqref{eq:def-green} can be rewritten in the qubit representation as 
\begin{align}
    G^{\rm R}_{a,b}(t) &= -\frac{i}{2}\Theta(t)\left[ \mathrm{Re}\bra{E_0}e^{itH}\overrightarrow{X_a}e^{-itH}\overrightarrow{X_b}\ket{E_0} \right. \nonumber\\
    &\left. \qquad\qquad+\mathrm{Re}\bra{E_0}e^{itH}\overrightarrow{Y_a}e^{-itH}\overrightarrow{Y_b}\ket{E_0} \right.\nonumber\\
    &\left. \qquad\qquad -i\mathrm{Re}\bra{E_0}e^{itH}\overrightarrow{X_a}e^{-itH}\overrightarrow{Y_b}\ket{E_0}  \right. \nonumber\\
    &\left. \qquad\qquad-i\mathrm{Re}\bra{E_0}e^{itH}\overrightarrow{Y_a}e^{-itH}\overrightarrow{X_b}\ket{E_0} \right],
    \label{eq:def_green_JW}
\end{align}
where $\overrightarrow{X_a} \equiv X_a Z_{a-1} \cdots Z_{0}$ and $\overrightarrow{Y_a} \equiv Y_a Z_{a-1} \cdots Z_{0}$ with $P_a \equiv P_{i_\sigma}$ for $a=(i,\sigma)$ and $P=X,Y,Z$. 
From Eq.~\eqref{eq:def_green_JW}, we see that the GF is described in the form of a dynamical correlation function as in Eq.~\eqref{eq:dynamics_AB}. 
Thus, based on the subspace designing guideline described by Eq.~\eqref{eq:subspace_AB}, we can naively introduce the following subspace
\begin{align}
    \mathcal{S}' &= \mathrm{span}\Bigl( 
    \Bigl\{ U_{\rm trot}^{(L)}(t_n)\ket{E_0} \nonumber\\
    &\qquad\qquad \cup \bigl\{ U_{\rm trot}^{(L)}(t_n)R_{\overrightarrow{X_{j}}}(\phi)\ket{E_0} \bigr\}_{j=0}^{N_q-1} \nonumber\\
    &\qquad\qquad \cup \bigl\{ U_{\rm trot}^{(L)}(t_n)R_{\overrightarrow{Y_{j}}}(\phi)\ket{E_0} \bigr\}_{j=0}^{N_q-1}
    \Bigr\}_{n=0}^{N_t}
    \Bigr),
    \label{eq:subspace_GF_naive}
\end{align}
where $R_{P}(\phi)=e^{i\phi P}$ denotes the multi-qubit Pauli rotation gate with arbitrary angle $\phi$ and multi-qubit Pauli operator $P$. 
To ensure the nonorthogonality of the basis states, the angle $\phi$ is chosen as $0<\phi<\pi/2$. 
The physical meaning of the above subspace can be well understood in the original fermionic representation. 
The ground state $\ket{E_0}$ is a state with a fixed particle number $N_e$ because the Hamiltonian~\eqref{eq:Sr2CuO3_hamil} preserves the particle number. 
Since the multi-qubit Pauli rotation gates $R_{\overrightarrow{X_{a}}}(\phi)$ and $R_{\overrightarrow{Y_{a}}}(\phi)$ correspond to the single-particle excitation operator $e^{i\phi(c_a+c_a^{\dag})}$ and $e^{-\phi(c_a-c_a^{\dag})}$ in the original fermionic representation, the above subspace consists of states with the particle number $N_e$ and $N_e\pm1$. 
This is consistent with the definition of the GF~\eqref{eq:def-green} that is described using only states with the particle number $N_e$ and $N_e\pm1$.  
However, the subspace~\eqref{eq:subspace_GF_naive} can not be directly applied to the LSVQC since the state preparation circuit is nonlocal due to the nonlocality of $R_{\overrightarrow{X_{a}}}(\phi)$ and $R_{\overrightarrow{Y_{a}}}(\phi)$. 
To avoid this issue, we alternatively adopt the following subspace
\begin{align}
    \mathcal{S} &= \mathrm{span}\Bigl( 
    \Bigl\{ U_{\rm trot}^{(L)}(t_n)\ket{E_0} \nonumber\\
    &\qquad\qquad \cup \bigl\{ U_{\rm trot}^{(L)}(t_n)R_{X_{j}}(\phi)\ket{E_0} \bigr\}_{j=0}^{N_q-1} \nonumber\\
    &\qquad\qquad \cup \bigl\{ U_{\rm trot}^{(L)}(t_n)R_{Y_{j}}(\phi)\ket{E_0} \bigr\}_{j=0}^{N_q-1}
    \Bigr\}_{n=0}^{N_t}
    \Bigr),
    \label{eq:subspace_GF}
\end{align}
Here, the nonlocal rotation gates $R_{\overrightarrow{X_{a}}}(\phi)$ and $R_{\overrightarrow{Y_{a}}}(\phi)$ are replaced to the local rotation gates $R_{X_{a}}(\phi)$ and $R_{Y_{a}}(\phi)$, respectively. 
Note that this replacement does not change the fermionic particle number of the basis states, and hence the subspace~\eqref{eq:subspace_GF} is still physically reasonable to the GF computation. 
By setting $N_t=\mathcal{O}(\mathrm{poly}(L))$, the dimension of the subspace~\eqref{eq:subspace_GF} is bounded to polynomial of the system size as $(N_t+1)(2L+1)=\mathcal{O}(\mathrm{poly}(L))$. 
We set $(N_t, \Delta_t)=(1,0.5)$ and $\phi=0.4\pi$ in following the numerical simulations. 

\begin{figure*}[tbp]
    \includegraphics[width=17.5cm]{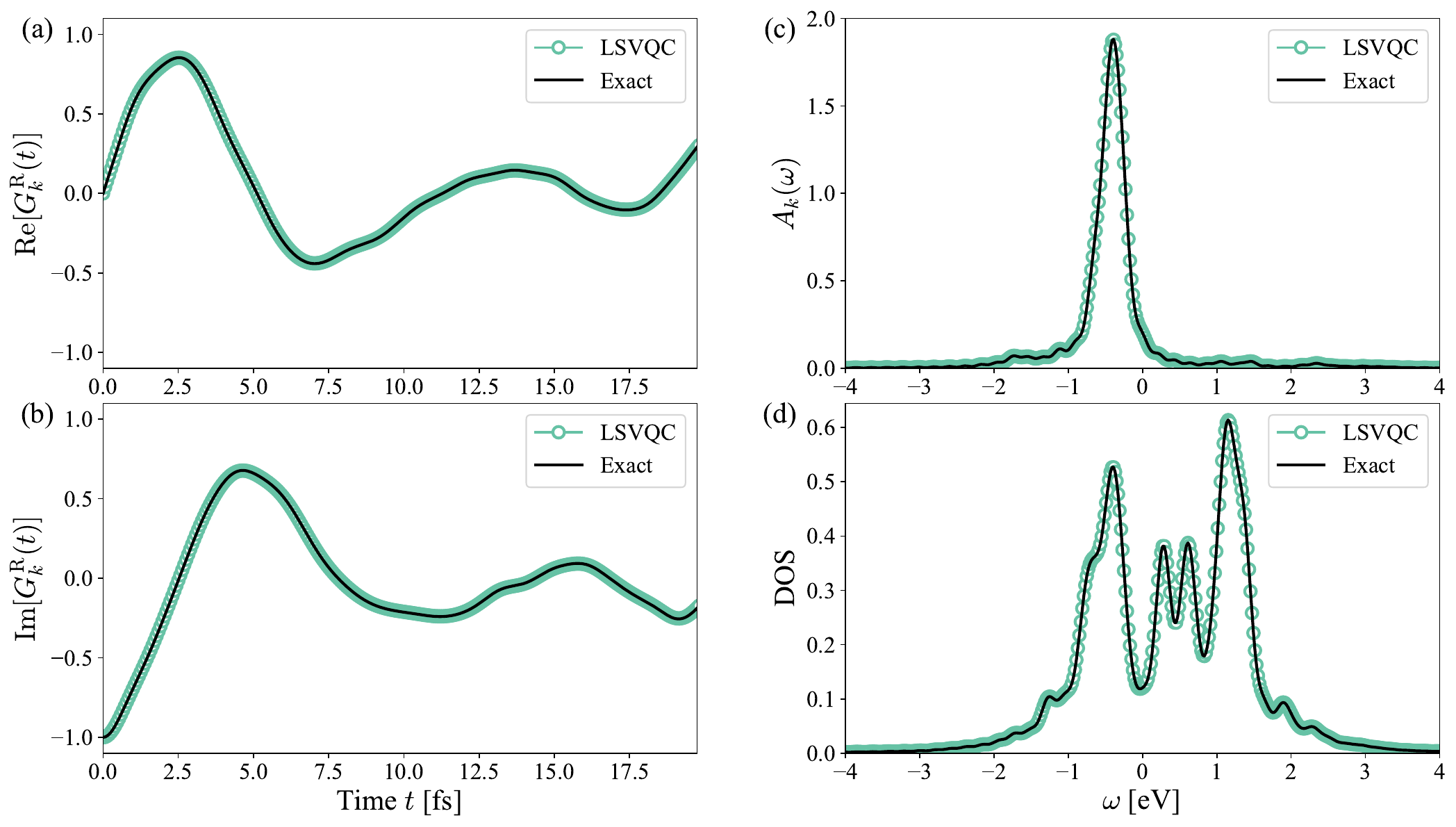}
    \caption{(a) Real part and (b) imaginary part of GF, (c) spectral function $A_{\mathbf{k}}(\omega)$, and (d) DOS of the \textit{ab initio} downfolding model for \ce{Sr2CuO3} with $L=8$. The green circles and the black solid line indicate the results obtained by the LSVQC and exact calculation, respectively. 
    The LSVQC is performed for $\tilde{L}=4$, $\tau=0.1$, and $N_L=5$. We show the GF $G_{\mathbf{k}}^{\rm R}(t)$ and spectral function $A_{\mathbf{k}}(\omega)$ for $\mathbf{k}=\pi/4$. The GF is calculated in the time domain $t\in[0,30]$ with the time step $\tau=0.1$ (in the unit of $\hbar=1$). The spectral function and DOS are obtained by performing the Fourier transformation of the real-time GF with $\eta=0.1$. The time $t$ in the horizontal axis of (a) and (b) is shown in the femtosecond (fs) unit, while the frequency $\omega$ in the horizontal axis of (c) and (d) is shown in the electronvolt (eV) unit.} 
    \label{fig:LSVQC_DF_GF}
\end{figure*}
Figure~\ref{fig:LSVQC_DF_GF} shows the GF, spectral functions, and DOS obtained by the LSVQC and exact calculation. 
The ground state $\ket{E_0}$ is approximately prepared by the standard variational quantum eigensolver (VQE)~\cite{VQE_peruzzo2014variational} using the VHA with the depth $N_L=5$. 
In the VQE calculation, the initial state is chosen as a noninteracting ground state prepared by the Givens rotation network~\cite{GivensPRA2018}. 
The particle number of the ground state $\ket{E_0}$ is set to be the half-filling (i.e., $N_e=L=8$). 
The LSVQC is performed to compile the time-evolution operator at $\tau=0.1$, and the long-time dynamics at $t=n\tau$ ($n\in\mathbb{N}$) is obtained by repeatedly applying the optimized circuit to the VQE ground state. 
The exact GF is calculated by using the VQE ground state $\ket{E_0}$ and exact time-evolution unitary $e^{-itH^{(L)}}$, which is implemented as a sparse matrix within the subspace with the particle numbers $N_e=8$, $N_e+1=9$, and $N_e-1=7$.  
From Fig.~\ref{fig:LSVQC_DF_GF}, we see that the exact GF, spectral function, and DOS are accurately reproduced by the LSVQC. 

\begin{figure*}[tbp]
    \includegraphics[width=17.5cm]{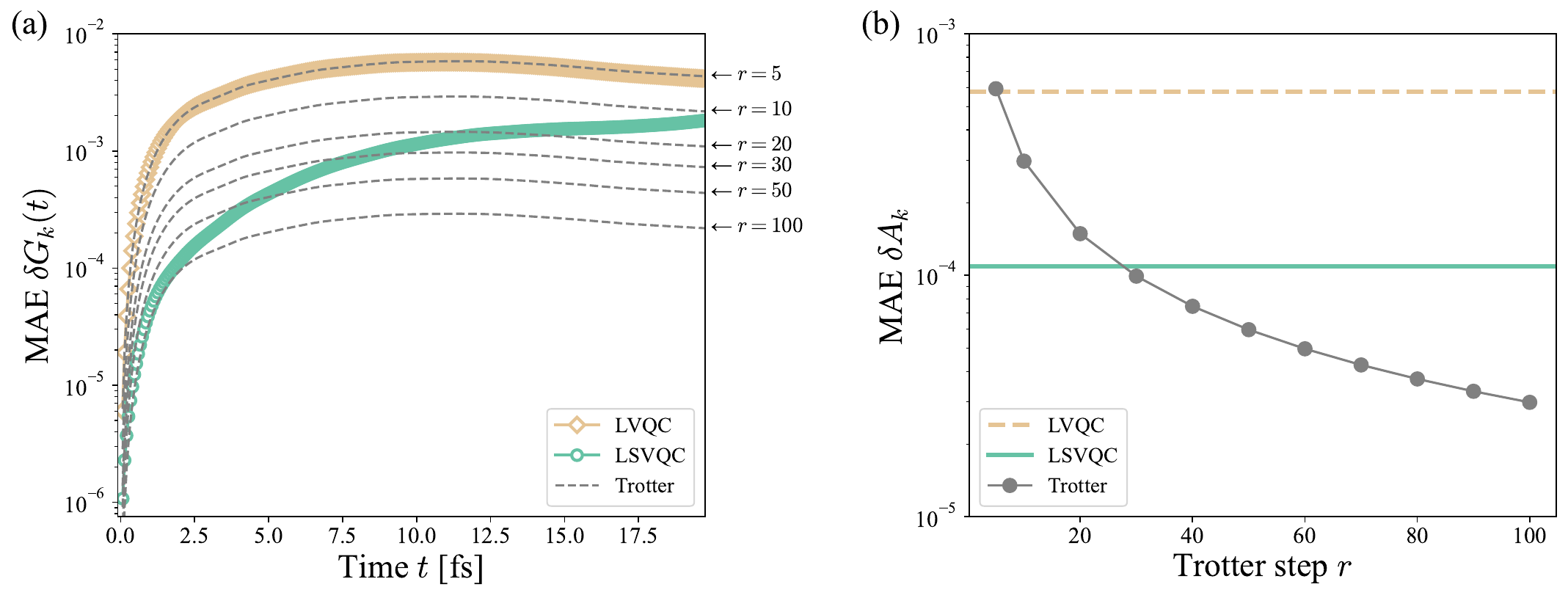}
    \caption{(a) MAE of the GF $\delta G_{\mathbf{k}}(t)$ for $\mathbf{k}=\pi/4$. The green circles (brown diamonds) represent the MAE of the LSVQC (LVQC). The grey dashed lines indicate the MAE of the Trotterization for various values of the depth parameter $r$. The time $t$ in the horizontal axis is shown in the femtosecond (fs) unit.
    (b) MAE of the spectral fucntion $\delta A_{\mathbf{k}}$ for $\mathbf{k}=\pi/4$.
    The green solid line (brown dashed line) represents the result of the LSVQC (LVQC). The grey dots indicate the MAE of the Trotterization as a function of the depth parameter $r$. 
    The LSVQC and LVQC are performed for $\tilde{L}=4$, $\tau=0.1$, and $N_L=5$. $\delta G_{\mathbf{k}}(t)$ is calculated in the time domain $t\in[0,30]$ with the time step $\tau=0.1$ (in the unit of $\hbar=1$). 
    $\delta A_{\mathbf{k}}$ is calculated with $\eta=0.1$, $\omega_0=5$ and $N_\omega=250$. } 
    \label{fig:LSVQC_DF_GF_MAE}
\end{figure*}
To see the superiority of the LSVQC compared to the Trotterization, we investigate the MAE of the GF 
\begin{align}
    \delta G_{\mathbf{k}}(n\tau) &= \frac{1}{n} \sum_{j=1}^{n}  \left|G^{\rm R, approx}_{\mathbf{k}}(j\tau)  - G^{\rm R, exact}_{\mathbf{k}}(j\tau) \right|,
    \label{eq:MAE_G}
\end{align} 
where $G^{\rm R, approx}_{\mathbf{k}}$ ($G^{\rm R, exact}_{\mathbf{k}}$) is the approximate (exact) value of the GF. 
Similarly, we define the MAE of the spectral function as
\begin{align}
    \delta A_{\mathbf{k}} &= \frac{1}{2N_\omega+1} \sum_{j=-N_\omega}^{N_\omega}  \left|A^{\rm  approx}_{\mathbf{k}}(\omega_j)  - A^{\rm exact}_{\mathbf{k}}(\omega_j) \right|,
    \label{eq:MAE_A}
\end{align} 
where $\omega_j=j\omega_0/N_\omega$ ($j=-N_\omega, -N_\omega+1, \cdots, N_\omega$) is a discritized frequency point, and $2N_\omega+1$ is the total number of the frequencies. 
In Fig.~\ref{fig:LSVQC_DF_GF_MAE}(a), we show the MAE of the GF $\delta G_{\mathbf{k}}$ at $\mathbf{k}=\pi/4$ for the LSVQC (green circles), LVQC (brown diamonds), and Trotterization (grey dashed lines). 
The LSVQC and LVQC are performed under the same condition with $\tilde{L}=4$ and $N_L=5$. 
The MAE $\delta G_{\mathbf{k}}$ of the LSVQC is much smaller than that of the Trotterization in a wide range of time, while the MAE of the LVQC is almost the same with the Trotterizarion (i.e., the LVQC fails to realize the depth compression against the Trotterization). 
Specifically, the LSVQC achieves the MAE comparable to the value of the Trotterization with $r=10\mathchar`-100$. 
Thus, the LSVQC succeeds in the depth compression against the Trotterization as $R^{\rm LSVQC}_{\epsilon} \simeq 2 \mathchar`- 20$. 
The depth compression is especially significant in a short time regime. 
This behavior is consistent with the result in Fig.~\ref{fig:LSVQC_Heisenberg_fildeity}(b). 
As shown in Sec.\ref{sec:numerical_test}, the depth compression rate at a long-time regime will be improved by logarithmically increasing the compilation size $\tilde{L}$ based on Eq.~\eqref{eq:L_compilation_explicit_longT}. 
Figure~\ref{fig:LSVQC_DF_GF_MAE}(b) shows the MAE of the spectral function $\delta A_{\mathbf{k}}$ at $\mathbf{k}=\pi/4$.  
We set $\omega_0=5$ and $N_\omega=250$. 
The horizontal green solid (brown dashed) line indicates the MAE $\delta A_{\mathbf{k}}$ for the LSVQC (LVQC) at $N_L=5$. 
We also show the result for the Trotterization as a function of the number of the Trotter step $r$. 
The MAE $\delta A_{\mathbf{k}}$ of the LSVQC (LVQC) is comparable to that of the Trotterization with $r\simeq 30$ ($r\simeq 5$). 
Thus, the LSVQC realizes the depth compression with $R^{\rm LSVQC}_{\epsilon}\simeq30/5=6$, while the LVQC fails to realize the depth compression (i.e., the depth compression rate of the LVQC is $5/5=1$).  
A remarkable point of the above results is that the LSVQC achieves a significantly better depth compression rate than the LVQC, although both the LSVQC and LVQC use the same ansatz. 
Similar to the results in Fig.~\ref{fig:LSVQC_DF_doublon}, this result strongly supports the intuition behind the LSVQC. The subspace-based compiling of the LSVQC leads to high accuracy compared to full Hilbert space compiling (i.e., LVQC) when using the same ansatz with limited expressive power. 

\section{Resource estimation}\label{sec:resource}
Quantitative estimates of the computational resources required for executing quantum algorithms provide crucial information for determining the feasibility of practical quantum computers.
In this section, we discuss the utility of the LSVQC for large-scale quantum simulations from the viewpoint of computational resources. 
We estimate the gate count needed to simulate the dynamics of the two-dimensional (2D)  Fermi-Hubbard model and \textit{ab initio} downfolding models for some representative materials. 
We estimate only the resources required to implement the time-evolution operator, as it consumes the most resources in quantum simulations.  
Then, we discuss the utility of the LSVQC on near-term quantum computing architectures in the NISQ and early-FTQC era. 

\subsection{Method}
Here, we outline the procedure of our resource estimation method. 
We first count the number of quantum gates needed to implement the time-evolution operator with an accuracy $\epsilon$ using the first-order Trotterization. 
It is realized by estimating the required Trotter step $r$ based on the asymptotic scaling of the Trotter error bound.  
We employ two types of Trotter error bound, i.e., {\it worst-case} and {\it average-case} errors~\cite{TrotErrorAve_PhysRevLett.129.270502}. 
The worst-case Trotter error is quantified by the spectral (i.e., operator) norm corresponding to the largest singular value.  
Although the worst-case error analysis is rigorous, it sometimes overestimates the error for some input states~\cite{TrotTight1_csahinouglu2021hamiltonian,TrotTight2_su2021nearly,TrotTight3_heyl2019quantum,TrotTight4_PRXQuantum.2.040305,TrotTight5_an2021time} and gives estimates that are too pessimistic for practical applications. 
To avoid such issues, we also adopt the average-case Trotter error~\cite{TrotErrorAve_PhysRevLett.129.270502}, which dictates the average error of the Trotterization for typical input states that appear in the context of quantum computing. 
Then, the required Trotter step for the Fermi-Hubbard and \textit{ab initio} downfolding models is estimated as
\begin{align}
    r^{\rm trot}_{\epsilon} &= 
    \begin{cases}
        \mathcal{O}(Lt^2/\epsilon) & (\text{worst-case}) \\
        \mathcal{O}(\sqrt{L}t^2/\epsilon) & (\text{average-case})
    \end{cases}
    . \label{eq:r_trot_hubbard}
\end{align}
where $L$ is the size of the lattice. 
In the following analysis, we asymptotically estimate $r^{\rm trot}_{\epsilon}$ using Eq.~\eqref{eq:r_trot_hubbard} by setting the prefactor as 1 for simplicity. 
More details about the Trotter error analysis are summarized in Appendix~\ref{append:resource}. 

Next, we heuristically estimate the gate count needed to achieve the accuracy $\epsilon$ using the LSVQC. 
To do this, we assume that the LSVQC realizes the depth compression over the Trotterization with $R_{\epsilon}^{\rm LSVQC}>1$ using the VHA (Eq.~\eqref{eq:VHA_Sr2CuO3}), whose circuit structure is equivalent to the Trotterization. 
Since we can not analytically estimate the depth compression ratio $R_{\epsilon}^{\rm LSVQC}$ due to the variational nature of the LSVQC, we empirically expect that the LSVQC achieves $R_{\epsilon}^{\rm LSVQC}\simeq10$ based on the numerical results in Sec.~\ref{sec:numerical_test} and~\ref{sec:DF}. 
We also assume that $R_{\epsilon}^{\rm LSVQC}$ is hardly affected by the system size because of the size insensitivity of the local compilation theorem (Eqs.~\eqref{eq:LSVQC_ineq_LLET} and~\eqref{eq:LSVQC_ineq_LET}) and numerical results in Sec.~\ref{sec:numerical_test} (Fig.~\ref{fig:LSVQC_Heisenberg_depth}). 
Under these assumptions, the number of layers of the VHA $N_L$ for the LSVQC is empirically estimated as 
\begin{align}
    N_L &= 
    \begin{cases}
        \lfloor r^{\rm trot}_{\epsilon}/R_{\epsilon}^{\rm LSVQC} \rfloor & r^{\rm trot}_{\epsilon}/R_{\epsilon}^{\rm LSVQC} > 1 \\
        1 & r^{\rm trot}_{\epsilon}/R_{\epsilon}^{\rm LSVQC} \leq 1
    \end{cases}
    . \label{eq:NL_estimate_LSVQC}
\end{align}
In the following analysis, we set $R_{\epsilon}^{\rm LSVQC}=10$. 

\subsubsection{NISQ devices}\label{subsec:resource-NISQ}
Under the above setup, we perform resource estimation for NISQ devices based on the capability of the error mitigation techniques. 
The error mitigation can be realized with a manageable number of additional sampling overhead when the total error in a quantum circuit is in the order of unity. 
For example, the probabilistic error cancellation can suppress the errors with a reasonable sampling overhead when $N_{\rm gate} p_{\rm gate} \lesssim 2$~\cite{PBE_PhysRevX.8.031027}, where $N_{\rm gate}$ is the number of gates and $p_{\rm gate}$ is the typical error rate of each gate operation. 
We here note that the single-qubit gate error rate $p_{1q}$ is in general much smaller than the two-qubit gate error rate $p_{2q}$. 
Since two-qubit gates are the dominant source of errors, we set $N_{2q} p_{2q} \leq 2$ as a condition of successful error mitigation in NISQ devices, where $N_{2q}$ is the number of two-qubit gates in a given quantum circuit. 

\subsubsection{Early-FTQC devices}\label{subsec:resource-STAR}
We also perform resource estimation for early-FTQC devices. 
Specifically, we consider the Space-Time efficient Analog Rotation quantum computing architecture (STAR architecture) proposed in Ref.~\cite{STAR_PRXQuantum.5.010337} as a prototypical quantum computing architecture in the early-FTQC era.  
The STAR architecture realizes the universal quantum computation utilizing the fault-tolerant Clifford and analog rotation gates. 
The Clifford gates are executed in a fault-tolerant manner using the lattice surgery~\cite{litinski2019game}, while the analog rotation gates are performed directly utilizing a special magic state. 
The STAR architecture is suitable for the early-FTQC era since the $T$ gate decomposition and magic state distillation, which are computationally expensive parts in the conventional FTQC architectures, are avoided. 

In the STAR architecture, the logical error mainly originates from the analog rotation gates that are not fault-tolerant. 
This is in contrast to the NISQ devices in which the error rate of the single-qubit rotation gates is generally much lower than that of the two-qubit gates. 
In Ref~\cite{STAR_PRXQuantum.5.010337}, the logical error rate from the analog rotation gates is obtained as
\begin{align}
    P_{L, \mathrm{rotation}} &= 4p_{\rm phys}/15 + \mathcal{O}(p_{\rm phys}^2), 
    \label{eq:p_logical_star}
\end{align}
where $p_{\rm phys}$ is the physical error rate. 
The above logical error from the analog rotation gates is mitigated by the probabilistic error cancellation~\cite{PBE_PhysRevX.8.031027}, while the logical error from the Clifford gates can be negligible by taking a sufficiently large code distance.  
Therefore, we only count the number of analog rotation gates $R_Z(\theta)=e^{-i(\theta/2)Z}$ as a possible source of the logical error. 
Similar to the analysis for NISQ devices, we judge that the quantum simulation is executable on the STAR architecture when the success condition of the probabilistic error cancellation $N_{R_Z}P_{L, \mathrm{rotation}}\leq2$~\cite{PBE_PhysRevX.8.031027} is satisfied. 
Here, $N_{R_Z}$ is the number of analog rotation gates $R_Z(\theta)$ in a given logical quantum circuit. 
$N_{R_Z}$ can be estimated by counting the number of terms in a given Hamiltonian since each term produces a multi-qubit Pauli rotation gate (i.e., single analog rotation gate and some Clifford gates) in the Trotterization or VHA circuit. 

\subsection{Results}

\subsubsection{Fermi-Hubbard model}
First, we show the results of resource estimation for the 2D Fermi-Hubbard model. 
We assume that the Hamiltonian is defined on a 2D square lattice with the open boundary condition.   
We also impose the half-filling condition for simplicity. 

Figure~\ref{fig:resource_hubbard_nisq} shows the results for the NISQ devices. 
The gate count is estimated specifically for superconducting qubit devices using the results of Ref.~\cite{FH_resource_PhysRevApplied.14.014059} (see Appendix~\ref{append:resource-FH} for details). 
We set the accuracy $\epsilon=0.01$ (i.e., 1\% error) for both the worst-case and average-case error schemes. 
Note that the average-case error estimation produces computational resources about one order of magnitude smaller than that for worst-case error estimation. 
Unless otherwise mentioned, we discuss the resource requirement based on the average-case error estimation in the following. 
The horizontal dashed lines indicate the allowed maximum number of two-qubit gates for the physical error rate of $p_{2q}=[10^{-3},10^{-4},10^{-5}]$. 
In Fig.~\ref{fig:resource_hubbard_nisq}(a), we show that the two-qubit gate count $N_{2q}$ for $5\times5$ site Fermi-Hubbard model (i.e., $L=25$) as a function of simulation time $t$. 
The required gate resources increase with time $t$. 
Note that the simulation of quantum many-body systems using classical computers generally becomes difficult when $t\gtrsim1$ in the unit of $\hbar/t_{\rm hop}$~\cite{qsim1_daley2022practical,qsim2_fauseweh2024quantum}, where $t_{\rm hop}$ is the hopping parameter (see Eq.~\eqref{eq:hamil_hubbard} for the definition of $t_{\rm hop}$). 
We see that the LSVQC (Trotterization) allows us to realize a quantum simulation of $t\simeq1$ when the two-qubit gate error rate is $p_{2q}\simeq 10^{-4}$ ($p_{2q}\simeq 10^{-5}$).
Figure~\ref{fig:resource_hubbard_nisq}(b) shows the size $L$ dependence of the two-qubit gate count for $t=1$.
As a typical guideline, the quantum advantage is expected at around $L=10\times10$ and $t=1$ for the quantum dynamics simulation of the 2D Fermi-Hubbard model~\cite{qsim1_daley2022practical,qsim2_fauseweh2024quantum}.
We see that the LSVQC (Trotterization) requires $p_{2q}\simeq 10^{-5}$ ($p_{2q}\simeq 10^{-6}$) to simulate a $10\times10$ site Fermi-Hubbard model. 
In summary, the acceptable two-qubit gate error rate of the LSVQC is about one order of magnitude larger than the Trotterization, although hardware improvement is still necessary to achieve $p_{2q}\lesssim 10^{-4}$ in superconducting qubit devices.
\begin{figure}[tbp]
    \includegraphics[width=8.5cm]{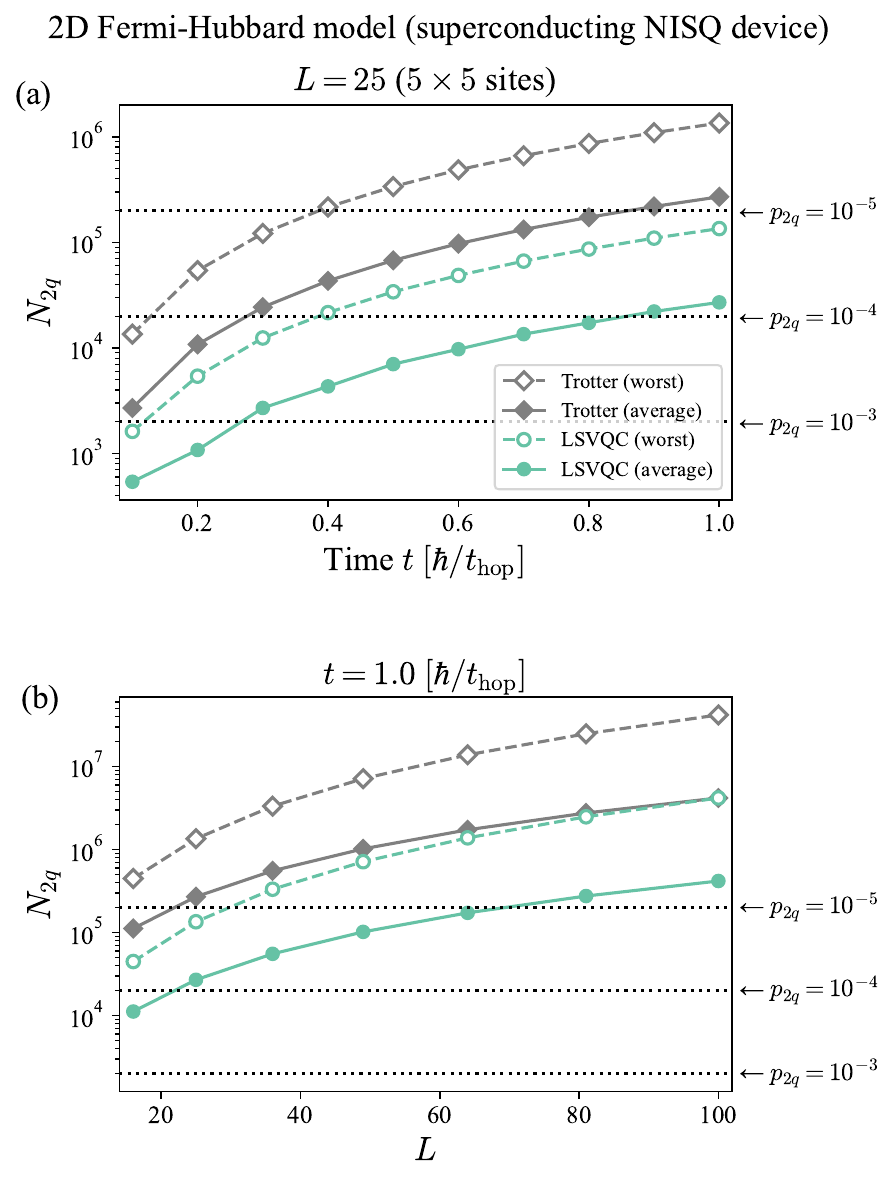}
    \caption{Gate count of the 2D Fermi-Hubbard model for superconducting qubit NISQ devices. 
    (a) Two-qubit gate count $N_{2q}$ as a function of the simulation time $t$ for $L=5\times5$. 
    (b) Two-qubit gate count $N_{2q}$ as a function of the lattice size $L$ for $t=1.0$. 
    The green filled (open) circles and grey filled (open) diamonds represent the results for the LSVQC and Trotterization in the average-case (worst-case) error scheme, respectively. 
    The dotted horizontal lines indicate the maximum two-qubit gate count allowed under two-qubit gate error rate $p_{2q}=[10^{-3},10^{-4},10^{-5}]$. } 
    \label{fig:resource_hubbard_nisq}
\end{figure}

Next, we perform resource estimation for the STAR architecture. 
Figure~\ref{fig:resource_hubbard_star}(a) shows the analog rotation gate count needed to simulate the time evolution of $5\times5$ site Fermi-Hubbard model as a function of time $t$. 
The horizontal dashed lines indicate the maximum number of analog rotation gates available in the STAR architecture for the physical error rate of $p_{\rm phys}=[10^{-3},10^{-4},10^{-5}]$. 
We see that the LSVQC realizes dynamics simulation of $t=1$ at $p_{\rm phys}\simeq10^{-3}$, which is achievable on current quantum computing devices. 
On the other hand, the Trotterization requires $p_{\rm phys}\lesssim10^{-4}$ at $t=1$. 
This is a remarkable improvement from the results for NISQ devices (Fig.~\ref{fig:resource_hubbard_nisq}), in which the dynamics simulation of $t=1$ and $L=5\times5$ requires $p_{\rm phys}\simeq10^{-4}$ even if we utilize the LSVQC. 
In Fig.~\ref{fig:resource_hubbard_star}(b), we show the analog rotation gate count needed to simulate the time evolution of $t=1$ as a function of the lattice size $L$. 
To simulate the $10\times10$ site model, the LSVQC (Trotterization) requires the physical error rate of $p_{\rm phys}\simeq10^{-4}$ ($p_{\rm phys}\simeq10^{-5}$). 
For both LSVQC and Trotterization, the acceptable physical error rate on the STAR architecture is about one order of magnitude larger than that on the NISQ devices.
These results indicate that combining the LSVQC with the STAR architecture allows us to further mitigate the hardware requirement needed to realize large-scale quantum simulations. 
We here note that the required number of physical qubits on the STAR architecture is generally much greater than that on the NISQ devices. 
To create $n$ logical qubits on the STAR architecture, at least $(1.5n+5)\times2d^2$ physical qubits are required for the code distance $d$~\cite{STAR_PRXQuantum.5.010337} (e.g., $73810$ physical qubits are required for $L=10\times10$ and $d=11$). 
Strictly speaking, the code distance $d$ should be carefully chosen considering the logical circuit depth under the lattice surgery~\cite{litinski2019game}.  
This suggests that the LSVQC also allows us to reduce the number of physical qubits by reducing the required code distance. 
More detailed resource estimation considering such effects is left for future work. 
\begin{figure}[tbp]
    \includegraphics[width=8.5cm]{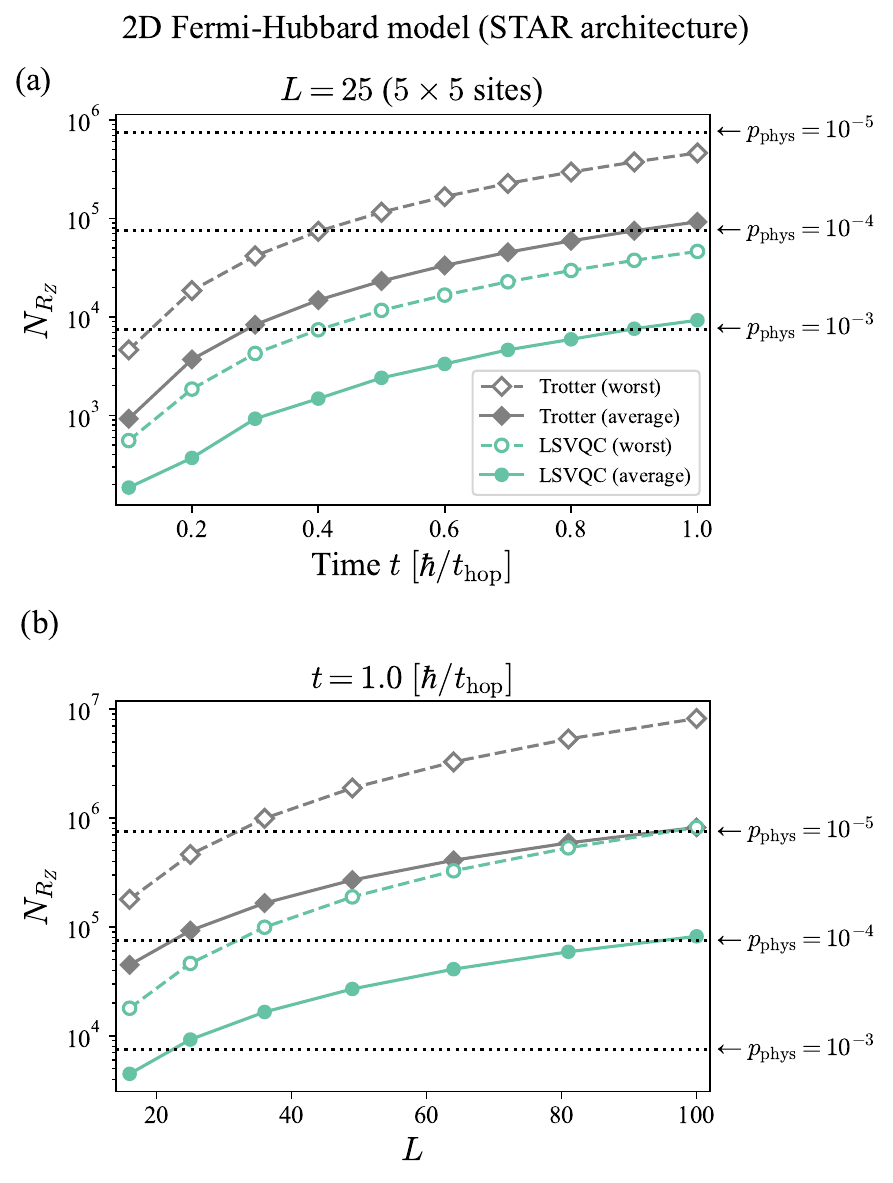}
    \caption{Gate count of the 2D Fermi-Hubbard model for STAR architecture. 
    (a) Analog rotation gate count $N_{R_Z}$ as a function of the simulation time $t$ for $L=5\times5$. 
    (b) Analog rotation gate count $N_{R_Z}$ as a function of the lattice size $L$ for $t=1.0$. 
    The green filled (open) circles and grey filled (open) diamonds represent the results for the LSVQC and Trotterization in the average-case (worst-case) error scheme, respectively. 
    The dotted horizontal lines indicate the maximum analog rotation gate count allowed under physical error rate $p_{\rm phys}=[10^{-3},10^{-4},10^{-5}]$. } 
    \label{fig:resource_hubbard_star}
\end{figure}

\subsubsection{\textit{Ab initio} downfolding models}
Next, we estimate the gate count for the \textit{ab initio} downfolding models (Eq.~\eqref{eq:DF_hamiltonian}) for some representative materials. 
We use the result of Ref.~\cite{DFresource_PhysRevA.106.012612}, in which the authors estimate the number of the single-qubit and two-qubit gates needed to implement a single Trotter step for the \textit{ab initio} downfolding models of strongly correlated materials. 
Our contribution in this section is to show the gate count and acceptable gate error rate for the quantum dynamics simulation, not just to implement a single Trotter step.
See Appendix~\ref{append:resource-DF} for details of the method of the gate count. 

As a benchmark result, we count the gate resources required to perform a short-time dynamics simulation of time $t=0.1$ with 1\% accuracy. 
The number of unit cells $N_{\rm cells}$ is set to $N_{\rm cells}=10$ so that the number of qubits becomes on the order of $10^2$ at maximum. 
Although $N_{\rm cells}=10$ may be too small for practical calculations, it is already in a region difficult to tackle with classical computers. 
The gate count and acceptable gate error rate are summarized in Table~\ref{tab:DF_resource}.
The column ``NISQ" (``STAR") indicates the resource estimates for NISQ devices (STAR architecture) based on the two-qubit CNOT gate count (analog rotation gate count). 
The acceptable error rate of the LSVQC is about one order of magnitude larger than the Trotterization for all materials in both ``NISQ" and ``STAR".
The LSVQC allows us to simulate organic compounds (TMTSF)$_2$PF$_6$ and K$_3$C$_{60}$ on NISQ devices with $p_{2q}\simeq10^{-4}$. 
For the STAR architecture, more complex compounds such as LaFeAsO (iron-based superconductor) and SrVO$_3$ (perovskite oxide) are simulatable using the LSVQC when $p_{\rm phys}\simeq10^{-4}$. 
The above results suggest that the LSVQC has great utility for the \textit{ab initio} materials simulation by combining it with a sophisticated quantum computing architecture like the STAR architecture. 
\begin{table*}[htbp]
    \centering
    \caption{Gate counts of the \textit{ab initio} downfolding models for $t=0.1$ and $N_{\rm cells}=10$. The values enclosed in parentheses indicate the estimates based on the worst-case Trotter error bound.}
    \begin{ruledtabular}
    \begin{tabular}{ccccccc}
         & & &\multicolumn{2}{c}{NISQ} & \multicolumn{2}{c}{STAR} \\
         \cline{4-5}  \cline{6-7}
        Material & Qubits & Method & CNOT & $p_{2q}$ & $R_Z$ & $p_{\rm phys}$  \\
        \hline 
        \multirow{2}{*}{(TMTSF)$_2$PF$_6$} & \multirow{2}{*}{40} & Trotter & $5.1\times10^4(2.5\times10^5)$ & $3.9\times10^{-5}(7.9\times10^{-6})$ & $2.2\times10^4(1.1\times10^5)$ & $3.4\times10^{-4}(6.8\times10^{-5})$ \\
         & & LSVQC & $1.3\times10^4(2.5\times10^4)$ & $3.4\times10^{-4}(6.8\times10^{-5})$ & $5.5\times10^3(1.1\times10^4)$ &$1.4\times10^{-3}(6.8\times10^{-4})$ \\
        \multirow{2}{*}{K$_3$C$_{60}$} & \multirow{2}{*}{60} & Trotter & $7.3\times10^4(4.4\times10^5)$ & $2.7\times10^{-5}(4.6\times10^{-6})$ & $1.3\times10^4(8.0\times10^4)$ & $5.6\times10^{-4}(9.3\times10^{-5})$ \\
         & & LSVQC & $1.5\times10^4(4.4\times10^4)$ & $1.4\times10^{-4}(4.6\times10^{-5})$ & $2.7\times10^{3}(8.0\times10^3)$ & $2.8\times10^{-3}(9.3\times10^{-4})$ \\
         \multirow{2}{*}{LaFeAsO} & \multirow{2}{*}{200} & Trotter & $1.5\times10^6(1.5\times10^7)$ & $1.3\times10^{-6}(1.3\times10^{-7})$ & $4.3\times10^5(4.3\times10^6)$ & $1.8\times10^{-5}(1.8\times10^{-6})$ \\
         & & LSVQC & $1.5\times10^5(1.5\times10^6)$ & $1.3\times10^{-5}(1.3\times10^{-6})$ & $4.3\times10^4(4.3\times10^5)$ & $1.8\times10^{-4}(1.8\times10^{-5})$ \\
         \multirow{2}{*}{NiO} & \multirow{2}{*}{100} & Trotter & $3.2\times10^6(2.3\times10^7)$ & $6.3\times10^{-7}(8.8\times10^{-8})$ & $1.4\times10^6(9.7\times10^6)$ & $5.5\times10^{-6}(7.7\times10^{-7})$  \\
         & & LSVQC & $4.5\times10^5(2.3\times10^6)$ & $4.4\times10^{-6}(8.8\times10^{-7})$ & $1.9\times10^5(9.7\times10^5)$ & $3.9\times10^{-5}(7.7\times10^{-6})$ \\
         \multirow{2}{*}{SrVO$_3$} & \multirow{2}{*}{100} & Trotter & $8.5\times10^5(6.0\times10^6)$ & $2.4\times10^{-6}(3.3\times10^{-7})$ & $2.8\times10^5(2.0\times10^6)$ & $2.7\times10^{-5}(3.8\times10^{-6})$ \\
         & & LSVQC & $1.2\times10^5(6.0\times10^5)$ & $1.7\times10^{-5}(3.3\times10^{-6})$ & $4.0\times10^4(2.0\times10^5)$ & $1.9\times10^{-4}(3.8\times10^{-5})$ \\
    \end{tabular}
    \end{ruledtabular}
    \label{tab:DF_resource}
\end{table*}

\section{Discussion and Conclusion}\label{sec:Summary}
In this paper, we have developed the LSVQC algorithm for resource-efficient simulation of large-scale quantum many-body systems on near-term quantum computers such as NISQ devices or early-FTQCs. 
The LSVQC performs the local optimization of the ansatz circuit to reproduce the action of the time-evolution operators within a subspace. 
The subspace is designed considering the properties of the target physical quantities we aim to simulate. 
The local compilation is theoretically guaranteed by the LR bound and locality of quantum circuits similarly to the LVQC. 
By restricting the range of optimization within a physically reasonable small subspace, the LSVQC allows us to accurately simulate the target physical quantities with fewer computational resources than the Trotterization and LVQC.  

The validity of the LSVQC has been demonstrated by numerical simulation using classical computers. 
Our numerical results for the 1D Heisenberg model show that the performance of the LSVQC is sensitive to the choice of the subspace and achieves a high accuracy under a proper subspace. 
It has also been clarified that the performance of the LSVQC is hardly altered by increasing the system size, which supports the scalability of the LSVQC to classically intractable large-scale quantum simulations. 
Furthermore, we have demonstrated that the LSVQC is effective in realistic materials simulations by performing quantum many-body dynamics simulations of the \textit{ab initio} downfolding model of \ce{Sr2CuO3}. 
It has been shown that the required circuit depth to simulate quantum many-body dynamics using the LSVQC is about 1/20 of the Trotterization at best. 
The LSVQC also achieves significant depth compression against the LVQC specifically in the computation of the Green's function.
Our heuristic resource estimation based on the above numerical results clarifies that the LSVQC allows us to relax the requirement on the gate error rate of quantum hardware to realize quantum simulations compared to the Trotterization.

We here provide some comments about possible future directions. 
The performance of the LSVQC is highly influenced by the subspace and ansatz circuit. 
Although we adopt the Trotterized Krylov subspace and VHA in this work, other subspaces and ansatzes might be more appropriate depending on the problem we wish to apply the LSVQC or performance of quantum computing hardware.  
A more detailed investigation of the dependence of the LSVQC on the choice of subspace and ansatz is one possible direction for future study. 
Another possible direction is a combined effort of the LSVQC and other sophisticated quantum computing techniques. 
In this work, we have briefly shown that the LSVQC archives a significant relaxation of the requirement on the gate error rate when it is combined with the STAR architecture.
Improvement of the LSVQC into a form suitable to the STAR architecture is left for future work. 
There is also a potential for further resource reduction in materials simulations by combining the LSVQC with some sophisticated techniques taking into account detailed materials properties (e.g., a hybrid fermion-to-qubit mapping in Ref.~\cite{clinton2024towards}). 

In conclusion, the LSVQC is a quantum-classical hybrid algorithm that enables us to efficiently implement the time-evolution operator of large-scale quantum many-body systems on near-term quantum computers. 
The superiority of the LSVQC over previously proposed methods has been clarified by numerical simulations of not only a simple toy model but also a realistic materials model based on the \textit{ab initio} electronic structure calculation. 
We hope this work gives new insight into the development of algorithms towards realizing practical materials simulation on near-term quantum computers. 

\section*{Acknowledgement}
The authors thank Riki Toshio and Yutaro Akahoshi for their helpful comments and discussions. 
We acknowledge Masatoshi Ishii for technical support in performing numerical simulations. 

\appendix




\section{Faithfulness of cost function for subspace-based compiling}\label{append:subspace_cost}
Here, we show that the cost function $C_{\rm LET}^{B_\mathcal{S}}$ given by Eq.~\eqref{eq:cost_LET} satisfies the faithfulness condition. 
In other words, we show that $C_{\rm LET}^{\mathcal{S}}=0$ if and only if $U_{\tau}\ket{\psi}=e^{i\varphi}V\ket{\psi}$ for $\forall\ket{\psi}\in\mathcal{S}=\mathrm{span}(\{\ket{\Psi_k}\}_{k=0}^{N-1})$ (For brevity, we describe $U_\tau^{(L)}$ and $V^{(L)}$ as $U_\tau$ and $V$ in this section). 
Here, the basis $\{\ket{\Psi_k}\}_{k=0}^{N-1}$ is assumed to be a linearly independent and nonorthogonal set of states. 

First, we prove the forward direction. 
Since $0\leq|\bra{\Psi_k}V^{\dag}U_\tau\ket{\Psi_k}|^2\leq1$ for all $k$, $C_{\rm LET}^{\mathcal{S}}=0$ indicates that the action of $V$ on $\ket{\Psi_k}$ coincides with the action of $U_\tau$ on $\ket{\Psi_k}$ up to a nonzero phase factor $e^{i\varphi_k}$ as follows:
\begin{align}
    V\ket{\Psi_k} = e^{i\varphi_k}U_\tau \ket{\Psi_k} \text{ for } k\in\{0,1,\cdots,N-1\}. \label{eq:subspace_cost_eq1}
\end{align}
In addition, when the basis states are linearly independent and nonorthogonal, the global phase factor $e^{i\varphi_k}$ is independent of $k$ (i.e., $e^{i\varphi_k}=e^{i\varphi}$ for all $k$). 
To see this, following Ref.~\cite{gibbs2022long}, let us consider two linearly independent and nonorthogonal states $\ket{\Psi_0}$ and $\ket{\Psi_1}$.  
Then, $\ket{\Psi_1}$ can be described as
\begin{align}
    \ket{\Psi_1} = c\ket{\Psi_0} + c_{\perp}\ket{\Psi_0^\perp}, \label{eq:subspace_cost_eq2}
\end{align}
where $|c_{\perp}|^2=1-|c|^2$, $|c|^2>0$, and $\braket{\Psi_0|\Psi_0^\perp}=0$.  
From Eqs.~\eqref{eq:subspace_cost_eq1} and~\eqref{eq:subspace_cost_eq2}, we obtain 
\begin{align}
    e^{-i\varphi_1} &= \bra{\Psi_1}V^{\dag}U_\tau\ket{\Psi_1} \nonumber\\
    &= |c|^2\bra{\Psi_0}V^{\dag}U_\tau\ket{\Psi_0} + |c_{\perp}|^2\bra{\Psi_0^\perp}V^{\dag}U_\tau\ket{\Psi_0^\perp} \nonumber\\
    &= |c|^2e^{-i\varphi_0} + (1-|c|^2)\bra{\Psi_0^\perp}V^{\dag}U_\tau\ket{\Psi_0^\perp}. \label{eq:subspace_cost_eq3}
\end{align}
Equation~\eqref{eq:subspace_cost_eq3} can be rewritten as 
\begin{align}
    &e^{-i\varphi_1} - \bra{\Psi_0^\perp}V^{\dag}U_\tau\ket{\Psi_0^\perp} \nonumber\\
    &= |c|^2 \left( e^{-i\varphi_0} - \bra{\Psi_0^\perp}V^{\dag}U_\tau\ket{\Psi_0^\perp} \right) . \label{eq:subspace_cost_eq4}
\end{align}
Here, Eq.~\eqref{eq:subspace_cost_eq4} indicates that two vectors $\bm{v}_1$ and $\bm{v}_2$, which correspond to the left- and right-hand side on Eq.~\eqref{eq:subspace_cost_eq4} respectively, coincide on the complex plane. 
Since $|c|^2>0$, $|e^{-i\varphi_{0,1}}|=1$, and $|\bra{\Psi_0^\perp}V^{\dag}U_\tau\ket{\Psi_0^\perp}|\leq1$, the above condition is satisfied only when 
\begin{align}
    e^{-i\varphi_0} = e^{-i\varphi_1} = \bra{\Psi_0^\perp}V^{\dag}U_\tau\ket{\Psi_0^\perp} . \label{eq:subspace_cost_eq5}
\end{align}
Equation~\eqref{eq:subspace_cost_eq5} indicates that $e^{i\varphi_0}=e^{i\varphi_1}$. Repeating the above calculation for all pairs of basis states, we obtain $e^{i\varphi_k}=e^{i\varphi_k'}\equiv e^{i\varphi}$ for $\forall k,k' \in \{ 0, \cdots, N-1\}$. 
Thus, $e^{i\varphi_k}=e^{i\varphi}$ for all $k$ in Eq.~\eqref{eq:subspace_cost_eq1} since $\{\ket{\Psi_k}\}_{k=0}^{N-1}$ is defined as a linearly independent and nonorthogonal basis set. 
Here, we note that any state $\ket{\psi}$ belonging to the subspace $\mathcal{S}$ can be represented as a linear combination of $\{\ket{\Psi_k}\}_{k=0}^{N-1}$ as 
\begin{align}
    \ket{\psi} &= \sum_{k=0}^{N-1} c_k \ket{\Psi_k}. \label{eq:subspace_cost_eq6}
\end{align}
From Eqs.~\eqref{eq:subspace_cost_eq1} and \eqref{eq:subspace_cost_eq6}, we obtain 
\begin{align}
    V\ket{\psi} &= \sum_{k=0}^{N-1} c_k e^{i\varphi_k} U\ket{\Psi_k} = e^{i\varphi}U\ket{\psi}, \label{eq:subspace_cost_eq7}
\end{align}
where we used $e^{i\varphi_k}=e^{i\varphi}$. 
Therefore, $U_{\tau}\ket{\psi}=e^{i\varphi}V\ket{\psi}$ for $\forall\ket{\psi}\in\mathcal{S}$ if $C_{\rm LET}^{\mathcal{S}}=0$. 

The reverse direction is somewhat trivial. 
If $U_{\tau}\ket{\psi}=e^{i\varphi}V\ket{\psi}$ for $\forall\ket{\psi}\in\mathcal{S}$, then $U_{\tau}\ket{\Psi_k}=e^{i\varphi}V\ket{\Psi_k}$ for all $k$. This indicates that 
\begin{align}
    |\bra{\Psi_k}V^{\dag}U\ket{\Psi_k}|^2 = |e^{i\varphi}|^2 = 1 . 
    \label{eq:subspace_cost_eq8}
\end{align}
Inserting Eq.~\eqref{eq:subspace_cost_eq8} to the right-hand side of Eq.~\eqref{eq:cost_LET}, we find $C_{\rm LET}^{\mathcal{S}}=0$. 

\section{Derivation of local compilation theorem}\label{append:LSVQC_detail}
In this Appendix, we provide details of the derivation of the local compilation theorem. 

\subsection{Proof of Proposition 1}
Here, we provide a proof of Proposition~\ref{prop1}. 
The proof is based on the LR bound. 
As a preliminary, we first modify the inequality of the LR bound given by Eq.~\eqref{eq:LR_bound}. 
Let us consider a local operator $O_{X_{R,j}}$ acting on the $R$-size domain centered at the $j$-th site, i.e., $X_{R,j}=\{ j'\in \Lambda | \mathrm{dist}(j,j')\leq R/2\}$. 
Then, the expression of the LR bound~\eqref{eq:LR_bound} can be rewritten in the following form~\cite{LRbound_PhysRevA.101.022333,LVQC_PRXQuantum.3.040302}:
\begin{align}
    \left\| \left(U^{(L)}_{\tau}\right)^{\dag} O_{X_{R,j}} U^{(L)}_{\tau} -e^{i\tau H^{(L',j)}} O_{X_{R,j}} e^{-i\tau H^{(L',j)}} \right\| \nonumber\\
    \leq C'\int_{L'/2-d_H}^{\infty}dr e^{-(r-v\tau)/\xi},
    \label{eq:LR_bound_v2} 
\end{align}
where $C'$ is a constant of $\mathcal{O}(\tau)$ independent of the system size. 
$L'$ is determined such that $L'/2-r_H$ is larger than $R/2$ (i.e., the radius of the support of $O_{X_{R,j}}$) and $H^{(L',j)}$ is the locally-restricted $L'$-size Hamiltonian given by Eq.~\eqref{eq:H_localized}. 
Using Eq.~\eqref{eq:LR_bound_v2}, Proposition 1 is derived as follows:
\begin{proof}
    From the definition in Eq.~\eqref{eq:cost_LLET_j}, the difference between $C_{\rm LLET}^{(j),B_\mathcal{S}}(U^{(L)}_{\tau},V^{(L)})$ and $C_{\rm LLET}^{(j),B_\mathcal{S}}(U^{(L',j)}_{\tau},V^{(L)})$ can be written as follows ($\tilde{V}_k^{(L)} \equiv V^{(L)} W_k^{(L)}$): 
    \begin{align}
        &\left| C_{\rm LLET}^{(j),B_\mathcal{S}}(U^{(L)}_{\tau},V^{(L)}) -  C_{\rm LLET}^{(j),B_\mathcal{S}}(U^{(L',j)}_{\tau},V^{(L)}) \right| , \nonumber\\
        & \leq \frac{1}{N}\sum_{k=1}^{N}\Bigl| \bra{\Psi_k} (U^{(L)}_{\tau})^{\dag} \tilde{V}^{(L)}_{k}\Pi_j (\tilde{V}^{(L)}_{k})^{\dag} U_{\tau}^{(L)} \ket{\Psi_k} \nonumber\\
        &\qquad\quad - \bra{\Psi_k} (U^{(L',j)}_{\tau})^{\dag} \tilde{V}^{(L)}_{k}\Pi_j (\tilde{V}^{(L)}_{k})^{\dag} U_{\tau}^{(L',j)} \ket{\Psi_k} \Bigr| , \nonumber\\
        & \leq \frac{1}{2N}\sum_{k=1}^{N}\Bigl\| (U^{(L)}_{\tau})^{\dag} \tilde{V}^{(L)}_{k} (Z_j \otimes \mathbbm{1}_{\bar{j}} ) (\tilde{V}^{(L)}_{k})^{\dag} U_{\tau}^{(L)} \nonumber\\
        &\qquad\quad - (U^{(L',j)}_{\tau})^{\dag} \tilde{V}^{(L)}_{k} (Z_j \otimes \mathbbm{1}_{\bar{j}} )(\tilde{V}^{(L)}_{k})^{\dag} U_{\tau}^{(L',j)} \Bigr\|,
        \label{eq:cost_error_norm}
    \end{align}
    where we expanded the local projection operator $\Pi_j$ as $\Pi_j=\frac{1}{2}(\mathbbm{1}_j+Z_j)\otimes\mathbbm{1}_{\bar{j}}$ with $\mathbbm{1}_{\bar{j}}\equiv\mathbbm{1}_1\otimes \cdots \otimes \mathbbm{1}_{j-1} \otimes \mathbbm{1}_{j+1}\otimes \cdots \otimes \mathbbm{1}_L$. 
    Since we assume that the ansatz $V^{(L)}$ and the state preparation circuits $W_{k}^{(L)}$ are local such as in Eqs.~\eqref{eq:V_L} and~\eqref{eq:W_L}, $\tilde{V}^{(L)}_{k} (Z_j \otimes \mathbbm{1}_{\bar{j}} ) (\tilde{V}^{(L)}_{k})^{\dag}$ is regarded as $R_k$-local operator with $R_k=4(d_V+d_{W_k})$ (see Fig~\ref{fig:locality_VW}).
    \begin{figure}[tbp]
        \includegraphics[width=9cm]{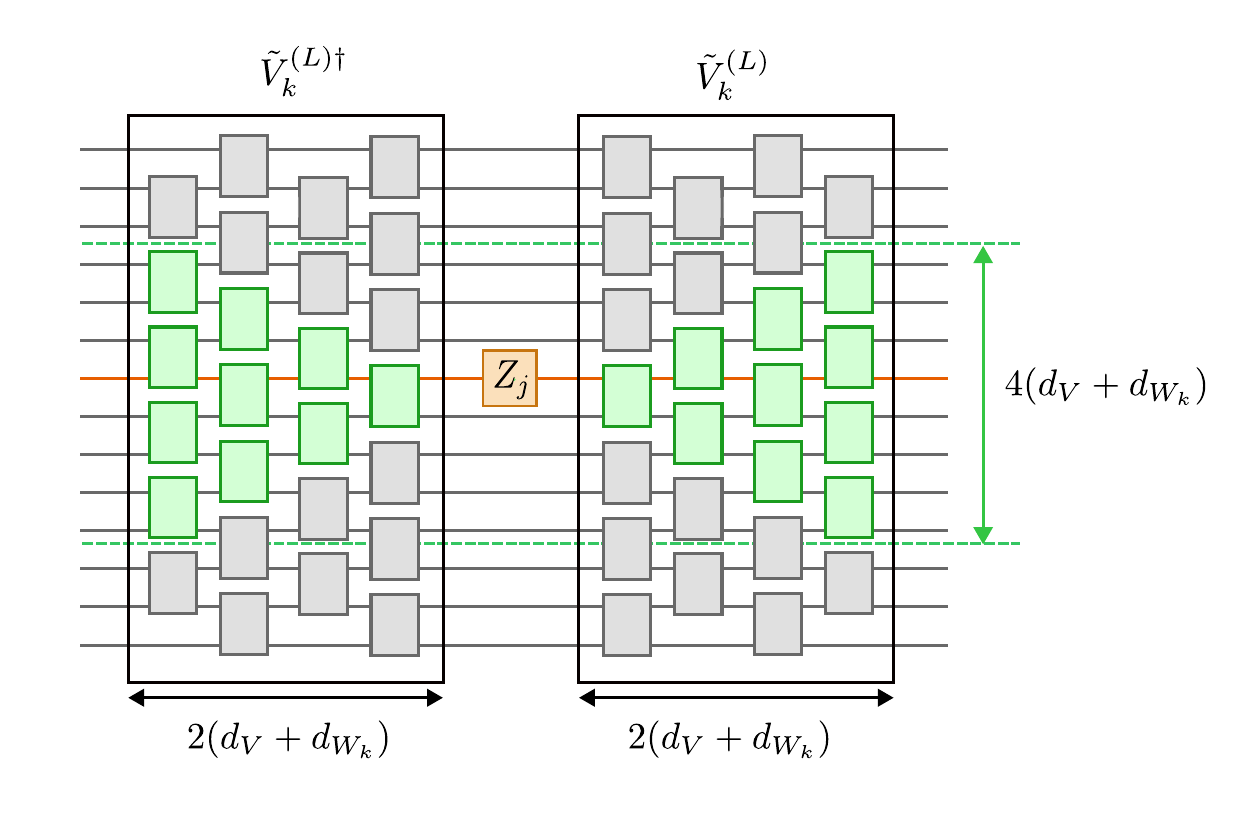}
        \caption{Schematic illustration of the locality of the operator $O_{k,j}\equiv\tilde{V}^{(L)}_{k} (Z_j \otimes \mathbbm{1}_{\bar{j}} ) (\tilde{V}^{(L)}_{k})^{\dag}$. Only green two-qubit gates are active in $O_{k,j}$, while other grey two-qubit gates vanish by taking the contraction. Consequently, $O_{k,j}$ is a $R_k$-local operator with $R_k=4(d_V+d_{W_k})$. } 
        \label{fig:locality_VW}
    \end{figure}
    Thus, we can apply Eq.~\eqref{eq:LR_bound_v2} to Eq.~\eqref{eq:cost_error_norm} by setting $R\equiv\max_k{R_k}=4(d_V+\max_k{d_{W_k}})$. 
    This leads to the following inequality:
    \begin{align}
        \left| C_{\rm LLET}^{(j),B_\mathcal{S}}(U^{(L)}_{\tau},V^{(L)}) -  C_{\rm LLET}^{(j),B_\mathcal{S}}(U^{(L',j)}_{\tau},V^{(L)}) \right| \nonumber\\
        \leq \frac{C'}{2}\int_{L'/2-d_H}^{\infty}dr e^{-(r-v\tau)/\xi}. \label{eq:LR_bound_cost}
    \end{align}
    For $D$-dimensional lattice, the integral in Eq.~\eqref{eq:LR_bound_cost} is evaluated as follows~\cite{LRbound_PhysRevA.101.022333,LVQC_PRXQuantum.3.040302}:
    \begin{align}
        &\int_{L'/2-d_H}^{\infty}dr e^{-(r-v\tau)/\xi} \nonumber\\
        &\leq C''\left( \frac{L'}{2}-r_H\right)^{D-1}e^{-(L'/2-r_H-v\tau)/\xi}, \label{eq:LR_error_integral}
    \end{align}
    where $C''$ is a constant. 
    Then, by describing $L'$ as in Eq.~\eqref{eq:L_restrict}, we obtain Eq.~\eqref{eq:LSVQC_proposition_1} from Eqs.~\eqref{eq:LR_bound_cost} and~\eqref{eq:LR_error_integral}. 
\end{proof}

\subsection{Proof of Proposition 2}
Proposition~\ref{prop2} is derived based on the LR bound and causal cone of the quantum circuit structure. 
The proof is given as follows: 
\begin{proof}
    We consider the causal cone of quantum circuits used to compute the cost functions $C_{\rm LLET}^{(j),B_\mathcal{S}}(U_{\tau}^{(L',j)},V^{(L)})$. 
    From Eq.~\eqref{eq:cost_LLET_j}, we only need to focus on the $\mathrm{Tr}[\Pi_j\rho_k(U_{\tau}^{(L',j)},V^{(L)})]$ term that can be schematically expressed in Fig.~\ref{fig:causal_cone}(a). 
    There are two causal cones in the structure of $\mathrm{Tr}[\Pi_j\rho_k(U_{\tau}^{(L',j)},V^{(L)})]$. 
    The first causal cone is from the locality of the measurement operator $\Pi_j$, which is depicted as an orange region in Fig.~\ref{fig:causal_cone}(a). 
    The grey two-qubit gates in Fig~\ref{fig:causal_cone}(a) are inactive gates due to this causal cone since they are contracted to the identity gate. 
    The second causal cone (blue regions in Fig.~\ref{fig:causal_cone}(a)) stems from the LR bound, i.e., the local restriction of the time-evolution operator. 
    Considering this second causal cone, the white two-qubit gates can be contracted to the identity gate.
    Hence, only green two-qubit gates in Fig.~\ref{fig:causal_cone}(a) are relevant to evaluate $\mathrm{Tr}[\Pi_j\rho_k(U_{\tau}^{(L',j)},V^{(L)})]$. 
    This indicates we can restrict the width of the quantum circuits $V^{(L)}$ and $W_k^{(L)}$ to some compilation size $\tilde{L}_k$ from $L$. 
    By considering the intersection of two causal cones as shown in Fig.~\ref{fig:causal_cone}(b), we can estimate that the proper compilation size is given by
    \begin{align}
        \tilde{L}_k = \mathrm{max}\left( \frac{L'}{2}+2(d_V+d_{W_k})+1, L'+4d_{W_k}\right).
        \label{eq:L_compilation_k}
    \end{align}
    When the compilation size $\tilde{L}$ is set to $\tilde{L}\geq \max_{k}{\tilde{L}_k}$ and satisfies Eq.~\eqref{eq:L_compilation}, the locally restricted circuits $V^{(\tilde{L},j)}$ and $W_k^{(\tilde{L},j)}$ contain all active two-qubit gates. 
   Thus, the term $\mathrm{Tr}[\Pi_j\rho_k(U_{\tau}^{(L',j)},V^{(L)})]$ can be rewritten as 
    \begin{align}
        &\mathrm{Tr}[\Pi_j\rho_k(U_{\tau}^{(L',j)},V^{(L)})] \nonumber\\
        & = \bra{\Psi_k}(U_{\tau}^{(L',j)})^{\dag} \tilde{V}_k^{(L)} \Pi_j (\tilde{V}_k^{(L)})^{\dag} U_{\tau}^{(L',j)} \ket{\Psi_k} \nonumber\\
        &= \bra{\bm{0}}(W_k^{(\tilde{L},j)})^{\dag}(\tilde{U}_{\tau}^{(L',j)})^{\dag} V_k^{(\tilde{L},j)} W_k^{(\tilde{L},j)} \nonumber\\
        &\quad\times\Pi_j (W_k^{(\tilde{L},j)})^{\dag} (V_k^{(\tilde{L},j)})^{\dag} \tilde{U}_{\tau}^{(L',j)} W_k^{(\tilde{L},j)} \ket{\bm{0}}. \label{eq:LSVQC_proposition_2_derivation}
    \end{align}
    From Eq.~\eqref{eq:LSVQC_proposition_2_derivation}, we can derive Eq.~\eqref{eq:LSVQC_proposition_2}. 
    \begin{figure}[htbp]
        \includegraphics[width=9.cm]{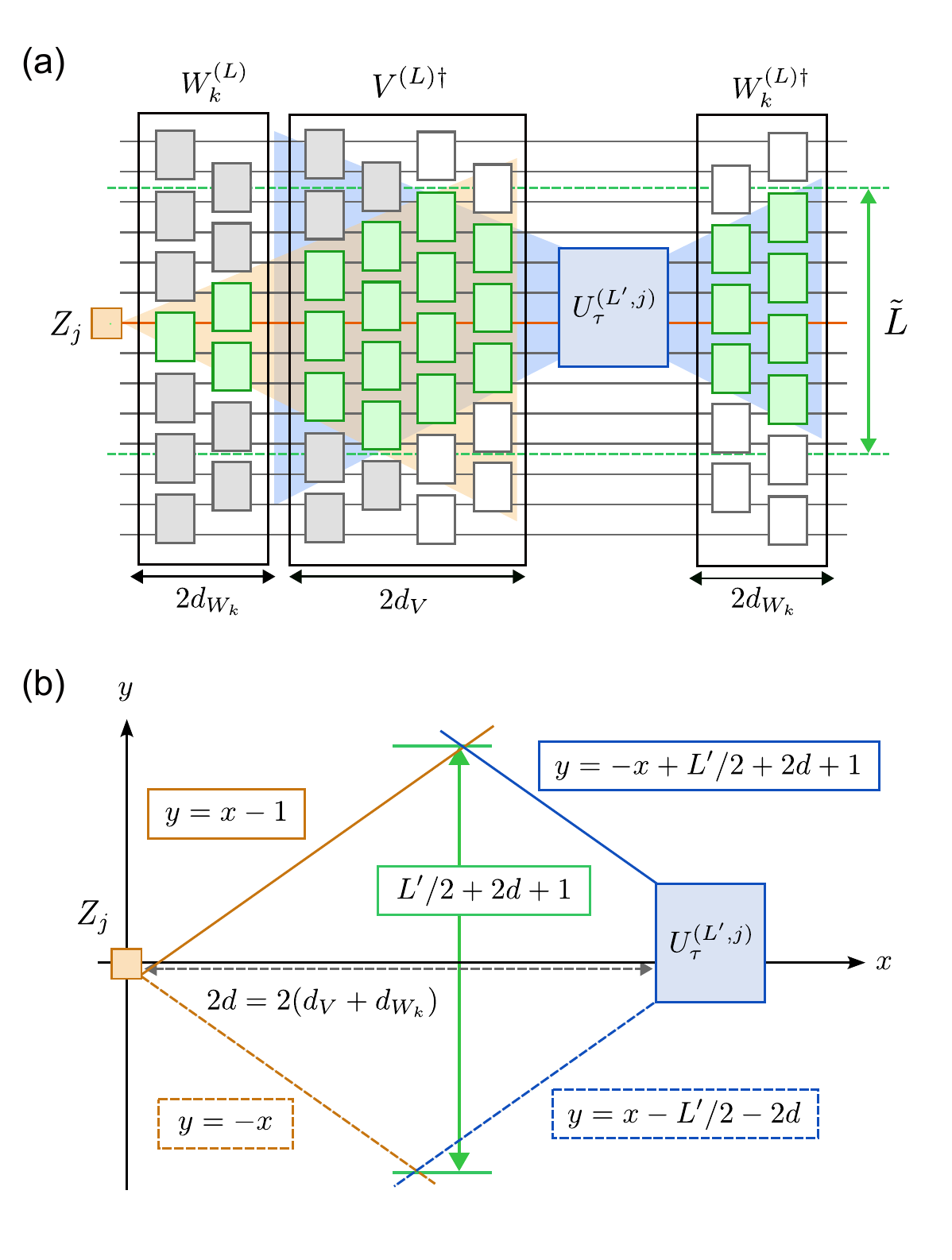}
        \caption{(a) Schematic illustration of the causal cones in a quantum circuit to evaluate the local cost function $C_{\rm LLET}^{(j),B_\mathcal{S}}(U_{\tau}^{(L',j)},V^{(L)})$. Only green two-qubit gates are active for computing the local cost function. The grey (white) two-qubit gates vanish by the contraction owing to the orange (blue) causal cone arising from the locality of the measurement operator (LR bound). 
        (b) Schematic illustration of the estimation of the compilation size. A proper compilation size is estimated by considering an intersection of two causal cones geometrically. } 
        \label{fig:causal_cone}
    \end{figure}
\end{proof}

\subsection{Proof of Theorem 1}
Theorem~\ref{thm1} (i.e., the local compilation theorem) is derived by combining Proposition~\ref{prop1} with Proposition~\ref{prop2} as follows: 
\begin{proof}
    From Eqs.~\eqref{eq:LSVQC_proposition_1} and~\eqref{eq:LSVQC_proposition_2}, we obtain 
    \begin{align}
        C_{\rm LLET}^{(j),B_\mathcal{S}}(U^{(L)}_{\tau},V^{(L)}) \leq C_{\rm LLET}^{(j),B_{\tilde{\mathcal{S}}_j}}(\tilde{U}_{\tau}^{(L',j)},V^{(\tilde{L},j)}) + \frac{1}{2}\epsilon_{\rm LR}. \label{eq:LSVQC_thm_deriv1}
    \end{align}
    Since the $H^{(L',j)}$ is a local restriction of $H^{(\tilde{L},j)}$ ($\tilde{L} \geq L'$), we can use Eq.~\eqref{eq:LR_bound_v2} by replacing $U^{(L)}_{\tau} \to U^{(\tilde{L},j)}_{\tau}$. 
    Thus, by performing the same discussion with the proof of Proposition~\ref{prop1}, we obtain 
    \begin{align}
        C_{\rm LLET}^{(j),B_{\tilde{\mathcal{S}}_j}}(\tilde{U}^{(L',j)}_{\tau},V^{(\tilde{L},j)}) \leq C_{\rm LLET}^{(j),B_{\tilde{\mathcal{S}}_j}}(U^{(\tilde{L},j)}_{\tau},V^{(\tilde{L},j)}) + \frac{1}{2}\epsilon_{\rm LR}. \label{eq:LSVQC_thm_deriv2}
    \end{align}
    Combining Eqs.~\eqref{eq:LSVQC_thm_deriv1} with~\eqref{eq:LSVQC_thm_deriv2}, we get
    \begin{align}
        C_{\rm LLET}^{(j),B_\mathcal{S}}(U^{(L)}_{\tau},V^{(L)}) \leq C_{\rm LLET}^{(j),B_{\tilde{\mathcal{S}}_j}}(U^{(\tilde{L},j)}_{\tau},V^{(\tilde{L},j)}) + \epsilon_{\rm LR}. \label{eq:LSVQC_thm_deriv3}
    \end{align}
    Taking the average of both sides of Eq.~\eqref{eq:LSVQC_thm_deriv3} over $j$, we obtain Eq.~\eqref{eq:LSVQC_ineq_LLET}. 
    Equation~\eqref{eq:LSVQC_ineq_LET} is straightforwardly derived from Eq.~\eqref{eq:LSVQC_ineq_LLET} using the relation between the global cost and local cost (Eq.~\eqref{eq:cost_ineq_general}). 
\end{proof}

\section{Extension of local compilation theorem}\label{append:local_comp_thms}
In this Appendix, we provide some modified versions of the local compilation theorem. 

\subsection{Translationally invariant systems}
We first consider modifying the original local compilation theorem for translationally invariant systems. 
Let us consider the translationally invariant Hamiltonian $H^{(L)}$ under the periodic boundary condition. 
In this case, if both the ansatz $V^{(L)}(\bm{\theta})$ and the state preparation circuits $\{W_k^{(L)}\}$ also exhibits translationally invariant structures, Theorem~\ref{thm1} can be simplified to the following theorem.  
\begin{thm}[Local compilation theorem for translationally invariant systems]\label{thm2}
    Suppose that we have an optimal parameter set $\bm{\theta}^*$ such that
    \begin{align}
        C_{\rm LLET}^{B_{\tilde{\mathcal{S}}}}(U_{\tau}^{(\tilde{L})},V^{(\tilde{L})}(\bm{\theta}^*))&\leq\tilde{\epsilon}_{\rm opt} , \label{eq:cost_optimal_condition_periodic}
    \end{align}
    with $B_{\tilde{\mathcal{S}}} = \{W_k^{(\tilde{L})}\ket{0}^{\otimes\tilde{L}}\}_{k=0}^{N-1}$ and $\tilde{\mathcal{S}} = \mathrm{span}(B_{\tilde{\mathcal{S}}})$. 
    Here, $U^{(\tilde{L})}_{\tau}=e^{-i\tau H^{(\tilde{L})}}$, $V^{(\tilde{L})}(\bm{\theta})$, and $\{W_k^{(\tilde{L})}\}$ are translationally invariant. 
    Then, the cost functions for the original $L$-size system are bounded as
    \begin{align}
        C_{\rm LLET}^{B_\mathcal{S}}(U_{\tau}^{(L)},V^{(L)}(\bm{\theta}^*)) &\leq \tilde{\epsilon}_{\rm opt} + \epsilon_{\rm LR}, \label{eq:LSVQC_ineq_LLET_translation} \\
        C_{\rm LET}^{B_\mathcal{S}}(U_{\tau}^{(L)},V^{(L)}(\bm{\theta}^*)) &\leq L\left(\tilde{\epsilon}_{\rm opt} + \epsilon_{\rm LR}\right). \label{eq:LSVQC_ineq_LET_translation}
    \end{align}
\end{thm}
\begin{proof}
    Since $U^{(L)}_{\tau}$, $V^{(L)}(\bm{\theta})$ and $\{W_k^{(L)}\}$ are assumed to be translationally invariant, the summands in Eq.~\eqref{eq:subsystem_cost_fn} are independent of $j$. 
    Although $U_{\tau}^{(\tilde{L},j)}$, $V^{(\tilde{L},j)}$ and $\{W_k^{(\tilde{L},j)}\}$ are not translationally invariant in general, we can recover the translational symmetry by padding some inactive gates at the boundary (see Fig. 5 in Ref.~\cite{LVQC_PRXQuantum.3.040302} for reference). 
    This allows us to rewrite the summands in Eq.~\eqref{eq:subsystem_cost_fn} to $C_{\rm LLET}^{B_{\tilde{\mathcal{S}}}}(U_{\tau}^{(\tilde{L})},V^{(\tilde{L})}(\bm{\theta}^*))$, and hence the above theorem is justified. 
\end{proof}

\subsection{Extension to long-time dynamics}
We have derived the local compilation theorem for a fixed time $\tau$ in the main text. 
However, in practical applications, we usually compute long-time dynamics by repeatedly applying the optimized circuit $V^{(L)}(\bm{\theta}^*)$ to some input states.  
Considering such situations, we explicitly provide a local compilation theorem for long-time dynamics simulation as follows: 
\begin{thm}[Local compilation theorem for long-time dynamics]\label{thm3}
    Suppose we aim to simulate long-time dynamics at $t=n\tau$ ($n\in\mathbb{N}$). 
    Then, using $\epsilon_{\rm opt}$ and $\epsilon_{\rm LR}$ in Eqs.~\eqref{eq:LSVQC_ineq_LLET} and~\eqref{eq:LSVQC_ineq_LET}, the cost functions for the original $L$-size system for time $t=n\tau$ are bounded as
    \begin{align}
        C_{\rm LLET}^{B_\mathcal{S}}(U_{n\tau}^{(L)},(V^{(L)}(\bm{\theta}^*))^n) &\lesssim n^2\left(\epsilon_{\rm opt} + \epsilon_{\rm LR}\right), \label{eq:LSVQC_ineq_LLET_longT} \\
        C_{\rm LET}^{B_\mathcal{S}}(U_{n\tau}^{(L)},(V^{(L)}(\bm{\theta}^*))^n) &\lesssim n^2L\left(\epsilon_{\rm opt} + \epsilon_{\rm LR}\right). \label{eq:LSVQC_ineq_LET_longT}
    \end{align}
\end{thm}
\begin{proof}
    The infidelity of a quantum state $\ket{\Psi_k}$ between $V$ and $U$ generally satisfies the following relation~\cite{gibbs2022long}: 
    \begin{align}
        &1 - \left| \bra{\Psi_k}(V^{\dag})^nU^n\ket{\Psi_k} \right|^2 \nonumber\\
        & \lesssim n^2 \left(1 - \left| \bra{\Psi_k}V^{\dag}U\ket{\Psi_k} \right|^2 \right),
    \end{align}
    where $n\in\mathbb{N}$. 
    Applying the above relation to the definition of the global cost function $C_{\rm LET}^{B_\mathcal{S}}$ (Eq.~\eqref{eq:cost_LET}), we obtain 
    \begin{align}
        C_{\rm LET}^{B_\mathcal{S}}(U_{n\tau}^{(L)},(V^{(L)}(\bm{\theta}^*))^n) \lesssim n^2 C_{\rm LET}^{B_\mathcal{S}}(U_{\tau}^{(L)},V^{(L)}(\bm{\theta}^*)) . \label{eq:cost_LET_n}
    \end{align}
    By combining Eqs.~\eqref{eq:cost_LET_n} with~\eqref{eq:LSVQC_ineq_LET}, we obtain Eq.~\eqref{eq:LSVQC_ineq_LET_longT}. 
    Then, Eq.~\eqref{eq:LSVQC_ineq_LLET_longT} is derived from Eqs.~\eqref{eq:LSVQC_ineq_LET_longT} using the relationship between the global cost and local cost given by Eq.~\eqref{eq:cost_ineq_general}. 
\end{proof}
Based on Eqs.~\eqref{eq:LSVQC_ineq_LLET_longT} and~\eqref{eq:LSVQC_ineq_LET_longT}, when we want to achieve $C_{\rm LLET,LET}^{B_\mathcal{S}}(U_{n\tau}^{(L)},(V^{(L)}(\bm{\theta}^*))^n)=\mathcal{O}(\epsilon)$ for some accuracy $\epsilon$, we need to choose a compilation size $\tilde{L}$ as
\begin{align}
    \tilde{L} &\gtrsim \mathcal{O}(\xi\log(n^2L^\alpha/\epsilon)) + r_H + v\tau + 2\max_{k}d_{W_k} \nonumber\\
    & + \max{\left( 2d_V+1, r_H+v\tau+2\max_{k}d_{W_k} \right)}. \label{eq:L_compilation_explicit_longT}
\end{align}
Equation~\eqref{eq:L_compilation_explicit_longT} indicates that the compilation size $\tilde{L}$ does not drastically increase as in increasing the simulation time, since it logarithmically increases as $\mathcal{O}(\log{n})$ for the time step $n=t/\tau$.

\section{\textit{Ab initio} downfolding}\label{append:DF}
In this Appendix, we describe details of the \textit{ab initio} downfolding method~\cite{downfolding_imada2010electronic}. 
As explained in Sec.~\ref{sec:DF}, the \textit{ab initio} downfolding is a technique to construct a low-energy effective lattice model for strongly correlated materials on the basis of maximally localized Wannier orbitals~\cite{Wannier_PhysRevB.56.12847}. 
Specifically, the effective lattice model is obtained in the following form: 
\begin{align}
    H &= \sum_{n,m}\sum_{\sigma=\uparrow,\downarrow}t_{nm}c_{n\sigma}^{\dag}c_{m\sigma} \nonumber\\
    &+ \frac{1}{2}\sum_{n,m}\sum_{\sigma,\sigma'}U_{nm}c_{n\sigma}^{\dag}c_{m\sigma'}^{\dag}c_{m\sigma'}c_{n\sigma}  \nonumber\\
    & + \frac{1}{2}\sum_{n,m}\sum_{\sigma,\sigma'}J_{nm}\Bigl(c_{n\sigma}^{\dag}c_{m\sigma'}^{\dag}c_{n\sigma'}c_{m\sigma}  + c_{n\sigma}^{\dag}c_{n\sigma'}^{\dag}c_{m\sigma'}c_{m\sigma} \Bigr), \label{eq:DF_hamiltonian}
\end{align}
where $n,m$ are the indices specifying the Wannier orbitals, $\sigma,\sigma'(=\uparrow,\downarrow)$ are the spin indices, and $c_{n\sigma}$ ($c_{n\sigma}^{\dag}$) is the annihilation (creation) operator of an electron with the spin-orbital index $n\sigma$. 
$t_{nm}$ is the one-body integral defined as
\begin{align}
    t_{nm} &= \int d\bm{r} \psi_{n}^{*}(\bm{r}) H_0 \psi_{m}(\bm{r}), \label{eq:DF_hopping}
\end{align}
where $\psi_{n}(\bm{r})$ is the $n$-th Wannier orbital and $H_0$ is the noninteracting Kohn-Sham Hamiltonian obtained by the DFT. 
The integral in Eq.~\eqref{eq:DF_hopping} is taken over the crystal volume. 
$U_{nm}$ and $J_{nm}$ are the effective Coulomb and effective exchange interactions defined as, 
\begin{align}
    U_{nm}
    &= \iint d\bm{r}d\bm{r}' |\psi_{n}(\bm{r})|^2 W(\bm{r},\bm{r}') |\psi_{m}(\bm{r}')|^2 , \label{eq:DF_U} \\
    J_{nm} 
    &=  \iint d\bm{r}d\bm{r}'\psi_{n}^*(\bm{r})\psi_{m}(\bm{r}) W(\bm{r},\bm{r}') \psi_{m}(\bm{r}')^*\psi_{n}(\bm{r}') \label{eq:DF_J}, 
\end{align}
where $W(\bm{r},\bm{r}')$ is the screened Coulomb interaction obtained based on the constrained random phase approximation~\cite{cRPA_PhysRevB.70.195104}. 
The parameters $t_{nm}$, $U_{nm}$, and $J_{nm}$ are obtained by using the RESPACK package~\cite{RESPACK1_PhysRevB.93.085124,RESPACK2_nakamura2009ab,RESPACK3_nakamura2008ab,REPACK4_PhysRevB.79.195110,RESPACK5_fujiwara2003generalization,RESPACK6_nakamura2021respack}.  



\section{Trotter error bound}\label{append:trotter}
Here, we review details of the Trotter error bound used in the resource estimation in Sec.~\ref{sec:resource}. 

\subsection{Worst-case error bound}
Using the results of Ref.~\cite{TrotError_PhysRevX.11.011020}, the Trotter error bound in the worst-case scenario for the Hamiltonian $H=\sum_{\gamma=1}^{\Gamma}H_\gamma$ is expressed as 
\begin{align}
    &\| (U_{1}(t/r))^r - e^{-it H} \| \leq \mathcal{O}(\alpha_{\rm comm} t^2/r),
    \label{eq:trotter_bound_worst}
\end{align}
where $U_{1}(t)=\prod_{\gamma=1}^{\Gamma}e^{-itH_\gamma}$ denote the first-order Trotter decomposition of the time evolution operator $e^{-itH}$, and $\alpha_{\rm comm}$ is the commutator scaling factor given by
\begin{align}
    \alpha_{\rm comm} &= \sum_{\gamma_1,\cdots,\gamma_{p+1}}^{\Gamma}\left\| [H_{\gamma_1}, \cdots [H_{\gamma_p}, H_{\gamma_{p+1}}] \cdots] \right\|. \label{eq:alpha_comm}
\end{align}
Here, $\| \cdots \|$ denotes the spectral norm. 
From Eq.~\eqref{eq:trotter_bound_worst}, the number of Trotter steps $r$ needed to achieve the accuracy $\epsilon$ in the worst-case scenario can be estimated as 
\begin{align}
    r &= \mathcal{O}(\alpha_{\rm comm}t^{2}/\epsilon).
    \label{eq:trotter_r_worst}
\end{align}
For the Fermi-Hubbard model, the coefficients $\alpha_{\rm comm}$ asymptotically scales as $\alpha_{\rm comm}=\mathcal{O}(L)$~\cite{TrotFH_PhysRevB.108.195105} with $L$ being the lattice size. 
We can expect the same scaling $\alpha_{\rm comm}=\mathcal{O}(L)$ for the \textit{ab initio} downfolding models owing to the similarity of the conventional Fermi-Hubbard model and the extended Fermi-Hubbard model obtained by the \textit{ab initio} downfolding. 

\subsection{Average-case error bound}
The average-case error captures the averaged error for some input states chosen randomly from an ensemble of quantum states $\mathcal{E}=\{(p_i,\psi_i)\}$, where $p_i$ is the probability obtaining the state $\psi_i$. 
Formally, it is defined as
\begin{align}
    R(U_{1}(t),e^{-itH}) = \mathbbm{E}_{\mathcal{E}}\left[ \mathcal{D}(U_1(t)\ket{\psi}, e^{-itH}\ket{\psi}) \right],  \label{eq:def_average_error}
\end{align}
where $\mathcal{D}(*,*)$ and $\mathbbm{E}_{\mathcal{E}}[\cdots]$ denote the distance measure and the expectation over the ensemble $\mathcal{E}$, respectively. 
From the results of Ref.~\cite{TrotErrorAve_PhysRevLett.129.270502}, the average-case Trotter error bound (for the 1-design ensemble) is obtained as 
\begin{align}
    R((U_{1}(t/r))^r,e^{-itH}) \leq \mathcal{O}(T_1 t^2/r),  \label{eq:trotter_bound_average}
\end{align}
where 
\begin{align}
    T_1 &= \sum_{\gamma_1,\cdots,\gamma_{p+1}}^{\Gamma}\frac{1}{\sqrt{d}} \left\| [H_{\gamma_1}, \cdots [H_{\gamma_p}, H_{\gamma_{p+1}}] \cdots] \right\|_{F}.  \label{eq:T1}
\end{align}
Here, $\| \cdots \|_F$ denotes the Frobenius norm and $d=2^n$ is the dimension of the Hilbert space of an $n$-qubit system. 
Thus, the Trotter step $r$ needed to achieve the accuracy $\epsilon$ in terms of the average-case error is estimated as 
\begin{align}
     r &= \mathcal{O}\left(T_1t^{2}/\epsilon\right). 
    \label{eq:trotter_r_average}
\end{align}
For an $n$-qubit nearest-neighbor Hamiltonian such as the Fermi-Hubbard model, the coefficient $T_1$ asymptotically scales as $T_1=\mathcal{O}(\sqrt{n})$~\cite{TrotErrorAve_PhysRevLett.129.270502}. 
Therefore, we expect $T_1=\mathcal{O}(\sqrt{L})$ for the $L$-site Fermi-Hubbard and \textit{ab initio} downfolding models in the main text. 

\section{Details of resource estimation}\label{append:resource}
In this Appendix, we summarize the details of the resource estimation in Sec.~\ref{sec:resource}. 

\subsection{Fermi-Hubbard model} \label{append:resource-FH}
In Sec.~\ref{sec:resource}, we estimate the gate count needed to simulate the following Fermi-Hubbard model at half-filling: 
\begin{align}
    H =& -t_{\rm hop}\sum_{\braket{i,j}}\sum_{\sigma=\uparrow,\downarrow}(c_i^{\dag}c_j+c_j^{\dag}c_i) \nonumber\\
    &+U\sum_{i}n_{i\uparrow}n_{i,\downarrow} -\frac{U}{2}\sum_{i,\sigma}n_{i\sigma}, 
    \label{eq:hamil_hubbard}
\end{align}
where $\braket{i,j}$ denotes neighboring sites, $t_{\rm hop}$ is the hopping integral, $U$ is the on-site Coulomb interaction, and the open boundary condition is adopted. 
For the $L$-site 2D Fermi-Hubbard model on a square lattice, the number of single-qubit gates $N_{1q}$ and two-qubit $i$SWAP gates $N_{2q}$ needed to implement the Trotterization circuit on superconducting qubit devices is formally estimated as~\cite{FH_resource_PhysRevApplied.14.014059}
\begin{align}
    N_{1q} &= 3Lr^{\rm trot}_{\epsilon}, \label{eq:N1q_trot}\\
    N_{2q} &= \left(4L^{3/2} + 2L - 2\sqrt{L}\right)r^{\rm trot}_{\epsilon}. \label{eq:N2q_trot}
\end{align}
Here, the gate count is obtained by adopting the Jordan-Wigner transformation and FSWAP network~\cite{fswap_PhysRevA.79.032316}, which is utilized to decompose nonlocal multi-qubit gates into nearest-neighbor coupling two-qubit gates. 
The number of Trotter steps $r^{\rm trot}_{\epsilon}$ is estimated using Eq.~\eqref{eq:r_trot_hubbard}. 
The gate count for the LSVQC using the VHA with the depth $N_L$ is obtained by replacing $r^{\rm trot}_{\epsilon}$ with $N_L$ in Eqs.~\eqref{eq:N1q_trot} and \eqref{eq:N2q_trot}. 
Here, $N_L$ is empirically estimated using Eq.~\eqref{eq:NL_estimate_LSVQC} with $R_{\epsilon}^{\rm LSVQC}=10$. 

The analog rotation gate count, which is used in the resource estimation for the STAR architecture, is obtained by rewriting the fermionic Hamiltonian~\eqref{eq:hamil_hubbard} into the qubit representation using the Jordan-Wigner transformation (Eqs.~\eqref{eq:JW_annihilation} and~\eqref{eq:JW_creation}) as follows: 
\begin{align}
    H =& -\frac{t_{\rm hop}}{2}\sum_{\braket{i,j},\sigma}\left(X_{i_\sigma}X_{j_\sigma}+Y_{i_\sigma}Y_{j_\sigma}\right)\overleftrightarrow{Z}_{i_\sigma,j_\sigma} \nonumber\\
    &+ \frac{U}{4}\sum_{i}Z_{i_\uparrow}Z_{i_\downarrow},\label{eq:hamil_hubbard_JW}
\end{align}
where
\begin{align}
    \overleftrightarrow{Z}_{i_\sigma,j_\sigma} \equiv 
    \begin{cases}
        Z_{i+1_\sigma}Z_{i+2_\sigma}\cdots Z_{j-1_\sigma} & (i<j) \\
        Z_{j+1_\sigma}Z_{j+2_\sigma}\cdots Z_{i-1_\sigma} & (i>j)
    \end{cases}.
\end{align}
From Eq.~\eqref{eq:hamil_hubbard_JW}, the number of analog rotation gates (i.e., the number of terms in the Hamiltonian~\eqref{eq:hamil_hubbard_JW}) is counted as $9L - 8\sqrt{L}$. 
Hence, the total number of the analog rotation gates in the Trotterization circuit is given by
\begin{align}
    N_{R_Z} = (9L - 8\sqrt{L})r^{\rm trot}_{\epsilon}.
    \label{eq:Rz_count_hubbard}
\end{align}
The analog rotation gate count for the LSVQC using the VHA is obtained by replacing $r^{\rm trot}_{\epsilon}$ with $N_L$ (estimated by Eq.~\eqref{eq:NL_estimate_LSVQC}) in Eq.~\eqref{eq:Rz_count_hubbard}. 

\subsection{\textit{Ab initio} downfolding models} \label{append:resource-DF}
The resource estimation for the \textit{ab initio} downfolding models is performed based on the results in Ref.~\cite{DFresource_PhysRevA.106.012612}. 
Specifically, in the Ref.~\cite{DFresource_PhysRevA.106.012612}, the authors estimated the gate count $N_{\rm gate}$ as a function of the number of unit cells $N_{\rm cells}$ as 
\begin{align}
    N_{\rm gate} &= \left( \sum_{i_{\rm  int}}N_{\rm terms/cell}^{i_{\rm int}}N_{\rm gates}^{i_{\rm int}} \right)N_{\rm cells} \nonumber\\
    & + \frac{N_{\rm qubits/cell}^2N_{\rm cells}^2 - 2N_{\rm qubits/cell}N_{\rm cells}}{2}N_{\rm gates}^{\rm fswap}, 
    \label{eq:Ngate_DF}
\end{align}
where the fermionic modes are encoded to qubits using the Jordan-Wigner transformation and the FSWAP network is applied to deal with the nonlocal Pauli string. 
$i_{\rm int}$ is an index specifying the type of the interaction (i.e., hopping interaction, Coulomb interaction, and exchange interaction), $N_{\rm gates}^{i_{\rm int}}$ is the number of gates needed to implement the interaction operator specified by the index $i_{\rm int}$, and $N_{\rm terms/cell}^{i_{\rm int}}$ is the number of interaction terms specified by the index $i_{\rm int}$ in the Hamiltonian. 
$N_{\rm gate}^{\rm fswap}$ is the gate count required for single FSWAP operation. 
$N_{\rm qubits/cell}$ is the number of qubits per unit cell. 
Based on Eq.~\eqref{eq:Ngate_DF}, we estimate the total number of gates needed to perform dynamics simulation as $N_{\rm gate}N_L$ ($N_{\rm gate}r$) using the LSVQC (Trotterization) with the circuit depth $N_L$ ($r$). 
The values of $N_{\rm gates}^{i_{\rm int}}$, $N_{\rm terms/cell}^{i_{\rm int}}$, $N_{\rm gate}^{\rm fswap}$, and $N_{\rm qubits/cell}$ are taken from the Ref~\cite{DFresource_PhysRevA.106.012612}. 
Specifically, the two-qubit (CNOT) gate count for the NISQ devices in Table~\ref{tab:DF_resource} is obtained based on Eq.~\eqref{eq:Ngate_DF}. 

On the other hand, the analog rotation gate count $N_{R_Z}$ for the STAR architecture is obtained by counting the number of terms in a given \textit{ab initio} Hamiltonian in the same way as that for the Fermi-Hubbard model. 
Specifically, it is estimated as $N_{R_Z}=rN_{\rm terms/cell}^{i_{\rm int}} N_{\rm cells}$ ($N_{R_Z}=N_L N_{\rm terms/cell}^{i_{\rm int}}N_{\rm cells}$) for the Trotterization (LSVQC) with the depth $r$ ($N_L$). 

\bibliography{main}

\begin{thebibliography}{108}%
\makeatletter
\providecommand \@ifxundefined [1]{%
 \@ifx{#1\undefined}
}%
\providecommand \@ifnum [1]{%
 \ifnum #1\expandafter \@firstoftwo
 \else \expandafter \@secondoftwo
 \fi
}%
\providecommand \@ifx [1]{%
 \ifx #1\expandafter \@firstoftwo
 \else \expandafter \@secondoftwo
 \fi
}%
\providecommand \natexlab [1]{#1}%
\providecommand \enquote  [1]{``#1''}%
\providecommand \bibnamefont  [1]{#1}%
\providecommand \bibfnamefont [1]{#1}%
\providecommand \citenamefont [1]{#1}%
\providecommand \href@noop [0]{\@secondoftwo}%
\providecommand \href [0]{\begingroup \@sanitize@url \@href}%
\providecommand \@href[1]{\@@startlink{#1}\@@href}%
\providecommand \@@href[1]{\endgroup#1\@@endlink}%
\providecommand \@sanitize@url [0]{\catcode `\\12\catcode `\$12\catcode `\&12\catcode `\#12\catcode `\^12\catcode `\_12\catcode `\%12\relax}%
\providecommand \@@startlink[1]{}%
\providecommand \@@endlink[0]{}%
\providecommand \url  [0]{\begingroup\@sanitize@url \@url }%
\providecommand \@url [1]{\endgroup\@href {#1}{\urlprefix }}%
\providecommand \urlprefix  [0]{URL }%
\providecommand \Eprint [0]{\href }%
\providecommand \doibase [0]{https://doi.org/}%
\providecommand \selectlanguage [0]{\@gobble}%
\providecommand \bibinfo  [0]{\@secondoftwo}%
\providecommand \bibfield  [0]{\@secondoftwo}%
\providecommand \translation [1]{[#1]}%
\providecommand \BibitemOpen [0]{}%
\providecommand \bibitemStop [0]{}%
\providecommand \bibitemNoStop [0]{.\EOS\space}%
\providecommand \EOS [0]{\spacefactor3000\relax}%
\providecommand \BibitemShut  [1]{\csname bibitem#1\endcsname}%
\let\auto@bib@innerbib\@empty
\bibitem [{\citenamefont {Kim}\ \emph {et~al.}(2023)\citenamefont {Kim}, \citenamefont {Eddins}, \citenamefont {Anand}, \citenamefont {Wei}, \citenamefont {Van Den~Berg}, \citenamefont {Rosenblatt}, \citenamefont {Nayfeh}, \citenamefont {Wu}, \citenamefont {Zaletel}, \citenamefont {Temme} \emph {et~al.}}]{kim2023evidence}%
  \BibitemOpen
  \bibfield  {author} {\bibinfo {author} {\bibfnamefont {Y.}~\bibnamefont {Kim}}, \bibinfo {author} {\bibfnamefont {A.}~\bibnamefont {Eddins}}, \bibinfo {author} {\bibfnamefont {S.}~\bibnamefont {Anand}}, \bibinfo {author} {\bibfnamefont {K.~X.}\ \bibnamefont {Wei}}, \bibinfo {author} {\bibfnamefont {E.}~\bibnamefont {Van Den~Berg}}, \bibinfo {author} {\bibfnamefont {S.}~\bibnamefont {Rosenblatt}}, \bibinfo {author} {\bibfnamefont {H.}~\bibnamefont {Nayfeh}}, \bibinfo {author} {\bibfnamefont {Y.}~\bibnamefont {Wu}}, \bibinfo {author} {\bibfnamefont {M.}~\bibnamefont {Zaletel}}, \bibinfo {author} {\bibfnamefont {K.}~\bibnamefont {Temme}}, \emph {et~al.},\ }\bibfield  {title} {\bibinfo {title} {Evidence for the utility of quantum computing before fault tolerance},\ }\href {https://www.nature.com/articles/s41586-023-06096-3} {\bibfield  {journal} {\bibinfo  {journal} {Nature}\ }\textbf {\bibinfo {volume} {618}},\ \bibinfo {pages} {500} (\bibinfo {year} {2023})}\BibitemShut {NoStop}%
\bibitem [{\citenamefont {Kitaev}(1995)}]{kitaev1995quantum}%
  \BibitemOpen
  \bibfield  {author} {\bibinfo {author} {\bibfnamefont {A.~Y.}\ \bibnamefont {Kitaev}},\ }\bibfield  {title} {\bibinfo {title} {{Quantum measurements and the Abelian stabilizer problem}},\ }\href {https://arxiv.org/abs/quant-ph/9511026} {\bibfield  {journal} {\bibinfo  {journal} {arXiv preprint quant-ph/9511026}\ } (\bibinfo {year} {1995})}\BibitemShut {NoStop}%
\bibitem [{\citenamefont {Cleve}\ \emph {et~al.}(1998)\citenamefont {Cleve}, \citenamefont {Ekert}, \citenamefont {Macchiavello},\ and\ \citenamefont {Mosca}}]{cleve1998quantum}%
  \BibitemOpen
  \bibfield  {author} {\bibinfo {author} {\bibfnamefont {R.}~\bibnamefont {Cleve}}, \bibinfo {author} {\bibfnamefont {A.}~\bibnamefont {Ekert}}, \bibinfo {author} {\bibfnamefont {C.}~\bibnamefont {Macchiavello}},\ and\ \bibinfo {author} {\bibfnamefont {M.}~\bibnamefont {Mosca}},\ }\bibfield  {title} {\bibinfo {title} {{Quantum algorithms revisited}},\ }\href {https://royalsocietypublishing.org/doi/10.1098/rspa.1998.0164} {\bibfield  {journal} {\bibinfo  {journal} {Proceedings of the Royal Society of London. Series A: Mathematical, Physical and Engineering Sciences}\ }\textbf {\bibinfo {volume} {454}},\ \bibinfo {pages} {339} (\bibinfo {year} {1998})}\BibitemShut {NoStop}%
\bibitem [{\citenamefont {Nielsen}\ and\ \citenamefont {Chuang}(2002)}]{nielsen-chuang}%
  \BibitemOpen
  \bibfield  {author} {\bibinfo {author} {\bibfnamefont {M.~A.}\ \bibnamefont {Nielsen}}\ and\ \bibinfo {author} {\bibfnamefont {I.}~\bibnamefont {Chuang}},\ }\href@noop {} {\emph {\bibinfo {title} {{Quantum computation and quantum information}}}}\ (\bibinfo  {publisher} {Cambridge University Press},\ \bibinfo {year} {2002})\BibitemShut {NoStop}%
\bibitem [{\citenamefont {Lloyd}(1996)}]{Lloyd1996universal}%
  \BibitemOpen
  \bibfield  {author} {\bibinfo {author} {\bibfnamefont {S.}~\bibnamefont {Lloyd}},\ }\bibfield  {title} {\bibinfo {title} {{Universal quantum simulators}},\ }\href {https://www.science.org/doi/10.1126/science.273.5278.1073} {\bibfield  {journal} {\bibinfo  {journal} {Science}\ }\textbf {\bibinfo {volume} {273}},\ \bibinfo {pages} {1073} (\bibinfo {year} {1996})}\BibitemShut {NoStop}%
\bibitem [{\citenamefont {Preskill}(2018)}]{preskill2018quantum}%
  \BibitemOpen
  \bibfield  {author} {\bibinfo {author} {\bibfnamefont {J.}~\bibnamefont {Preskill}},\ }\bibfield  {title} {\bibinfo {title} {{Quantum computing in the {NISQ} era and beyond}},\ }\href {https://quantum-journal.org/papers/q-2018-08-06-79/} {\bibfield  {journal} {\bibinfo  {journal} {Quantum}\ }\textbf {\bibinfo {volume} {2}},\ \bibinfo {pages} {79} (\bibinfo {year} {2018})}\BibitemShut {NoStop}%
\bibitem [{\citenamefont {Low}\ and\ \citenamefont {Chuang}(2019)}]{low2019hamiltonian}%
  \BibitemOpen
  \bibfield  {author} {\bibinfo {author} {\bibfnamefont {G.~H.}\ \bibnamefont {Low}}\ and\ \bibinfo {author} {\bibfnamefont {I.~L.}\ \bibnamefont {Chuang}},\ }\bibfield  {title} {\bibinfo {title} {Hamiltonian simulation by qubitization},\ }\href {https://quantum-journal.org/papers/q-2019-07-12-163/} {\bibfield  {journal} {\bibinfo  {journal} {Quantum}\ }\textbf {\bibinfo {volume} {3}},\ \bibinfo {pages} {163} (\bibinfo {year} {2019})}\BibitemShut {NoStop}%
\bibitem [{\citenamefont {Li}\ and\ \citenamefont {Benjamin}(2017)}]{VQS1_PhysRevX.7.021050}%
  \BibitemOpen
  \bibfield  {author} {\bibinfo {author} {\bibfnamefont {Y.}~\bibnamefont {Li}}\ and\ \bibinfo {author} {\bibfnamefont {S.~C.}\ \bibnamefont {Benjamin}},\ }\bibfield  {title} {\bibinfo {title} {{Efficient Variational Quantum Simulator Incorporating Active Error Minimization}},\ }\href {https://doi.org/10.1103/PhysRevX.7.021050} {\bibfield  {journal} {\bibinfo  {journal} {Phys. Rev. X}\ }\textbf {\bibinfo {volume} {7}},\ \bibinfo {pages} {021050} (\bibinfo {year} {2017})}\BibitemShut {NoStop}%
\bibitem [{\citenamefont {Yuan}\ \emph {et~al.}(2019)\citenamefont {Yuan}, \citenamefont {Endo}, \citenamefont {Zhao}, \citenamefont {Li},\ and\ \citenamefont {Benjamin}}]{VQS2_yuan2019theory}%
  \BibitemOpen
  \bibfield  {author} {\bibinfo {author} {\bibfnamefont {X.}~\bibnamefont {Yuan}}, \bibinfo {author} {\bibfnamefont {S.}~\bibnamefont {Endo}}, \bibinfo {author} {\bibfnamefont {Q.}~\bibnamefont {Zhao}}, \bibinfo {author} {\bibfnamefont {Y.}~\bibnamefont {Li}},\ and\ \bibinfo {author} {\bibfnamefont {S.~C.}\ \bibnamefont {Benjamin}},\ }\bibfield  {title} {\bibinfo {title} {Theory of variational quantum simulation},\ }\href {https://quantum-journal.org/papers/q-2019-10-07-191/?ref=https://githubhelp.com} {\bibfield  {journal} {\bibinfo  {journal} {Quantum}\ }\textbf {\bibinfo {volume} {3}},\ \bibinfo {pages} {191} (\bibinfo {year} {2019})}\BibitemShut {NoStop}%
\bibitem [{\citenamefont {Endo}\ \emph {et~al.}(2020)\citenamefont {Endo}, \citenamefont {Sun}, \citenamefont {Li}, \citenamefont {Benjamin},\ and\ \citenamefont {Yuan}}]{VQS3_PhysRevLett.125.010501}%
  \BibitemOpen
  \bibfield  {author} {\bibinfo {author} {\bibfnamefont {S.}~\bibnamefont {Endo}}, \bibinfo {author} {\bibfnamefont {J.}~\bibnamefont {Sun}}, \bibinfo {author} {\bibfnamefont {Y.}~\bibnamefont {Li}}, \bibinfo {author} {\bibfnamefont {S.~C.}\ \bibnamefont {Benjamin}},\ and\ \bibinfo {author} {\bibfnamefont {X.}~\bibnamefont {Yuan}},\ }\bibfield  {title} {\bibinfo {title} {{Variational Quantum Simulation of General Processes}},\ }\href {https://doi.org/10.1103/PhysRevLett.125.010501} {\bibfield  {journal} {\bibinfo  {journal} {Phys. Rev. Lett.}\ }\textbf {\bibinfo {volume} {125}},\ \bibinfo {pages} {010501} (\bibinfo {year} {2020})}\BibitemShut {NoStop}%
\bibitem [{\citenamefont {Yao}\ \emph {et~al.}(2021)\citenamefont {Yao}, \citenamefont {Gomes}, \citenamefont {Zhang}, \citenamefont {Wang}, \citenamefont {Ho}, \citenamefont {Iadecola},\ and\ \citenamefont {Orth}}]{VQS4_PRXQuantum.2.030307}%
  \BibitemOpen
  \bibfield  {author} {\bibinfo {author} {\bibfnamefont {Y.-X.}\ \bibnamefont {Yao}}, \bibinfo {author} {\bibfnamefont {N.}~\bibnamefont {Gomes}}, \bibinfo {author} {\bibfnamefont {F.}~\bibnamefont {Zhang}}, \bibinfo {author} {\bibfnamefont {C.-Z.}\ \bibnamefont {Wang}}, \bibinfo {author} {\bibfnamefont {K.-M.}\ \bibnamefont {Ho}}, \bibinfo {author} {\bibfnamefont {T.}~\bibnamefont {Iadecola}},\ and\ \bibinfo {author} {\bibfnamefont {P.~P.}\ \bibnamefont {Orth}},\ }\bibfield  {title} {\bibinfo {title} {{Adaptive Variational Quantum Dynamics Simulations}},\ }\href {https://doi.org/10.1103/PRXQuantum.2.030307} {\bibfield  {journal} {\bibinfo  {journal} {PRX Quantum}\ }\textbf {\bibinfo {volume} {2}},\ \bibinfo {pages} {030307} (\bibinfo {year} {2021})}\BibitemShut {NoStop}%
\bibitem [{\citenamefont {Barison}\ \emph {et~al.}(2021)\citenamefont {Barison}, \citenamefont {Vicentini},\ and\ \citenamefont {Carleo}}]{VQS5_barison2021efficient}%
  \BibitemOpen
  \bibfield  {author} {\bibinfo {author} {\bibfnamefont {S.}~\bibnamefont {Barison}}, \bibinfo {author} {\bibfnamefont {F.}~\bibnamefont {Vicentini}},\ and\ \bibinfo {author} {\bibfnamefont {G.}~\bibnamefont {Carleo}},\ }\bibfield  {title} {\bibinfo {title} {An efficient quantum algorithm for the time evolution of parameterized circuits},\ }\href {https://quantum-journal.org/papers/q-2021-07-28-512/} {\bibfield  {journal} {\bibinfo  {journal} {Quantum}\ }\textbf {\bibinfo {volume} {5}},\ \bibinfo {pages} {512} (\bibinfo {year} {2021})}\BibitemShut {NoStop}%
\bibitem [{\citenamefont {Khatri}\ \emph {et~al.}(2019)\citenamefont {Khatri}, \citenamefont {LaRose}, \citenamefont {Poremba}, \citenamefont {Cincio}, \citenamefont {Sornborger},\ and\ \citenamefont {Coles}}]{Khatri-QAQC}%
  \BibitemOpen
  \bibfield  {author} {\bibinfo {author} {\bibfnamefont {S.}~\bibnamefont {Khatri}}, \bibinfo {author} {\bibfnamefont {R.}~\bibnamefont {LaRose}}, \bibinfo {author} {\bibfnamefont {A.}~\bibnamefont {Poremba}}, \bibinfo {author} {\bibfnamefont {L.}~\bibnamefont {Cincio}}, \bibinfo {author} {\bibfnamefont {A.~T.}\ \bibnamefont {Sornborger}},\ and\ \bibinfo {author} {\bibfnamefont {P.~J.}\ \bibnamefont {Coles}},\ }\bibfield  {title} {\bibinfo {title} {{Quantum-assisted quantum compiling}},\ }\href {https://quantum-journal.org/papers/q-2019-05-13-140/} {\bibfield  {journal} {\bibinfo  {journal} {Quantum}\ }\textbf {\bibinfo {volume} {3}},\ \bibinfo {pages} {140} (\bibinfo {year} {2019})}\BibitemShut {NoStop}%
\bibitem [{\citenamefont {Sharma}\ \emph {et~al.}(2020)\citenamefont {Sharma}, \citenamefont {Khatri}, \citenamefont {Cerezo},\ and\ \citenamefont {Coles}}]{sharma2020noise}%
  \BibitemOpen
  \bibfield  {author} {\bibinfo {author} {\bibfnamefont {K.}~\bibnamefont {Sharma}}, \bibinfo {author} {\bibfnamefont {S.}~\bibnamefont {Khatri}}, \bibinfo {author} {\bibfnamefont {M.}~\bibnamefont {Cerezo}},\ and\ \bibinfo {author} {\bibfnamefont {P.~J.}\ \bibnamefont {Coles}},\ }\bibfield  {title} {\bibinfo {title} {{Noise resilience of variational quantum compiling}},\ }\href {https://iopscience.iop.org/article/10.1088/1367-2630/ab784c} {\bibfield  {journal} {\bibinfo  {journal} {New Journal of Physics}\ }\textbf {\bibinfo {volume} {22}},\ \bibinfo {pages} {043006} (\bibinfo {year} {2020})}\BibitemShut {NoStop}%
\bibitem [{\citenamefont {Bilek}\ and\ \citenamefont {Wold}(2022)}]{bilek2022recursive}%
  \BibitemOpen
  \bibfield  {author} {\bibinfo {author} {\bibfnamefont {S.}~\bibnamefont {Bilek}}\ and\ \bibinfo {author} {\bibfnamefont {K.}~\bibnamefont {Wold}},\ }\bibfield  {title} {\bibinfo {title} {Recursive variational quantum compiling},\ }\href {https://arxiv.org/abs/2203.08514} {\bibfield  {journal} {\bibinfo  {journal} {arXiv preprint arXiv:2203.08514}\ } (\bibinfo {year} {2022})}\BibitemShut {NoStop}%
\bibitem [{\citenamefont {Cirstoiu}\ \emph {et~al.}(2020)\citenamefont {Cirstoiu}, \citenamefont {Holmes}, \citenamefont {Iosue}, \citenamefont {Cincio}, \citenamefont {Coles},\ and\ \citenamefont {Sornborger}}]{cirstoiu2020variational}%
  \BibitemOpen
  \bibfield  {author} {\bibinfo {author} {\bibfnamefont {C.}~\bibnamefont {Cirstoiu}}, \bibinfo {author} {\bibfnamefont {Z.}~\bibnamefont {Holmes}}, \bibinfo {author} {\bibfnamefont {J.}~\bibnamefont {Iosue}}, \bibinfo {author} {\bibfnamefont {L.}~\bibnamefont {Cincio}}, \bibinfo {author} {\bibfnamefont {P.~J.}\ \bibnamefont {Coles}},\ and\ \bibinfo {author} {\bibfnamefont {A.}~\bibnamefont {Sornborger}},\ }\bibfield  {title} {\bibinfo {title} {Variational fast forwarding for quantum simulation beyond the coherence time},\ }\href {https://www.nature.com/articles/s41534-020-00302-0} {\bibfield  {journal} {\bibinfo  {journal} {npj Quantum Information}\ }\textbf {\bibinfo {volume} {6}},\ \bibinfo {pages} {82} (\bibinfo {year} {2020})}\BibitemShut {NoStop}%
\bibitem [{\citenamefont {Gibbs}\ \emph {et~al.}(2022)\citenamefont {Gibbs}, \citenamefont {Gili}, \citenamefont {Holmes}, \citenamefont {Commeau}, \citenamefont {Arrasmith}, \citenamefont {Cincio}, \citenamefont {Coles},\ and\ \citenamefont {Sornborger}}]{gibbs2022long}%
  \BibitemOpen
  \bibfield  {author} {\bibinfo {author} {\bibfnamefont {J.}~\bibnamefont {Gibbs}}, \bibinfo {author} {\bibfnamefont {K.}~\bibnamefont {Gili}}, \bibinfo {author} {\bibfnamefont {Z.}~\bibnamefont {Holmes}}, \bibinfo {author} {\bibfnamefont {B.}~\bibnamefont {Commeau}}, \bibinfo {author} {\bibfnamefont {A.}~\bibnamefont {Arrasmith}}, \bibinfo {author} {\bibfnamefont {L.}~\bibnamefont {Cincio}}, \bibinfo {author} {\bibfnamefont {P.~J.}\ \bibnamefont {Coles}},\ and\ \bibinfo {author} {\bibfnamefont {A.}~\bibnamefont {Sornborger}},\ }\bibfield  {title} {\bibinfo {title} {Long-time simulations for fixed input states on quantum hardware},\ }\href {https://www.nature.com/articles/s41534-022-00625-0} {\bibfield  {journal} {\bibinfo  {journal} {npj Quantum Information}\ }\textbf {\bibinfo {volume} {8}},\ \bibinfo {pages} {135} (\bibinfo {year} {2022})}\BibitemShut {NoStop}%
\bibitem [{\citenamefont {Cerezo}\ \emph {et~al.}(2021)\citenamefont {Cerezo}, \citenamefont {Sone}, \citenamefont {Volkoff}, \citenamefont {Cincio},\ and\ \citenamefont {Coles}}]{BP_cerezo2021cost}%
  \BibitemOpen
  \bibfield  {author} {\bibinfo {author} {\bibfnamefont {M.}~\bibnamefont {Cerezo}}, \bibinfo {author} {\bibfnamefont {A.}~\bibnamefont {Sone}}, \bibinfo {author} {\bibfnamefont {T.}~\bibnamefont {Volkoff}}, \bibinfo {author} {\bibfnamefont {L.}~\bibnamefont {Cincio}},\ and\ \bibinfo {author} {\bibfnamefont {P.~J.}\ \bibnamefont {Coles}},\ }\bibfield  {title} {\bibinfo {title} {Cost function dependent barren plateaus in shallow parametrized quantum circuits},\ }\href {https://www.nature.com/articles/s41467-021-21728-w} {\bibfield  {journal} {\bibinfo  {journal} {Nature communications}\ }\textbf {\bibinfo {volume} {12}},\ \bibinfo {pages} {1791} (\bibinfo {year} {2021})}\BibitemShut {NoStop}%
\bibitem [{\citenamefont {McClean}\ \emph {et~al.}(2018)\citenamefont {McClean}, \citenamefont {Boixo}, \citenamefont {Smelyanskiy}, \citenamefont {Babbush},\ and\ \citenamefont {Neven}}]{BP_mcclean2018barren}%
  \BibitemOpen
  \bibfield  {author} {\bibinfo {author} {\bibfnamefont {J.~R.}\ \bibnamefont {McClean}}, \bibinfo {author} {\bibfnamefont {S.}~\bibnamefont {Boixo}}, \bibinfo {author} {\bibfnamefont {V.~N.}\ \bibnamefont {Smelyanskiy}}, \bibinfo {author} {\bibfnamefont {R.}~\bibnamefont {Babbush}},\ and\ \bibinfo {author} {\bibfnamefont {H.}~\bibnamefont {Neven}},\ }\bibfield  {title} {\bibinfo {title} {Barren plateaus in quantum neural network training landscapes},\ }\href {https://www.nature.com/articles/s41467-018-07090-4} {\bibfield  {journal} {\bibinfo  {journal} {Nature communications}\ }\textbf {\bibinfo {volume} {9}},\ \bibinfo {pages} {4812} (\bibinfo {year} {2018})}\BibitemShut {NoStop}%
\bibitem [{\citenamefont {Suzuki}\ \emph {et~al.}(2022)\citenamefont {Suzuki}, \citenamefont {Endo}, \citenamefont {Fujii},\ and\ \citenamefont {Tokunaga}}]{EFTQC0_PRXQuantum.3.010345}%
  \BibitemOpen
  \bibfield  {author} {\bibinfo {author} {\bibfnamefont {Y.}~\bibnamefont {Suzuki}}, \bibinfo {author} {\bibfnamefont {S.}~\bibnamefont {Endo}}, \bibinfo {author} {\bibfnamefont {K.}~\bibnamefont {Fujii}},\ and\ \bibinfo {author} {\bibfnamefont {Y.}~\bibnamefont {Tokunaga}},\ }\bibfield  {title} {\bibinfo {title} {{Quantum Error Mitigation as a Universal Error Reduction Technique: Applications from the NISQ to the Fault-Tolerant Quantum Computing Eras}},\ }\href {https://doi.org/10.1103/PRXQuantum.3.010345} {\bibfield  {journal} {\bibinfo  {journal} {PRX Quantum}\ }\textbf {\bibinfo {volume} {3}},\ \bibinfo {pages} {010345} (\bibinfo {year} {2022})}\BibitemShut {NoStop}%
\bibitem [{\citenamefont {Campbell}(2021)}]{EFTQC1_campbell2021early}%
  \BibitemOpen
  \bibfield  {author} {\bibinfo {author} {\bibfnamefont {E.~T.}\ \bibnamefont {Campbell}},\ }\bibfield  {title} {\bibinfo {title} {Early fault-tolerant simulations of the hubbard model},\ }\href {https://iopscience.iop.org/article/10.1088/2058-9565/ac3110/meta} {\bibfield  {journal} {\bibinfo  {journal} {Quantum Science and Technology}\ }\textbf {\bibinfo {volume} {7}},\ \bibinfo {pages} {015007} (\bibinfo {year} {2021})}\BibitemShut {NoStop}%
\bibitem [{\citenamefont {Lin}\ and\ \citenamefont {Tong}(2022)}]{EFTQC2_PRXQuantum.3.010318}%
  \BibitemOpen
  \bibfield  {author} {\bibinfo {author} {\bibfnamefont {L.}~\bibnamefont {Lin}}\ and\ \bibinfo {author} {\bibfnamefont {Y.}~\bibnamefont {Tong}},\ }\bibfield  {title} {\bibinfo {title} {{Heisenberg-Limited Ground-State Energy Estimation for Early Fault-Tolerant Quantum Computers}},\ }\href {https://doi.org/10.1103/PRXQuantum.3.010318} {\bibfield  {journal} {\bibinfo  {journal} {PRX Quantum}\ }\textbf {\bibinfo {volume} {3}},\ \bibinfo {pages} {010318} (\bibinfo {year} {2022})}\BibitemShut {NoStop}%
\bibitem [{\citenamefont {Kshirsagar}\ \emph {et~al.}(2022)\citenamefont {Kshirsagar}, \citenamefont {Katabarwa},\ and\ \citenamefont {Johnson}}]{EFTQC3_kshirsagar2022proving}%
  \BibitemOpen
  \bibfield  {author} {\bibinfo {author} {\bibfnamefont {R.}~\bibnamefont {Kshirsagar}}, \bibinfo {author} {\bibfnamefont {A.}~\bibnamefont {Katabarwa}},\ and\ \bibinfo {author} {\bibfnamefont {P.~D.}\ \bibnamefont {Johnson}},\ }\bibfield  {title} {\bibinfo {title} {{On proving the robustness of algorithms for early fault-tolerant quantum computers}},\ }\href {https://arxiv.org/abs/2209.11322} {\bibfield  {journal} {\bibinfo  {journal} {arXiv preprint arXiv:2209.11322}\ } (\bibinfo {year} {2022})}\BibitemShut {NoStop}%
\bibitem [{\citenamefont {Ding}\ and\ \citenamefont {Lin}(2022)}]{EFTQC4_ding2022even}%
  \BibitemOpen
  \bibfield  {author} {\bibinfo {author} {\bibfnamefont {Z.}~\bibnamefont {Ding}}\ and\ \bibinfo {author} {\bibfnamefont {L.}~\bibnamefont {Lin}},\ }\bibfield  {title} {\bibinfo {title} {{Even shorter quantum circuit for phase estimation on early fault-tolerant quantum computers with applications to ground-state energy estimation}},\ }\href {https://arxiv.org/abs/2211.11973} {\bibfield  {journal} {\bibinfo  {journal} {arXiv preprint arXiv:2211.11973}\ } (\bibinfo {year} {2022})}\BibitemShut {NoStop}%
\bibitem [{\citenamefont {Kuroiwa}\ and\ \citenamefont {Nakagawa}(2023)}]{EFTQC5_kuroiwa2023clifford+}%
  \BibitemOpen
  \bibfield  {author} {\bibinfo {author} {\bibfnamefont {K.}~\bibnamefont {Kuroiwa}}\ and\ \bibinfo {author} {\bibfnamefont {Y.~O.}\ \bibnamefont {Nakagawa}},\ }\bibfield  {title} {\bibinfo {title} {{Clifford+ $ T $-gate Decomposition with Limited Number of $ T $ gates, its Error Analysis, and Performance of Unitary Coupled Cluster Ansatz in Pre-FTQC Era}},\ }\href {https://arxiv.org/abs/2301.04150} {\bibfield  {journal} {\bibinfo  {journal} {arXiv preprint arXiv:2301.04150}\ } (\bibinfo {year} {2023})}\BibitemShut {NoStop}%
\bibitem [{\citenamefont {Akahoshi}\ \emph {et~al.}(2024)\citenamefont {Akahoshi}, \citenamefont {Maruyama}, \citenamefont {Oshima}, \citenamefont {Sato},\ and\ \citenamefont {Fujii}}]{STAR_PRXQuantum.5.010337}%
  \BibitemOpen
  \bibfield  {author} {\bibinfo {author} {\bibfnamefont {Y.}~\bibnamefont {Akahoshi}}, \bibinfo {author} {\bibfnamefont {K.}~\bibnamefont {Maruyama}}, \bibinfo {author} {\bibfnamefont {H.}~\bibnamefont {Oshima}}, \bibinfo {author} {\bibfnamefont {S.}~\bibnamefont {Sato}},\ and\ \bibinfo {author} {\bibfnamefont {K.}~\bibnamefont {Fujii}},\ }\bibfield  {title} {\bibinfo {title} {{Partially Fault-Tolerant Quantum Computing Architecture with Error-Corrected Clifford Gates and Space-Time Efficient Analog Rotations}},\ }\href {https://doi.org/10.1103/PRXQuantum.5.010337} {\bibfield  {journal} {\bibinfo  {journal} {PRX Quantum}\ }\textbf {\bibinfo {volume} {5}},\ \bibinfo {pages} {010337} (\bibinfo {year} {2024})}\BibitemShut {NoStop}%
\bibitem [{\citenamefont {Motta}\ \emph {et~al.}(2024)\citenamefont {Motta}, \citenamefont {Kirby}, \citenamefont {Liepuoniute}, \citenamefont {Sung}, \citenamefont {Cohn}, \citenamefont {Mezzacapo}, \citenamefont {Klymko}, \citenamefont {Nguyen}, \citenamefont {Yoshioka},\ and\ \citenamefont {Rice}}]{SM_motta2024subspace}%
  \BibitemOpen
  \bibfield  {author} {\bibinfo {author} {\bibfnamefont {M.}~\bibnamefont {Motta}}, \bibinfo {author} {\bibfnamefont {W.}~\bibnamefont {Kirby}}, \bibinfo {author} {\bibfnamefont {I.}~\bibnamefont {Liepuoniute}}, \bibinfo {author} {\bibfnamefont {K.~J.}\ \bibnamefont {Sung}}, \bibinfo {author} {\bibfnamefont {J.}~\bibnamefont {Cohn}}, \bibinfo {author} {\bibfnamefont {A.}~\bibnamefont {Mezzacapo}}, \bibinfo {author} {\bibfnamefont {K.}~\bibnamefont {Klymko}}, \bibinfo {author} {\bibfnamefont {N.}~\bibnamefont {Nguyen}}, \bibinfo {author} {\bibfnamefont {N.}~\bibnamefont {Yoshioka}},\ and\ \bibinfo {author} {\bibfnamefont {J.~E.}\ \bibnamefont {Rice}},\ }\bibfield  {title} {\bibinfo {title} {Subspace methods for electronic structure simulations on quantum computers},\ }\href {https://iopscience.iop.org/article/10.1088/2516-1075/ad3592/meta} {\bibfield  {journal} {\bibinfo  {journal} {Electronic Structure}\ }\textbf {\bibinfo {volume} {6}},\ \bibinfo {pages} {013001} (\bibinfo {year} {2024})}\BibitemShut
  {NoStop}%
\bibitem [{\citenamefont {McClean}\ \emph {et~al.}(2017)\citenamefont {McClean}, \citenamefont {Kimchi-Schwartz}, \citenamefont {Carter},\ and\ \citenamefont {de~Jong}}]{QSE_PhysRevA.95.042308}%
  \BibitemOpen
  \bibfield  {author} {\bibinfo {author} {\bibfnamefont {J.~R.}\ \bibnamefont {McClean}}, \bibinfo {author} {\bibfnamefont {M.~E.}\ \bibnamefont {Kimchi-Schwartz}}, \bibinfo {author} {\bibfnamefont {J.}~\bibnamefont {Carter}},\ and\ \bibinfo {author} {\bibfnamefont {W.~A.}\ \bibnamefont {de~Jong}},\ }\bibfield  {title} {\bibinfo {title} {Hybrid quantum-classical hierarchy for mitigation of decoherence and determination of excited states},\ }\href {https://doi.org/10.1103/PhysRevA.95.042308} {\bibfield  {journal} {\bibinfo  {journal} {Phys. Rev. A}\ }\textbf {\bibinfo {volume} {95}},\ \bibinfo {pages} {042308} (\bibinfo {year} {2017})}\BibitemShut {NoStop}%
\bibitem [{\citenamefont {Colless}\ \emph {et~al.}(2018)\citenamefont {Colless}, \citenamefont {Ramasesh}, \citenamefont {Dahlen}, \citenamefont {Blok}, \citenamefont {Kimchi-Schwartz}, \citenamefont {McClean}, \citenamefont {Carter}, \citenamefont {de~Jong},\ and\ \citenamefont {Siddiqi}}]{QSE_PhysRevX.8.011021}%
  \BibitemOpen
  \bibfield  {author} {\bibinfo {author} {\bibfnamefont {J.~I.}\ \bibnamefont {Colless}}, \bibinfo {author} {\bibfnamefont {V.~V.}\ \bibnamefont {Ramasesh}}, \bibinfo {author} {\bibfnamefont {D.}~\bibnamefont {Dahlen}}, \bibinfo {author} {\bibfnamefont {M.~S.}\ \bibnamefont {Blok}}, \bibinfo {author} {\bibfnamefont {M.~E.}\ \bibnamefont {Kimchi-Schwartz}}, \bibinfo {author} {\bibfnamefont {J.~R.}\ \bibnamefont {McClean}}, \bibinfo {author} {\bibfnamefont {J.}~\bibnamefont {Carter}}, \bibinfo {author} {\bibfnamefont {W.~A.}\ \bibnamefont {de~Jong}},\ and\ \bibinfo {author} {\bibfnamefont {I.}~\bibnamefont {Siddiqi}},\ }\bibfield  {title} {\bibinfo {title} {{Computation of Molecular Spectra on a Quantum Processor with an Error-Resilient Algorithm}},\ }\href {https://doi.org/10.1103/PhysRevX.8.011021} {\bibfield  {journal} {\bibinfo  {journal} {Phys. Rev. X}\ }\textbf {\bibinfo {volume} {8}},\ \bibinfo {pages} {011021} (\bibinfo {year} {2018})}\BibitemShut {NoStop}%
\bibitem [{\citenamefont {Ollitrault}\ \emph {et~al.}(2020)\citenamefont {Ollitrault}, \citenamefont {Kandala}, \citenamefont {Chen}, \citenamefont {Barkoutsos}, \citenamefont {Mezzacapo}, \citenamefont {Pistoia}, \citenamefont {Sheldon}, \citenamefont {Woerner}, \citenamefont {Gambetta},\ and\ \citenamefont {Tavernelli}}]{qEOM_PhysRevResearch.2.043140}%
  \BibitemOpen
  \bibfield  {author} {\bibinfo {author} {\bibfnamefont {P.~J.}\ \bibnamefont {Ollitrault}}, \bibinfo {author} {\bibfnamefont {A.}~\bibnamefont {Kandala}}, \bibinfo {author} {\bibfnamefont {C.-F.}\ \bibnamefont {Chen}}, \bibinfo {author} {\bibfnamefont {P.~K.}\ \bibnamefont {Barkoutsos}}, \bibinfo {author} {\bibfnamefont {A.}~\bibnamefont {Mezzacapo}}, \bibinfo {author} {\bibfnamefont {M.}~\bibnamefont {Pistoia}}, \bibinfo {author} {\bibfnamefont {S.}~\bibnamefont {Sheldon}}, \bibinfo {author} {\bibfnamefont {S.}~\bibnamefont {Woerner}}, \bibinfo {author} {\bibfnamefont {J.~M.}\ \bibnamefont {Gambetta}},\ and\ \bibinfo {author} {\bibfnamefont {I.}~\bibnamefont {Tavernelli}},\ }\bibfield  {title} {\bibinfo {title} {Quantum equation of motion for computing molecular excitation energies on a noisy quantum processor},\ }\href {https://doi.org/10.1103/PhysRevResearch.2.043140} {\bibfield  {journal} {\bibinfo  {journal} {Phys. Rev. Res.}\ }\textbf {\bibinfo {volume} {2}},\ \bibinfo {pages} {043140} (\bibinfo {year}
  {2020})}\BibitemShut {NoStop}%
\bibitem [{\citenamefont {Parrish}\ and\ \citenamefont {McMahon}(2019)}]{QFD_parrish2019quantum}%
  \BibitemOpen
  \bibfield  {author} {\bibinfo {author} {\bibfnamefont {R.~M.}\ \bibnamefont {Parrish}}\ and\ \bibinfo {author} {\bibfnamefont {P.~L.}\ \bibnamefont {McMahon}},\ }\bibfield  {title} {\bibinfo {title} {{Quantum filter diagonalization: Quantum eigendecomposition without full quantum phase estimation}},\ }\href {https://arxiv.org/abs/1909.08925} {\bibfield  {journal} {\bibinfo  {journal} {arXiv preprint arXiv:1909.08925}\ } (\bibinfo {year} {2019})}\BibitemShut {NoStop}%
\bibitem [{\citenamefont {Stair}\ \emph {et~al.}(2020)\citenamefont {Stair}, \citenamefont {Huang},\ and\ \citenamefont {Evangelista}}]{MRSQK_stair2020multireference}%
  \BibitemOpen
  \bibfield  {author} {\bibinfo {author} {\bibfnamefont {N.~H.}\ \bibnamefont {Stair}}, \bibinfo {author} {\bibfnamefont {R.}~\bibnamefont {Huang}},\ and\ \bibinfo {author} {\bibfnamefont {F.~A.}\ \bibnamefont {Evangelista}},\ }\bibfield  {title} {\bibinfo {title} {{A multireference quantum Krylov algorithm for strongly correlated electrons}},\ }\href {https://pubs.acs.org/doi/abs/10.1021/acs.jctc.9b01125} {\bibfield  {journal} {\bibinfo  {journal} {Journal of chemical theory and computation}\ }\textbf {\bibinfo {volume} {16}},\ \bibinfo {pages} {2236} (\bibinfo {year} {2020})}\BibitemShut {NoStop}%
\bibitem [{\citenamefont {Motta}\ \emph {et~al.}(2020)\citenamefont {Motta}, \citenamefont {Sun}, \citenamefont {Tan}, \citenamefont {O’Rourke}, \citenamefont {Ye}, \citenamefont {Minnich}, \citenamefont {Brandao},\ and\ \citenamefont {Chan}}]{QLanczos_motta2020determining}%
  \BibitemOpen
  \bibfield  {author} {\bibinfo {author} {\bibfnamefont {M.}~\bibnamefont {Motta}}, \bibinfo {author} {\bibfnamefont {C.}~\bibnamefont {Sun}}, \bibinfo {author} {\bibfnamefont {A.~T.}\ \bibnamefont {Tan}}, \bibinfo {author} {\bibfnamefont {M.~J.}\ \bibnamefont {O’Rourke}}, \bibinfo {author} {\bibfnamefont {E.}~\bibnamefont {Ye}}, \bibinfo {author} {\bibfnamefont {A.~J.}\ \bibnamefont {Minnich}}, \bibinfo {author} {\bibfnamefont {F.~G.}\ \bibnamefont {Brandao}},\ and\ \bibinfo {author} {\bibfnamefont {G.~K.-L.}\ \bibnamefont {Chan}},\ }\bibfield  {title} {\bibinfo {title} {Determining eigenstates and thermal states on a quantum computer using quantum imaginary time evolution},\ }\href {https://www.nature.com/articles/s41567-019-0704-4} {\bibfield  {journal} {\bibinfo  {journal} {Nature Physics}\ }\textbf {\bibinfo {volume} {16}},\ \bibinfo {pages} {205} (\bibinfo {year} {2020})}\BibitemShut {NoStop}%
\bibitem [{\citenamefont {Seki}\ and\ \citenamefont {Yunoki}(2021)}]{QuantumPower_PRXQuantum.2.010333}%
  \BibitemOpen
  \bibfield  {author} {\bibinfo {author} {\bibfnamefont {K.}~\bibnamefont {Seki}}\ and\ \bibinfo {author} {\bibfnamefont {S.}~\bibnamefont {Yunoki}},\ }\bibfield  {title} {\bibinfo {title} {{Quantum Power Method by a Superposition of Time-Evolved States}},\ }\href {https://doi.org/10.1103/PRXQuantum.2.010333} {\bibfield  {journal} {\bibinfo  {journal} {PRX Quantum}\ }\textbf {\bibinfo {volume} {2}},\ \bibinfo {pages} {010333} (\bibinfo {year} {2021})}\BibitemShut {NoStop}%
\bibitem [{\citenamefont {Kanno}\ \emph {et~al.}(2023)\citenamefont {Kanno}, \citenamefont {Kohda}, \citenamefont {Imai}, \citenamefont {Koh}, \citenamefont {Mitarai}, \citenamefont {Mizukami},\ and\ \citenamefont {Nakagawa}}]{QSCI_kanno2023quantum}%
  \BibitemOpen
  \bibfield  {author} {\bibinfo {author} {\bibfnamefont {K.}~\bibnamefont {Kanno}}, \bibinfo {author} {\bibfnamefont {M.}~\bibnamefont {Kohda}}, \bibinfo {author} {\bibfnamefont {R.}~\bibnamefont {Imai}}, \bibinfo {author} {\bibfnamefont {S.}~\bibnamefont {Koh}}, \bibinfo {author} {\bibfnamefont {K.}~\bibnamefont {Mitarai}}, \bibinfo {author} {\bibfnamefont {W.}~\bibnamefont {Mizukami}},\ and\ \bibinfo {author} {\bibfnamefont {Y.~O.}\ \bibnamefont {Nakagawa}},\ }\bibfield  {title} {\bibinfo {title} {{Quantum-Selected Configuration Interaction: classical diagonalization of Hamiltonians in subspaces selected by quantum computers}},\ }\href {https://arxiv.org/abs/2302.11320} {\bibfield  {journal} {\bibinfo  {journal} {arXiv preprint arXiv:2302.11320}\ } (\bibinfo {year} {2023})}\BibitemShut {NoStop}%
\bibitem [{\citenamefont {Nakagawa}\ \emph {et~al.}(2023)\citenamefont {Nakagawa}, \citenamefont {Kamoshita}, \citenamefont {Mizukami}, \citenamefont {Sudo},\ and\ \citenamefont {Ohnishi}}]{nakagawa2023adapt-qsci}%
  \BibitemOpen
  \bibfield  {author} {\bibinfo {author} {\bibfnamefont {Y.~O.}\ \bibnamefont {Nakagawa}}, \bibinfo {author} {\bibfnamefont {M.}~\bibnamefont {Kamoshita}}, \bibinfo {author} {\bibfnamefont {W.}~\bibnamefont {Mizukami}}, \bibinfo {author} {\bibfnamefont {S.}~\bibnamefont {Sudo}},\ and\ \bibinfo {author} {\bibfnamefont {Y.-y.}\ \bibnamefont {Ohnishi}},\ }\bibfield  {title} {\bibinfo {title} {Adapt-qsci: Adaptive construction of input state for quantum-selected configuration interaction},\ }\href {https://arxiv.org/abs/2311.01105} {\bibfield  {journal} {\bibinfo  {journal} {arXiv preprint arXiv:2311.01105}\ } (\bibinfo {year} {2023})}\BibitemShut {NoStop}%
\bibitem [{\citenamefont {{\c{S}}ahino{\u{g}}lu}\ and\ \citenamefont {Somma}(2021{\natexlab{a}})}]{TrotSubspace_csahinouglu2021hamiltonian}%
  \BibitemOpen
  \bibfield  {author} {\bibinfo {author} {\bibfnamefont {B.}~\bibnamefont {{\c{S}}ahino{\u{g}}lu}}\ and\ \bibinfo {author} {\bibfnamefont {R.~D.}\ \bibnamefont {Somma}},\ }\bibfield  {title} {\bibinfo {title} {Hamiltonian simulation in the low-energy subspace},\ }\href {https://www.nature.com/articles/s41534-021-00451-w} {\bibfield  {journal} {\bibinfo  {journal} {npj Quantum Information}\ }\textbf {\bibinfo {volume} {7}},\ \bibinfo {pages} {119} (\bibinfo {year} {2021}{\natexlab{a}})}\BibitemShut {NoStop}%
\bibitem [{\citenamefont {Heya}\ \emph {et~al.}(2023)\citenamefont {Heya}, \citenamefont {Nakanishi}, \citenamefont {Mitarai}, \citenamefont {Yan}, \citenamefont {Zuo}, \citenamefont {Suzuki}, \citenamefont {Sugiyama}, \citenamefont {Tamate}, \citenamefont {Tabuchi}, \citenamefont {Fujii},\ and\ \citenamefont {Nakamura}}]{SVQS_PhysRevResearch.5.023078}%
  \BibitemOpen
  \bibfield  {author} {\bibinfo {author} {\bibfnamefont {K.}~\bibnamefont {Heya}}, \bibinfo {author} {\bibfnamefont {K.~M.}\ \bibnamefont {Nakanishi}}, \bibinfo {author} {\bibfnamefont {K.}~\bibnamefont {Mitarai}}, \bibinfo {author} {\bibfnamefont {Z.}~\bibnamefont {Yan}}, \bibinfo {author} {\bibfnamefont {K.}~\bibnamefont {Zuo}}, \bibinfo {author} {\bibfnamefont {Y.}~\bibnamefont {Suzuki}}, \bibinfo {author} {\bibfnamefont {T.}~\bibnamefont {Sugiyama}}, \bibinfo {author} {\bibfnamefont {S.}~\bibnamefont {Tamate}}, \bibinfo {author} {\bibfnamefont {Y.}~\bibnamefont {Tabuchi}}, \bibinfo {author} {\bibfnamefont {K.}~\bibnamefont {Fujii}},\ and\ \bibinfo {author} {\bibfnamefont {Y.}~\bibnamefont {Nakamura}},\ }\bibfield  {title} {\bibinfo {title} {Subspace variational quantum simulator},\ }\href {https://doi.org/10.1103/PhysRevResearch.5.023078} {\bibfield  {journal} {\bibinfo  {journal} {Phys. Rev. Res.}\ }\textbf {\bibinfo {volume} {5}},\ \bibinfo {pages} {023078} (\bibinfo {year} {2023})}\BibitemShut
  {NoStop}%
\bibitem [{\citenamefont {Lim}\ \emph {et~al.}(2021)\citenamefont {Lim}, \citenamefont {Haug}, \citenamefont {Kwek},\ and\ \citenamefont {Bharti}}]{CQFF_lim2021fast}%
  \BibitemOpen
  \bibfield  {author} {\bibinfo {author} {\bibfnamefont {K.~H.}\ \bibnamefont {Lim}}, \bibinfo {author} {\bibfnamefont {T.}~\bibnamefont {Haug}}, \bibinfo {author} {\bibfnamefont {L.~C.}\ \bibnamefont {Kwek}},\ and\ \bibinfo {author} {\bibfnamefont {K.}~\bibnamefont {Bharti}},\ }\bibfield  {title} {\bibinfo {title} {{Fast-forwarding with NISQ processors without feedback loop}},\ }\href {https://iopscience.iop.org/article/10.1088/2058-9565/ac2e52/meta} {\bibfield  {journal} {\bibinfo  {journal} {Quantum Science and Technology}\ }\textbf {\bibinfo {volume} {7}},\ \bibinfo {pages} {015001} (\bibinfo {year} {2021})}\BibitemShut {NoStop}%
\bibitem [{\citenamefont {Lieb}\ and\ \citenamefont {Robinson}(1972)}]{lieb1972finite}%
  \BibitemOpen
  \bibfield  {author} {\bibinfo {author} {\bibfnamefont {E.~H.}\ \bibnamefont {Lieb}}\ and\ \bibinfo {author} {\bibfnamefont {D.~W.}\ \bibnamefont {Robinson}},\ }\bibfield  {title} {\bibinfo {title} {{The finite group velocity of quantum spin systems}},\ }\href {https://link.springer.com/article/10.1007/BF01645779} {\bibfield  {journal} {\bibinfo  {journal} {Commun. Math. Phys.}\ }\textbf {\bibinfo {volume} {28}},\ \bibinfo {pages} {251} (\bibinfo {year} {1972})}\BibitemShut {NoStop}%
\bibitem [{\citenamefont {Mizuta}\ \emph {et~al.}(2022)\citenamefont {Mizuta}, \citenamefont {Nakagawa}, \citenamefont {Mitarai},\ and\ \citenamefont {Fujii}}]{LVQC_PRXQuantum.3.040302}%
  \BibitemOpen
  \bibfield  {author} {\bibinfo {author} {\bibfnamefont {K.}~\bibnamefont {Mizuta}}, \bibinfo {author} {\bibfnamefont {Y.~O.}\ \bibnamefont {Nakagawa}}, \bibinfo {author} {\bibfnamefont {K.}~\bibnamefont {Mitarai}},\ and\ \bibinfo {author} {\bibfnamefont {K.}~\bibnamefont {Fujii}},\ }\bibfield  {title} {\bibinfo {title} {{Local Variational Quantum Compilation of Large-Scale Hamiltonian Dynamics}},\ }\href {https://doi.org/10.1103/PRXQuantum.3.040302} {\bibfield  {journal} {\bibinfo  {journal} {PRX Quantum}\ }\textbf {\bibinfo {volume} {3}},\ \bibinfo {pages} {040302} (\bibinfo {year} {2022})}\BibitemShut {NoStop}%
\bibitem [{\citenamefont {Kanasugi}\ \emph {et~al.}(2023)\citenamefont {Kanasugi}, \citenamefont {Tsutsui}, \citenamefont {Nakagawa}, \citenamefont {Maruyama}, \citenamefont {Oshima},\ and\ \citenamefont {Sato}}]{LVQC_GF_PhysRevResearch.5.033070}%
  \BibitemOpen
  \bibfield  {author} {\bibinfo {author} {\bibfnamefont {S.}~\bibnamefont {Kanasugi}}, \bibinfo {author} {\bibfnamefont {S.}~\bibnamefont {Tsutsui}}, \bibinfo {author} {\bibfnamefont {Y.~O.}\ \bibnamefont {Nakagawa}}, \bibinfo {author} {\bibfnamefont {K.}~\bibnamefont {Maruyama}}, \bibinfo {author} {\bibfnamefont {H.}~\bibnamefont {Oshima}},\ and\ \bibinfo {author} {\bibfnamefont {S.}~\bibnamefont {Sato}},\ }\bibfield  {title} {\bibinfo {title} {{Computation of Green's function by local variational quantum compilation}},\ }\href {https://doi.org/10.1103/PhysRevResearch.5.033070} {\bibfield  {journal} {\bibinfo  {journal} {Phys. Rev. Res.}\ }\textbf {\bibinfo {volume} {5}},\ \bibinfo {pages} {033070} (\bibinfo {year} {2023})}\BibitemShut {NoStop}%
\bibitem [{\citenamefont {Ami}\ \emph {et~al.}(1995)\citenamefont {Ami}, \citenamefont {Crawford}, \citenamefont {Harlow}, \citenamefont {Wang}, \citenamefont {Johnston}, \citenamefont {Huang},\ and\ \citenamefont {Erwin}}]{SCO_AFM_PhysRevB.51.5994}%
  \BibitemOpen
  \bibfield  {author} {\bibinfo {author} {\bibfnamefont {T.}~\bibnamefont {Ami}}, \bibinfo {author} {\bibfnamefont {M.~K.}\ \bibnamefont {Crawford}}, \bibinfo {author} {\bibfnamefont {R.~L.}\ \bibnamefont {Harlow}}, \bibinfo {author} {\bibfnamefont {Z.~R.}\ \bibnamefont {Wang}}, \bibinfo {author} {\bibfnamefont {D.~C.}\ \bibnamefont {Johnston}}, \bibinfo {author} {\bibfnamefont {Q.}~\bibnamefont {Huang}},\ and\ \bibinfo {author} {\bibfnamefont {R.~W.}\ \bibnamefont {Erwin}},\ }\bibfield  {title} {\bibinfo {title} {{Magnetic susceptibility and low-temperature structure of the linear chain cuprate ${\mathrm{Sr}}_{2}$${\mathrm{CuO}}_{3}$}},\ }\href {https://doi.org/10.1103/PhysRevB.51.5994} {\bibfield  {journal} {\bibinfo  {journal} {Phys. Rev. B}\ }\textbf {\bibinfo {volume} {51}},\ \bibinfo {pages} {5994} (\bibinfo {year} {1995})}\BibitemShut {NoStop}%
\bibitem [{\citenamefont {Motoyama}\ \emph {et~al.}(1996)\citenamefont {Motoyama}, \citenamefont {Eisaki},\ and\ \citenamefont {Uchida}}]{SCO_AFM_PhysRevLett.76.3212}%
  \BibitemOpen
  \bibfield  {author} {\bibinfo {author} {\bibfnamefont {N.}~\bibnamefont {Motoyama}}, \bibinfo {author} {\bibfnamefont {H.}~\bibnamefont {Eisaki}},\ and\ \bibinfo {author} {\bibfnamefont {S.}~\bibnamefont {Uchida}},\ }\bibfield  {title} {\bibinfo {title} {{Magnetic Susceptibility of Ideal Spin 1 $/$2 Heisenberg Antiferromagnetic Chain Systems, ${\mathrm{Sr}}_{2}{\mathrm{CuO}}_{3}$ and ${\mathrm{SrCuO}}_{2}$}},\ }\href {https://doi.org/10.1103/PhysRevLett.76.3212} {\bibfield  {journal} {\bibinfo  {journal} {Phys. Rev. Lett.}\ }\textbf {\bibinfo {volume} {76}},\ \bibinfo {pages} {3212} (\bibinfo {year} {1996})}\BibitemShut {NoStop}%
\bibitem [{\citenamefont {Neudert}\ \emph {et~al.}(1998)\citenamefont {Neudert}, \citenamefont {Knupfer}, \citenamefont {Golden}, \citenamefont {Fink}, \citenamefont {Stephan}, \citenamefont {Penc}, \citenamefont {Motoyama}, \citenamefont {Eisaki},\ and\ \citenamefont {Uchida}}]{SCO_separation_PhysRevLett.81.657}%
  \BibitemOpen
  \bibfield  {author} {\bibinfo {author} {\bibfnamefont {R.}~\bibnamefont {Neudert}}, \bibinfo {author} {\bibfnamefont {M.}~\bibnamefont {Knupfer}}, \bibinfo {author} {\bibfnamefont {M.~S.}\ \bibnamefont {Golden}}, \bibinfo {author} {\bibfnamefont {J.}~\bibnamefont {Fink}}, \bibinfo {author} {\bibfnamefont {W.}~\bibnamefont {Stephan}}, \bibinfo {author} {\bibfnamefont {K.}~\bibnamefont {Penc}}, \bibinfo {author} {\bibfnamefont {N.}~\bibnamefont {Motoyama}}, \bibinfo {author} {\bibfnamefont {H.}~\bibnamefont {Eisaki}},\ and\ \bibinfo {author} {\bibfnamefont {S.}~\bibnamefont {Uchida}},\ }\bibfield  {title} {\bibinfo {title} {{Manifestation of Spin-Charge Separation in the Dynamic Dielectric Response of One-Dimensional ${\mathrm{Sr}}_{2}\mathrm{Cu}{O}_{3}$}},\ }\href {https://doi.org/10.1103/PhysRevLett.81.657} {\bibfield  {journal} {\bibinfo  {journal} {Phys. Rev. Lett.}\ }\textbf {\bibinfo {volume} {81}},\ \bibinfo {pages} {657} (\bibinfo {year} {1998})}\BibitemShut {NoStop}%
\bibitem [{\citenamefont {Fujisawa}\ \emph {et~al.}(1999)\citenamefont {Fujisawa}, \citenamefont {Yokoya}, \citenamefont {Takahashi}, \citenamefont {Miyasaka}, \citenamefont {Kibune},\ and\ \citenamefont {Takagi}}]{SCO_separation_PhysRevB.59.7358}%
  \BibitemOpen
  \bibfield  {author} {\bibinfo {author} {\bibfnamefont {H.}~\bibnamefont {Fujisawa}}, \bibinfo {author} {\bibfnamefont {T.}~\bibnamefont {Yokoya}}, \bibinfo {author} {\bibfnamefont {T.}~\bibnamefont {Takahashi}}, \bibinfo {author} {\bibfnamefont {S.}~\bibnamefont {Miyasaka}}, \bibinfo {author} {\bibfnamefont {M.}~\bibnamefont {Kibune}},\ and\ \bibinfo {author} {\bibfnamefont {H.}~\bibnamefont {Takagi}},\ }\bibfield  {title} {\bibinfo {title} {{Angle-resolved photoemission study of ${\mathrm{Sr}}_{2}{\mathrm{CuO}}_{3}$}},\ }\href {https://doi.org/10.1103/PhysRevB.59.7358} {\bibfield  {journal} {\bibinfo  {journal} {Phys. Rev. B}\ }\textbf {\bibinfo {volume} {59}},\ \bibinfo {pages} {7358} (\bibinfo {year} {1999})}\BibitemShut {NoStop}%
\bibitem [{\citenamefont {Liu}\ \emph {et~al.}(2014)\citenamefont {Liu}, \citenamefont {Shen}, \citenamefont {Liu}, \citenamefont {Li}, \citenamefont {Feng}, \citenamefont {Yu}, \citenamefont {Uchida},\ and\ \citenamefont {Jin}}]{SCO_SC_liu2014new}%
  \BibitemOpen
  \bibfield  {author} {\bibinfo {author} {\bibfnamefont {Y.}~\bibnamefont {Liu}}, \bibinfo {author} {\bibfnamefont {X.}~\bibnamefont {Shen}}, \bibinfo {author} {\bibfnamefont {Q.}~\bibnamefont {Liu}}, \bibinfo {author} {\bibfnamefont {X.}~\bibnamefont {Li}}, \bibinfo {author} {\bibfnamefont {S.}~\bibnamefont {Feng}}, \bibinfo {author} {\bibfnamefont {R.}~\bibnamefont {Yu}}, \bibinfo {author} {\bibfnamefont {S.}~\bibnamefont {Uchida}},\ and\ \bibinfo {author} {\bibfnamefont {C.}~\bibnamefont {Jin}},\ }\bibfield  {title} {\bibinfo {title} {A new modulated structure in sr2cuo3+ $\delta$ superconductor synthesized under high pressure},\ }\href {https://www.sciencedirect.com/science/article/pii/S0921453413004176} {\bibfield  {journal} {\bibinfo  {journal} {Physica C: Superconductivity and its applications}\ }\textbf {\bibinfo {volume} {497}},\ \bibinfo {pages} {34} (\bibinfo {year} {2014})}\BibitemShut {NoStop}%
\bibitem [{\citenamefont {Atia}\ and\ \citenamefont {Aharonov}(2017)}]{FF_atia2017fast}%
  \BibitemOpen
  \bibfield  {author} {\bibinfo {author} {\bibfnamefont {Y.}~\bibnamefont {Atia}}\ and\ \bibinfo {author} {\bibfnamefont {D.}~\bibnamefont {Aharonov}},\ }\bibfield  {title} {\bibinfo {title} {Fast-forwarding of hamiltonians and exponentially precise measurements},\ }\href {https://www.nature.com/articles/s41467-017-01637-7} {\bibfield  {journal} {\bibinfo  {journal} {Nature communications}\ }\textbf {\bibinfo {volume} {8}},\ \bibinfo {pages} {1572} (\bibinfo {year} {2017})}\BibitemShut {NoStop}%
\bibitem [{\citenamefont {Berry}\ \emph {et~al.}(2007)\citenamefont {Berry}, \citenamefont {Ahokas}, \citenamefont {Cleve},\ and\ \citenamefont {Sanders}}]{FF_berry2007efficient}%
  \BibitemOpen
  \bibfield  {author} {\bibinfo {author} {\bibfnamefont {D.~W.}\ \bibnamefont {Berry}}, \bibinfo {author} {\bibfnamefont {G.}~\bibnamefont {Ahokas}}, \bibinfo {author} {\bibfnamefont {R.}~\bibnamefont {Cleve}},\ and\ \bibinfo {author} {\bibfnamefont {B.~C.}\ \bibnamefont {Sanders}},\ }\bibfield  {title} {\bibinfo {title} {Efficient quantum algorithms for simulating sparse hamiltonians},\ }\href {https://link.springer.com/article/10.1007/s00220-006-0150-x} {\bibfield  {journal} {\bibinfo  {journal} {Communications in Mathematical Physics}\ }\textbf {\bibinfo {volume} {270}},\ \bibinfo {pages} {359} (\bibinfo {year} {2007})}\BibitemShut {NoStop}%
\bibitem [{\citenamefont {Gu}\ \emph {et~al.}(2021)\citenamefont {Gu}, \citenamefont {Somma},\ and\ \citenamefont {{\c{S}}ahino{\u{g}}lu}}]{FF_gu2021fast}%
  \BibitemOpen
  \bibfield  {author} {\bibinfo {author} {\bibfnamefont {S.}~\bibnamefont {Gu}}, \bibinfo {author} {\bibfnamefont {R.~D.}\ \bibnamefont {Somma}},\ and\ \bibinfo {author} {\bibfnamefont {B.}~\bibnamefont {{\c{S}}ahino{\u{g}}lu}},\ }\bibfield  {title} {\bibinfo {title} {Fast-forwarding quantum evolution},\ }\href {https://quantum-journal.org/papers/q-2021-11-15-577/} {\bibfield  {journal} {\bibinfo  {journal} {Quantum}\ }\textbf {\bibinfo {volume} {5}},\ \bibinfo {pages} {577} (\bibinfo {year} {2021})}\BibitemShut {NoStop}%
\bibitem [{\citenamefont {Loke}\ and\ \citenamefont {Wang}(2017)}]{FF_loke2017efficient}%
  \BibitemOpen
  \bibfield  {author} {\bibinfo {author} {\bibfnamefont {T.}~\bibnamefont {Loke}}\ and\ \bibinfo {author} {\bibfnamefont {J.}~\bibnamefont {Wang}},\ }\bibfield  {title} {\bibinfo {title} {Efficient quantum circuits for continuous-time quantum walks on composite graphs},\ }\href {https://iopscience.iop.org/article/10.1088/1751-8121/aa53a9/meta} {\bibfield  {journal} {\bibinfo  {journal} {Journal of Physics A: Mathematical and Theoretical}\ }\textbf {\bibinfo {volume} {50}},\ \bibinfo {pages} {055303} (\bibinfo {year} {2017})}\BibitemShut {NoStop}%
\bibitem [{\citenamefont {Nachtergaele}\ and\ \citenamefont {Sims}(2006)}]{LR_nachtergaele2006lieb}%
  \BibitemOpen
  \bibfield  {author} {\bibinfo {author} {\bibfnamefont {B.}~\bibnamefont {Nachtergaele}}\ and\ \bibinfo {author} {\bibfnamefont {R.}~\bibnamefont {Sims}},\ }\bibfield  {title} {\bibinfo {title} {Lieb-robinson bounds and the exponential clustering theorem},\ }\href {https://link.springer.com/article/10.1007/s00220-006-1556-1} {\bibfield  {journal} {\bibinfo  {journal} {Communications in mathematical physics}\ }\textbf {\bibinfo {volume} {265}},\ \bibinfo {pages} {119} (\bibinfo {year} {2006})}\BibitemShut {NoStop}%
\bibitem [{\citenamefont {Nachtergaele}\ \emph {et~al.}(2006)\citenamefont {Nachtergaele}, \citenamefont {Ogata},\ and\ \citenamefont {Sims}}]{LR_nachtergaele2006propagation}%
  \BibitemOpen
  \bibfield  {author} {\bibinfo {author} {\bibfnamefont {B.}~\bibnamefont {Nachtergaele}}, \bibinfo {author} {\bibfnamefont {Y.}~\bibnamefont {Ogata}},\ and\ \bibinfo {author} {\bibfnamefont {R.}~\bibnamefont {Sims}},\ }\bibfield  {title} {\bibinfo {title} {Propagation of correlations in quantum lattice systems},\ }\href {https://link.springer.com/article/10.1007/s10955-006-9143-6} {\bibfield  {journal} {\bibinfo  {journal} {Journal of statistical physics}\ }\textbf {\bibinfo {volume} {124}},\ \bibinfo {pages} {1} (\bibinfo {year} {2006})}\BibitemShut {NoStop}%
\bibitem [{\citenamefont {Hastings}\ and\ \citenamefont {Koma}(2006)}]{LR_hastings2006spectral}%
  \BibitemOpen
  \bibfield  {author} {\bibinfo {author} {\bibfnamefont {M.~B.}\ \bibnamefont {Hastings}}\ and\ \bibinfo {author} {\bibfnamefont {T.}~\bibnamefont {Koma}},\ }\bibfield  {title} {\bibinfo {title} {Spectral gap and exponential decay of correlations},\ }\href {https://link.springer.com/article/10.1007/s00220-006-0030-4} {\bibfield  {journal} {\bibinfo  {journal} {Communications in mathematical physics}\ }\textbf {\bibinfo {volume} {265}},\ \bibinfo {pages} {781} (\bibinfo {year} {2006})}\BibitemShut {NoStop}%
\bibitem [{\citenamefont {Foss-Feig}\ \emph {et~al.}(2015)\citenamefont {Foss-Feig}, \citenamefont {Gong}, \citenamefont {Clark},\ and\ \citenamefont {Gorshkov}}]{LR_PhysRevLett.114.157201}%
  \BibitemOpen
  \bibfield  {author} {\bibinfo {author} {\bibfnamefont {M.}~\bibnamefont {Foss-Feig}}, \bibinfo {author} {\bibfnamefont {Z.-X.}\ \bibnamefont {Gong}}, \bibinfo {author} {\bibfnamefont {C.~W.}\ \bibnamefont {Clark}},\ and\ \bibinfo {author} {\bibfnamefont {A.~V.}\ \bibnamefont {Gorshkov}},\ }\bibfield  {title} {\bibinfo {title} {Nearly linear light cones in long-range interacting quantum systems},\ }\href {https://doi.org/10.1103/PhysRevLett.114.157201} {\bibfield  {journal} {\bibinfo  {journal} {Phys. Rev. Lett.}\ }\textbf {\bibinfo {volume} {114}},\ \bibinfo {pages} {157201} (\bibinfo {year} {2015})}\BibitemShut {NoStop}%
\bibitem [{\citenamefont {Matsuta}\ \emph {et~al.}(2017)\citenamefont {Matsuta}, \citenamefont {Koma},\ and\ \citenamefont {Nakamura}}]{LR_matsuta2017improving}%
  \BibitemOpen
  \bibfield  {author} {\bibinfo {author} {\bibfnamefont {T.}~\bibnamefont {Matsuta}}, \bibinfo {author} {\bibfnamefont {T.}~\bibnamefont {Koma}},\ and\ \bibinfo {author} {\bibfnamefont {S.}~\bibnamefont {Nakamura}},\ }\bibfield  {title} {\bibinfo {title} {Improving the lieb--robinson bound for long-range interactions},\ }in\ \href {https://link.springer.com/article/10.1007/s00023-016-0526-1} {\emph {\bibinfo {booktitle} {Annales Henri Poincar{\'e}}}},\ Vol.~\bibinfo {volume} {18}\ (\bibinfo {organization} {Springer},\ \bibinfo {year} {2017})\ pp.\ \bibinfo {pages} {519--528}\BibitemShut {NoStop}%
\bibitem [{\citenamefont {Else}\ \emph {et~al.}(2020)\citenamefont {Else}, \citenamefont {Machado}, \citenamefont {Nayak},\ and\ \citenamefont {Yao}}]{LRbound_PhysRevA.101.022333}%
  \BibitemOpen
  \bibfield  {author} {\bibinfo {author} {\bibfnamefont {D.~V.}\ \bibnamefont {Else}}, \bibinfo {author} {\bibfnamefont {F.}~\bibnamefont {Machado}}, \bibinfo {author} {\bibfnamefont {C.}~\bibnamefont {Nayak}},\ and\ \bibinfo {author} {\bibfnamefont {N.~Y.}\ \bibnamefont {Yao}},\ }\bibfield  {title} {\bibinfo {title} {{Improved Lieb-Robinson bound for many-body Hamiltonians with power-law interactions}},\ }\href {https://doi.org/10.1103/PhysRevA.101.022333} {\bibfield  {journal} {\bibinfo  {journal} {Phys. Rev. A}\ }\textbf {\bibinfo {volume} {101}},\ \bibinfo {pages} {022333} (\bibinfo {year} {2020})}\BibitemShut {NoStop}%
\bibitem [{\citenamefont {Kuwahara}\ and\ \citenamefont {Saito}(2020)}]{LR_PhysRevX.10.031010}%
  \BibitemOpen
  \bibfield  {author} {\bibinfo {author} {\bibfnamefont {T.}~\bibnamefont {Kuwahara}}\ and\ \bibinfo {author} {\bibfnamefont {K.}~\bibnamefont {Saito}},\ }\bibfield  {title} {\bibinfo {title} {Strictly linear light cones in long-range interacting systems of arbitrary dimensions},\ }\href {https://doi.org/10.1103/PhysRevX.10.031010} {\bibfield  {journal} {\bibinfo  {journal} {Phys. Rev. X}\ }\textbf {\bibinfo {volume} {10}},\ \bibinfo {pages} {031010} (\bibinfo {year} {2020})}\BibitemShut {NoStop}%
\bibitem [{\citenamefont {Tran}\ \emph {et~al.}(2021)\citenamefont {Tran}, \citenamefont {Guo}, \citenamefont {Baldwin}, \citenamefont {Ehrenberg}, \citenamefont {Gorshkov},\ and\ \citenamefont {Lucas}}]{LR_PhysRevLett.127.160401}%
  \BibitemOpen
  \bibfield  {author} {\bibinfo {author} {\bibfnamefont {M.~C.}\ \bibnamefont {Tran}}, \bibinfo {author} {\bibfnamefont {A.~Y.}\ \bibnamefont {Guo}}, \bibinfo {author} {\bibfnamefont {C.~L.}\ \bibnamefont {Baldwin}}, \bibinfo {author} {\bibfnamefont {A.}~\bibnamefont {Ehrenberg}}, \bibinfo {author} {\bibfnamefont {A.~V.}\ \bibnamefont {Gorshkov}},\ and\ \bibinfo {author} {\bibfnamefont {A.}~\bibnamefont {Lucas}},\ }\bibfield  {title} {\bibinfo {title} {Lieb-robinson light cone for power-law interactions},\ }\href {https://doi.org/10.1103/PhysRevLett.127.160401} {\bibfield  {journal} {\bibinfo  {journal} {Phys. Rev. Lett.}\ }\textbf {\bibinfo {volume} {127}},\ \bibinfo {pages} {160401} (\bibinfo {year} {2021})}\BibitemShut {NoStop}%
\bibitem [{\citenamefont {Poland}\ \emph {et~al.}(2020)\citenamefont {Poland}, \citenamefont {Beer},\ and\ \citenamefont {Osborne}}]{poland2020no}%
  \BibitemOpen
  \bibfield  {author} {\bibinfo {author} {\bibfnamefont {K.}~\bibnamefont {Poland}}, \bibinfo {author} {\bibfnamefont {K.}~\bibnamefont {Beer}},\ and\ \bibinfo {author} {\bibfnamefont {T.~J.}\ \bibnamefont {Osborne}},\ }\bibfield  {title} {\bibinfo {title} {{No Free Lunch for Quantum Machine Learning}},\ }\href {https://arxiv.org/abs/2003.14103} {\bibfield  {journal} {\bibinfo  {journal} {arXiv preprint arXiv:2003.14103}\ } (\bibinfo {year} {2020})}\BibitemShut {NoStop}%
\bibitem [{\citenamefont {Poulin}\ \emph {et~al.}(2011)\citenamefont {Poulin}, \citenamefont {Qarry}, \citenamefont {Somma},\ and\ \citenamefont {Verstraete}}]{QsimSpace_PhysRevLett.106.170501}%
  \BibitemOpen
  \bibfield  {author} {\bibinfo {author} {\bibfnamefont {D.}~\bibnamefont {Poulin}}, \bibinfo {author} {\bibfnamefont {A.}~\bibnamefont {Qarry}}, \bibinfo {author} {\bibfnamefont {R.}~\bibnamefont {Somma}},\ and\ \bibinfo {author} {\bibfnamefont {F.}~\bibnamefont {Verstraete}},\ }\bibfield  {title} {\bibinfo {title} {{Quantum Simulation of Time-Dependent Hamiltonians and the Convenient Illusion of Hilbert Space}},\ }\href {https://doi.org/10.1103/PhysRevLett.106.170501} {\bibfield  {journal} {\bibinfo  {journal} {Phys. Rev. Lett.}\ }\textbf {\bibinfo {volume} {106}},\ \bibinfo {pages} {170501} (\bibinfo {year} {2011})}\BibitemShut {NoStop}%
\bibitem [{\citenamefont {Jamet}\ \emph {et~al.}(2022)\citenamefont {Jamet}, \citenamefont {Agarwal},\ and\ \citenamefont {Rungger}}]{jamet2022quantum}%
  \BibitemOpen
  \bibfield  {author} {\bibinfo {author} {\bibfnamefont {F.}~\bibnamefont {Jamet}}, \bibinfo {author} {\bibfnamefont {A.}~\bibnamefont {Agarwal}},\ and\ \bibinfo {author} {\bibfnamefont {I.}~\bibnamefont {Rungger}},\ }\bibfield  {title} {\bibinfo {title} {{Quantum subspace expansion algorithm for Green's functions}},\ }\href {https://arxiv.org/abs/2205.00094} {\bibfield  {journal} {\bibinfo  {journal} {arXiv preprint arXiv:2205.00094}\ } (\bibinfo {year} {2022})}\BibitemShut {NoStop}%
\bibitem [{\citenamefont {Virtanen}\ \emph {et~al.}(2020)\citenamefont {Virtanen}, \citenamefont {Gommers}, \citenamefont {Oliphant}, \citenamefont {Haberland}, \citenamefont {Reddy}, \citenamefont {Cournapeau}, \citenamefont {Burovski}, \citenamefont {Peterson}, \citenamefont {Weckesser}, \citenamefont {Bright} \emph {et~al.}}]{SciPy}%
  \BibitemOpen
  \bibfield  {author} {\bibinfo {author} {\bibfnamefont {P.}~\bibnamefont {Virtanen}}, \bibinfo {author} {\bibfnamefont {R.}~\bibnamefont {Gommers}}, \bibinfo {author} {\bibfnamefont {T.~E.}\ \bibnamefont {Oliphant}}, \bibinfo {author} {\bibfnamefont {M.}~\bibnamefont {Haberland}}, \bibinfo {author} {\bibfnamefont {T.}~\bibnamefont {Reddy}}, \bibinfo {author} {\bibfnamefont {D.}~\bibnamefont {Cournapeau}}, \bibinfo {author} {\bibfnamefont {E.}~\bibnamefont {Burovski}}, \bibinfo {author} {\bibfnamefont {P.}~\bibnamefont {Peterson}}, \bibinfo {author} {\bibfnamefont {W.}~\bibnamefont {Weckesser}}, \bibinfo {author} {\bibfnamefont {J.}~\bibnamefont {Bright}}, \emph {et~al.},\ }\bibfield  {title} {\bibinfo {title} {{SciPy 1.0: fundamental algorithms for scientific computing in Python}},\ }\href {https://www.nature.com/articles/s41592-019-0686-2} {\bibfield  {journal} {\bibinfo  {journal} {Nature methods}\ }\textbf {\bibinfo {volume} {17}},\ \bibinfo {pages} {261} (\bibinfo {year} {2020})}\BibitemShut {NoStop}%
\bibitem [{qul()}]{qulacs}%
  \BibitemOpen
  \href@noop {} {\bibinfo {title} {{Qulacs}}},\ \Eprint {https://arxiv.org/abs/https://github.com/qulacs/qulacs} {https://github.com/qulacs/qulacs} \BibitemShut {NoStop}%
\bibitem [{\citenamefont {Imada}\ and\ \citenamefont {Miyake}(2010)}]{downfolding_imada2010electronic}%
  \BibitemOpen
  \bibfield  {author} {\bibinfo {author} {\bibfnamefont {M.}~\bibnamefont {Imada}}\ and\ \bibinfo {author} {\bibfnamefont {T.}~\bibnamefont {Miyake}},\ }\bibfield  {title} {\bibinfo {title} {{Electronic structure calculation by first principles for strongly correlated electron systems}},\ }\href {https://journals.jps.jp/doi/10.1143/JPSJ.79.112001} {\bibfield  {journal} {\bibinfo  {journal} {Journal of the Physical Society of Japan}\ }\textbf {\bibinfo {volume} {79}},\ \bibinfo {pages} {112001} (\bibinfo {year} {2010})}\BibitemShut {NoStop}%
\bibitem [{\citenamefont {Giannozzi}\ \emph {et~al.}(2009)\citenamefont {Giannozzi}, \citenamefont {Baroni}, \citenamefont {Bonini}, \citenamefont {Calandra}, \citenamefont {Car}, \citenamefont {Cavazzoni}, \citenamefont {Ceresoli}, \citenamefont {Chiarotti}, \citenamefont {Cococcioni}, \citenamefont {Dabo} \emph {et~al.}}]{ESPRESSO1_giannozzi2009quantum}%
  \BibitemOpen
  \bibfield  {author} {\bibinfo {author} {\bibfnamefont {P.}~\bibnamefont {Giannozzi}}, \bibinfo {author} {\bibfnamefont {S.}~\bibnamefont {Baroni}}, \bibinfo {author} {\bibfnamefont {N.}~\bibnamefont {Bonini}}, \bibinfo {author} {\bibfnamefont {M.}~\bibnamefont {Calandra}}, \bibinfo {author} {\bibfnamefont {R.}~\bibnamefont {Car}}, \bibinfo {author} {\bibfnamefont {C.}~\bibnamefont {Cavazzoni}}, \bibinfo {author} {\bibfnamefont {D.}~\bibnamefont {Ceresoli}}, \bibinfo {author} {\bibfnamefont {G.~L.}\ \bibnamefont {Chiarotti}}, \bibinfo {author} {\bibfnamefont {M.}~\bibnamefont {Cococcioni}}, \bibinfo {author} {\bibfnamefont {I.}~\bibnamefont {Dabo}}, \emph {et~al.},\ }\bibfield  {title} {\bibinfo {title} {{QUANTUM ESPRESSO: a modular and open-source software project for quantum simulations of materials}},\ }\href {https://iopscience.iop.org/article/10.1088/0953-8984/21/39/395502/meta} {\bibfield  {journal} {\bibinfo  {journal} {Journal of physics: Condensed matter}\ }\textbf {\bibinfo {volume} {21}},\
  \bibinfo {pages} {395502} (\bibinfo {year} {2009})}\BibitemShut {NoStop}%
\bibitem [{\citenamefont {Giannozzi}\ \emph {et~al.}(2017)\citenamefont {Giannozzi}, \citenamefont {Andreussi}, \citenamefont {Brumme}, \citenamefont {Bunau}, \citenamefont {Nardelli}, \citenamefont {Calandra}, \citenamefont {Car}, \citenamefont {Cavazzoni}, \citenamefont {Ceresoli}, \citenamefont {Cococcioni} \emph {et~al.}}]{ESPRESSO2_giannozzi2017advanced}%
  \BibitemOpen
  \bibfield  {author} {\bibinfo {author} {\bibfnamefont {P.}~\bibnamefont {Giannozzi}}, \bibinfo {author} {\bibfnamefont {O.}~\bibnamefont {Andreussi}}, \bibinfo {author} {\bibfnamefont {T.}~\bibnamefont {Brumme}}, \bibinfo {author} {\bibfnamefont {O.}~\bibnamefont {Bunau}}, \bibinfo {author} {\bibfnamefont {M.~B.}\ \bibnamefont {Nardelli}}, \bibinfo {author} {\bibfnamefont {M.}~\bibnamefont {Calandra}}, \bibinfo {author} {\bibfnamefont {R.}~\bibnamefont {Car}}, \bibinfo {author} {\bibfnamefont {C.}~\bibnamefont {Cavazzoni}}, \bibinfo {author} {\bibfnamefont {D.}~\bibnamefont {Ceresoli}}, \bibinfo {author} {\bibfnamefont {M.}~\bibnamefont {Cococcioni}}, \emph {et~al.},\ }\bibfield  {title} {\bibinfo {title} {{Advanced capabilities for materials modelling with Quantum ESPRESSO}},\ }\href {https://iopscience.iop.org/article/10.1088/1361-648X/aa8f79} {\bibfield  {journal} {\bibinfo  {journal} {Journal of physics: Condensed matter}\ }\textbf {\bibinfo {volume} {29}},\ \bibinfo {pages} {465901} (\bibinfo {year}
  {2017})}\BibitemShut {NoStop}%
\bibitem [{\citenamefont {Giannozzi}\ \emph {et~al.}(2020)\citenamefont {Giannozzi}, \citenamefont {Baseggio}, \citenamefont {Bonf{\`a}}, \citenamefont {Brunato}, \citenamefont {Car}, \citenamefont {Carnimeo}, \citenamefont {Cavazzoni}, \citenamefont {De~Gironcoli}, \citenamefont {Delugas}, \citenamefont {Ferrari~Ruffino} \emph {et~al.}}]{ESPRESSO3_giannozzi2020quantum}%
  \BibitemOpen
  \bibfield  {author} {\bibinfo {author} {\bibfnamefont {P.}~\bibnamefont {Giannozzi}}, \bibinfo {author} {\bibfnamefont {O.}~\bibnamefont {Baseggio}}, \bibinfo {author} {\bibfnamefont {P.}~\bibnamefont {Bonf{\`a}}}, \bibinfo {author} {\bibfnamefont {D.}~\bibnamefont {Brunato}}, \bibinfo {author} {\bibfnamefont {R.}~\bibnamefont {Car}}, \bibinfo {author} {\bibfnamefont {I.}~\bibnamefont {Carnimeo}}, \bibinfo {author} {\bibfnamefont {C.}~\bibnamefont {Cavazzoni}}, \bibinfo {author} {\bibfnamefont {S.}~\bibnamefont {De~Gironcoli}}, \bibinfo {author} {\bibfnamefont {P.}~\bibnamefont {Delugas}}, \bibinfo {author} {\bibfnamefont {F.}~\bibnamefont {Ferrari~Ruffino}}, \emph {et~al.},\ }\bibfield  {title} {\bibinfo {title} {{Quantum ESPRESSO toward the exascale}},\ }\href {https://pubs.aip.org/aip/jcp/article-abstract/152/15/154105/1058748/Quantum-ESPRESSO-toward-the-exascale?redirectedFrom=fulltext} {\bibfield  {journal} {\bibinfo  {journal} {The Journal of chemical physics}\ }\textbf {\bibinfo {volume} {152}}
  (\bibinfo {year} {2020})}\BibitemShut {NoStop}%
\bibitem [{\citenamefont {Perdew}\ \emph {et~al.}(1996)\citenamefont {Perdew}, \citenamefont {Burke},\ and\ \citenamefont {Ernzerhof}}]{PBEapprox_PhysRevLett.77.3865}%
  \BibitemOpen
  \bibfield  {author} {\bibinfo {author} {\bibfnamefont {J.~P.}\ \bibnamefont {Perdew}}, \bibinfo {author} {\bibfnamefont {K.}~\bibnamefont {Burke}},\ and\ \bibinfo {author} {\bibfnamefont {M.}~\bibnamefont {Ernzerhof}},\ }\bibfield  {title} {\bibinfo {title} {{Generalized Gradient Approximation Made Simple}},\ }\href {https://doi.org/10.1103/PhysRevLett.77.3865} {\bibfield  {journal} {\bibinfo  {journal} {Phys. Rev. Lett.}\ }\textbf {\bibinfo {volume} {77}},\ \bibinfo {pages} {3865} (\bibinfo {year} {1996})}\BibitemShut {NoStop}%
\bibitem [{\citenamefont {Hamann}\ \emph {et~al.}(1979)\citenamefont {Hamann}, \citenamefont {Schl\"uter},\ and\ \citenamefont {Chiang}}]{norm_PhysRevLett.43.1494}%
  \BibitemOpen
  \bibfield  {author} {\bibinfo {author} {\bibfnamefont {D.~R.}\ \bibnamefont {Hamann}}, \bibinfo {author} {\bibfnamefont {M.}~\bibnamefont {Schl\"uter}},\ and\ \bibinfo {author} {\bibfnamefont {C.}~\bibnamefont {Chiang}},\ }\bibfield  {title} {\bibinfo {title} {{Norm-Conserving Pseudopotentials}},\ }\href {https://doi.org/10.1103/PhysRevLett.43.1494} {\bibfield  {journal} {\bibinfo  {journal} {Phys. Rev. Lett.}\ }\textbf {\bibinfo {volume} {43}},\ \bibinfo {pages} {1494} (\bibinfo {year} {1979})}\BibitemShut {NoStop}%
\bibitem [{\citenamefont {Hamann}(2013)}]{norm_PhysRevB.88.085117}%
  \BibitemOpen
  \bibfield  {author} {\bibinfo {author} {\bibfnamefont {D.~R.}\ \bibnamefont {Hamann}},\ }\bibfield  {title} {\bibinfo {title} {{Optimized norm-conserving Vanderbilt pseudopotentials}},\ }\href {https://doi.org/10.1103/PhysRevB.88.085117} {\bibfield  {journal} {\bibinfo  {journal} {Phys. Rev. B}\ }\textbf {\bibinfo {volume} {88}},\ \bibinfo {pages} {085117} (\bibinfo {year} {2013})}\BibitemShut {NoStop}%
\bibitem [{\citenamefont {Marzari}\ and\ \citenamefont {Vanderbilt}(1997)}]{Wannier_PhysRevB.56.12847}%
  \BibitemOpen
  \bibfield  {author} {\bibinfo {author} {\bibfnamefont {N.}~\bibnamefont {Marzari}}\ and\ \bibinfo {author} {\bibfnamefont {D.}~\bibnamefont {Vanderbilt}},\ }\bibfield  {title} {\bibinfo {title} {{Maximally localized generalized Wannier functions for composite energy bands}},\ }\href {https://doi.org/10.1103/PhysRevB.56.12847} {\bibfield  {journal} {\bibinfo  {journal} {Phys. Rev. B}\ }\textbf {\bibinfo {volume} {56}},\ \bibinfo {pages} {12847} (\bibinfo {year} {1997})}\BibitemShut {NoStop}%
\bibitem [{\citenamefont {Aryasetiawan}\ \emph {et~al.}(2004)\citenamefont {Aryasetiawan}, \citenamefont {Imada}, \citenamefont {Georges}, \citenamefont {Kotliar}, \citenamefont {Biermann},\ and\ \citenamefont {Lichtenstein}}]{cRPA_PhysRevB.70.195104}%
  \BibitemOpen
  \bibfield  {author} {\bibinfo {author} {\bibfnamefont {F.}~\bibnamefont {Aryasetiawan}}, \bibinfo {author} {\bibfnamefont {M.}~\bibnamefont {Imada}}, \bibinfo {author} {\bibfnamefont {A.}~\bibnamefont {Georges}}, \bibinfo {author} {\bibfnamefont {G.}~\bibnamefont {Kotliar}}, \bibinfo {author} {\bibfnamefont {S.}~\bibnamefont {Biermann}},\ and\ \bibinfo {author} {\bibfnamefont {A.~I.}\ \bibnamefont {Lichtenstein}},\ }\bibfield  {title} {\bibinfo {title} {{Frequency-dependent local interactions and low-energy effective models from electronic structure calculations}},\ }\href {https://doi.org/10.1103/PhysRevB.70.195104} {\bibfield  {journal} {\bibinfo  {journal} {Phys. Rev. B}\ }\textbf {\bibinfo {volume} {70}},\ \bibinfo {pages} {195104} (\bibinfo {year} {2004})}\BibitemShut {NoStop}%
\bibitem [{\citenamefont {Nakamura}\ \emph {et~al.}(2016)\citenamefont {Nakamura}, \citenamefont {Nohara}, \citenamefont {Yosimoto},\ and\ \citenamefont {Nomura}}]{RESPACK1_PhysRevB.93.085124}%
  \BibitemOpen
  \bibfield  {author} {\bibinfo {author} {\bibfnamefont {K.}~\bibnamefont {Nakamura}}, \bibinfo {author} {\bibfnamefont {Y.}~\bibnamefont {Nohara}}, \bibinfo {author} {\bibfnamefont {Y.}~\bibnamefont {Yosimoto}},\ and\ \bibinfo {author} {\bibfnamefont {Y.}~\bibnamefont {Nomura}},\ }\bibfield  {title} {\bibinfo {title} {{Ab initio $GW$ plus cumulant calculation for isolated band systems: Application to organic conductor ${(\mathrm{TMTSF})}_{2}{\mathrm{PF}}_{6}$ and transition-metal oxide ${\mathrm{SrVO}}_{3}$}},\ }\href {https://doi.org/10.1103/PhysRevB.93.085124} {\bibfield  {journal} {\bibinfo  {journal} {Phys. Rev. B}\ }\textbf {\bibinfo {volume} {93}},\ \bibinfo {pages} {085124} (\bibinfo {year} {2016})}\BibitemShut {NoStop}%
\bibitem [{\citenamefont {Nakamura}\ \emph {et~al.}(2009)\citenamefont {Nakamura}, \citenamefont {Yoshimoto}, \citenamefont {Kosugi}, \citenamefont {Arita},\ and\ \citenamefont {Imada}}]{RESPACK2_nakamura2009ab}%
  \BibitemOpen
  \bibfield  {author} {\bibinfo {author} {\bibfnamefont {K.}~\bibnamefont {Nakamura}}, \bibinfo {author} {\bibfnamefont {Y.}~\bibnamefont {Yoshimoto}}, \bibinfo {author} {\bibfnamefont {T.}~\bibnamefont {Kosugi}}, \bibinfo {author} {\bibfnamefont {R.}~\bibnamefont {Arita}},\ and\ \bibinfo {author} {\bibfnamefont {M.}~\bibnamefont {Imada}},\ }\bibfield  {title} {\bibinfo {title} {{Ab initio derivation of low-energy model for $\kappa$-ET type organic conductors}},\ }\href {https://journals.jps.jp/doi/10.1143/JPSJ.78.083710} {\bibfield  {journal} {\bibinfo  {journal} {Journal of the Physical Society of Japan}\ }\textbf {\bibinfo {volume} {78}},\ \bibinfo {pages} {083710} (\bibinfo {year} {2009})}\BibitemShut {NoStop}%
\bibitem [{\citenamefont {Nakamura}\ \emph {et~al.}(2008)\citenamefont {Nakamura}, \citenamefont {Arita},\ and\ \citenamefont {Imada}}]{RESPACK3_nakamura2008ab}%
  \BibitemOpen
  \bibfield  {author} {\bibinfo {author} {\bibfnamefont {K.}~\bibnamefont {Nakamura}}, \bibinfo {author} {\bibfnamefont {R.}~\bibnamefont {Arita}},\ and\ \bibinfo {author} {\bibfnamefont {M.}~\bibnamefont {Imada}},\ }\bibfield  {title} {\bibinfo {title} {{Ab initio derivation of low-energy model for iron-based superconductors LaFeAsO and LaFePO}},\ }\href {https://journals.jps.jp/doi/abs/10.1143/JPSJ.77.093711} {\bibfield  {journal} {\bibinfo  {journal} {Journal of the Physical Society of Japan}\ }\textbf {\bibinfo {volume} {77}},\ \bibinfo {pages} {093711} (\bibinfo {year} {2008})}\BibitemShut {NoStop}%
\bibitem [{\citenamefont {Nohara}\ \emph {et~al.}(2009)\citenamefont {Nohara}, \citenamefont {Yamamoto},\ and\ \citenamefont {Fujiwara}}]{REPACK4_PhysRevB.79.195110}%
  \BibitemOpen
  \bibfield  {author} {\bibinfo {author} {\bibfnamefont {Y.}~\bibnamefont {Nohara}}, \bibinfo {author} {\bibfnamefont {S.}~\bibnamefont {Yamamoto}},\ and\ \bibinfo {author} {\bibfnamefont {T.}~\bibnamefont {Fujiwara}},\ }\bibfield  {title} {\bibinfo {title} {{Electronic structure of perovskite-type transition metal oxides $\text{La}M{\text{O}}_{3}$ $(M=\text{Ti}\ensuremath{\sim}\text{Cu})$ by $\text{U}+\text{GW}$ approximation}},\ }\href {https://doi.org/10.1103/PhysRevB.79.195110} {\bibfield  {journal} {\bibinfo  {journal} {Phys. Rev. B}\ }\textbf {\bibinfo {volume} {79}},\ \bibinfo {pages} {195110} (\bibinfo {year} {2009})}\BibitemShut {NoStop}%
\bibitem [{\citenamefont {Fujiwara}\ \emph {et~al.}(2003)\citenamefont {Fujiwara}, \citenamefont {Yamamoto},\ and\ \citenamefont {Ishii}}]{RESPACK5_fujiwara2003generalization}%
  \BibitemOpen
  \bibfield  {author} {\bibinfo {author} {\bibfnamefont {T.}~\bibnamefont {Fujiwara}}, \bibinfo {author} {\bibfnamefont {S.}~\bibnamefont {Yamamoto}},\ and\ \bibinfo {author} {\bibfnamefont {Y.}~\bibnamefont {Ishii}},\ }\bibfield  {title} {\bibinfo {title} {{Generalization of the iterative perturbation theory and metal--insulator transition in multi-orbital Hubbard bands}},\ }\href {https://journals.jps.jp/doi/10.1143/JPSJ.72.777} {\bibfield  {journal} {\bibinfo  {journal} {Journal of the Physical Society of Japan}\ }\textbf {\bibinfo {volume} {72}},\ \bibinfo {pages} {777} (\bibinfo {year} {2003})}\BibitemShut {NoStop}%
\bibitem [{\citenamefont {Nakamura}\ \emph {et~al.}(2021)\citenamefont {Nakamura}, \citenamefont {Yoshimoto}, \citenamefont {Nomura}, \citenamefont {Tadano}, \citenamefont {Kawamura}, \citenamefont {Kosugi}, \citenamefont {Yoshimi}, \citenamefont {Misawa},\ and\ \citenamefont {Motoyama}}]{RESPACK6_nakamura2021respack}%
  \BibitemOpen
  \bibfield  {author} {\bibinfo {author} {\bibfnamefont {K.}~\bibnamefont {Nakamura}}, \bibinfo {author} {\bibfnamefont {Y.}~\bibnamefont {Yoshimoto}}, \bibinfo {author} {\bibfnamefont {Y.}~\bibnamefont {Nomura}}, \bibinfo {author} {\bibfnamefont {T.}~\bibnamefont {Tadano}}, \bibinfo {author} {\bibfnamefont {M.}~\bibnamefont {Kawamura}}, \bibinfo {author} {\bibfnamefont {T.}~\bibnamefont {Kosugi}}, \bibinfo {author} {\bibfnamefont {K.}~\bibnamefont {Yoshimi}}, \bibinfo {author} {\bibfnamefont {T.}~\bibnamefont {Misawa}},\ and\ \bibinfo {author} {\bibfnamefont {Y.}~\bibnamefont {Motoyama}},\ }\bibfield  {title} {\bibinfo {title} {{RESPACK: An ab initio tool for derivation of effective low-energy model of material}},\ }\href {https://www.sciencedirect.com/science/article/pii/S001046552030391X?via%3Dihub} {\bibfield  {journal} {\bibinfo  {journal} {Computer Physics Communications}\ }\textbf {\bibinfo {volume} {261}},\ \bibinfo {pages} {107781} (\bibinfo {year} {2021})}\BibitemShut {NoStop}%
\bibitem [{\citenamefont {Momma}\ and\ \citenamefont {Izumi}(2011)}]{VESTA_momma2011vesta}%
  \BibitemOpen
  \bibfield  {author} {\bibinfo {author} {\bibfnamefont {K.}~\bibnamefont {Momma}}\ and\ \bibinfo {author} {\bibfnamefont {F.}~\bibnamefont {Izumi}},\ }\bibfield  {title} {\bibinfo {title} {Vesta 3 for three-dimensional visualization of crystal, volumetric and morphology data},\ }\href {https://onlinelibrary.wiley.com/doi/abs/10.1107/S0021889811038970} {\bibfield  {journal} {\bibinfo  {journal} {Journal of applied crystallography}\ }\textbf {\bibinfo {volume} {44}},\ \bibinfo {pages} {1272} (\bibinfo {year} {2011})}\BibitemShut {NoStop}%
\bibitem [{\citenamefont {Jordan}\ and\ \citenamefont {Wigner}(1928)}]{Jordan1928}%
  \BibitemOpen
  \bibfield  {author} {\bibinfo {author} {\bibfnamefont {P.}~\bibnamefont {Jordan}}\ and\ \bibinfo {author} {\bibfnamefont {E.}~\bibnamefont {Wigner}},\ }\bibfield  {title} {\bibinfo {title} {{\"U}ber das paulische {\"a}quivalenzverbot},\ }\href {https://doi.org/10.1007/BF01331938} {\bibfield  {journal} {\bibinfo  {journal} {Zeitschrift f{\"u}r Physik}\ }\textbf {\bibinfo {volume} {47}},\ \bibinfo {pages} {631} (\bibinfo {year} {1928})}\BibitemShut {NoStop}%
\bibitem [{\citenamefont {Wecker}\ \emph {et~al.}(2015)\citenamefont {Wecker}, \citenamefont {Hastings},\ and\ \citenamefont {Troyer}}]{Wecker2015}%
  \BibitemOpen
  \bibfield  {author} {\bibinfo {author} {\bibfnamefont {D.}~\bibnamefont {Wecker}}, \bibinfo {author} {\bibfnamefont {M.~B.}\ \bibnamefont {Hastings}},\ and\ \bibinfo {author} {\bibfnamefont {M.}~\bibnamefont {Troyer}},\ }\bibfield  {title} {\bibinfo {title} {{Progress towards practical quantum variational algorithms}},\ }\href {https://doi.org/10.1103/PhysRevA.92.042303} {\bibfield  {journal} {\bibinfo  {journal} {Phys. Rev. A}\ }\textbf {\bibinfo {volume} {92}},\ \bibinfo {pages} {042303} (\bibinfo {year} {2015})}\BibitemShut {NoStop}%
\bibitem [{\citenamefont {Reiner}\ \emph {et~al.}(2019)\citenamefont {Reiner}, \citenamefont {Wilhelm-Mauch}, \citenamefont {Sch{\"o}n},\ and\ \citenamefont {Marthaler}}]{reiner2019finding}%
  \BibitemOpen
  \bibfield  {author} {\bibinfo {author} {\bibfnamefont {J.-M.}\ \bibnamefont {Reiner}}, \bibinfo {author} {\bibfnamefont {F.}~\bibnamefont {Wilhelm-Mauch}}, \bibinfo {author} {\bibfnamefont {G.}~\bibnamefont {Sch{\"o}n}},\ and\ \bibinfo {author} {\bibfnamefont {M.}~\bibnamefont {Marthaler}},\ }\bibfield  {title} {\bibinfo {title} {{Finding the ground state of the Hubbard model by variational methods on a quantum computer with gate errors}},\ }\href {https://iopscience.iop.org/article/10.1088/2058-9565/ab1e85/meta} {\bibfield  {journal} {\bibinfo  {journal} {Quantum Science and Technology}\ }\textbf {\bibinfo {volume} {4}},\ \bibinfo {pages} {035005} (\bibinfo {year} {2019})}\BibitemShut {NoStop}%
\bibitem [{\citenamefont {Jiang}\ \emph {et~al.}(2018)\citenamefont {Jiang}, \citenamefont {Sung}, \citenamefont {Kechedzhi}, \citenamefont {Smelyanskiy},\ and\ \citenamefont {Boixo}}]{GivensPRA2018}%
  \BibitemOpen
  \bibfield  {author} {\bibinfo {author} {\bibfnamefont {Z.}~\bibnamefont {Jiang}}, \bibinfo {author} {\bibfnamefont {K.~J.}\ \bibnamefont {Sung}}, \bibinfo {author} {\bibfnamefont {K.}~\bibnamefont {Kechedzhi}}, \bibinfo {author} {\bibfnamefont {V.~N.}\ \bibnamefont {Smelyanskiy}},\ and\ \bibinfo {author} {\bibfnamefont {S.}~\bibnamefont {Boixo}},\ }\bibfield  {title} {\bibinfo {title} {{Quantum Algorithms to Simulate Many-Body Physics of Correlated Fermions}},\ }\href {https://doi.org/10.1103/PhysRevApplied.9.044036} {\bibfield  {journal} {\bibinfo  {journal} {Phys. Rev. Appl.}\ }\textbf {\bibinfo {volume} {9}},\ \bibinfo {pages} {044036} (\bibinfo {year} {2018})}\BibitemShut {NoStop}%
\bibitem [{\citenamefont {McClean}\ \emph {et~al.}(2020)\citenamefont {McClean}, \citenamefont {Rubin}, \citenamefont {Sung}, \citenamefont {Kivlichan}, \citenamefont {Bonet-Monroig}, \citenamefont {Cao}, \citenamefont {Dai}, \citenamefont {Fried}, \citenamefont {Gidney}, \citenamefont {Gimby} \emph {et~al.}}]{openfermion2020}%
  \BibitemOpen
  \bibfield  {author} {\bibinfo {author} {\bibfnamefont {J.~R.}\ \bibnamefont {McClean}}, \bibinfo {author} {\bibfnamefont {N.~C.}\ \bibnamefont {Rubin}}, \bibinfo {author} {\bibfnamefont {K.~J.}\ \bibnamefont {Sung}}, \bibinfo {author} {\bibfnamefont {I.~D.}\ \bibnamefont {Kivlichan}}, \bibinfo {author} {\bibfnamefont {X.}~\bibnamefont {Bonet-Monroig}}, \bibinfo {author} {\bibfnamefont {Y.}~\bibnamefont {Cao}}, \bibinfo {author} {\bibfnamefont {C.}~\bibnamefont {Dai}}, \bibinfo {author} {\bibfnamefont {E.~S.}\ \bibnamefont {Fried}}, \bibinfo {author} {\bibfnamefont {C.}~\bibnamefont {Gidney}}, \bibinfo {author} {\bibfnamefont {B.}~\bibnamefont {Gimby}}, \emph {et~al.},\ }\bibfield  {title} {\bibinfo {title} {Openfermion: the electronic structure package for quantum computers},\ }\href {https://iopscience.iop.org/article/10.1088/2058-9565/ab8ebc/meta} {\bibfield  {journal} {\bibinfo  {journal} {Quantum Science and Technology}\ }\textbf {\bibinfo {volume} {5}},\ \bibinfo {pages} {034014} (\bibinfo {year}
  {2020})}\BibitemShut {NoStop}%
\bibitem [{\citenamefont {Bonch-Bruevich}\ and\ \citenamefont {Tyablikov}(2015)}]{bonch2015green}%
  \BibitemOpen
  \bibfield  {author} {\bibinfo {author} {\bibfnamefont {V.~L.}\ \bibnamefont {Bonch-Bruevich}}\ and\ \bibinfo {author} {\bibfnamefont {S.~V.}\ \bibnamefont {Tyablikov}},\ }\href@noop {} {\emph {\bibinfo {title} {{The Green function method in statistical mechanics}}}}\ (\bibinfo  {publisher} {Courier Dover Publications},\ \bibinfo {year} {2015})\BibitemShut {NoStop}%
\bibitem [{\citenamefont {Abrikosov}\ \emph {et~al.}(2012)\citenamefont {Abrikosov}, \citenamefont {Gorkov},\ and\ \citenamefont {Dzyaloshinski}}]{abrikosov2012methods}%
  \BibitemOpen
  \bibfield  {author} {\bibinfo {author} {\bibfnamefont {A.~A.}\ \bibnamefont {Abrikosov}}, \bibinfo {author} {\bibfnamefont {L.~P.}\ \bibnamefont {Gorkov}},\ and\ \bibinfo {author} {\bibfnamefont {I.~E.}\ \bibnamefont {Dzyaloshinski}},\ }\href@noop {} {\emph {\bibinfo {title} {{Methods of quantum field theory in statistical physics}}}}\ (\bibinfo  {publisher} {Courier Corporation},\ \bibinfo {year} {2012})\BibitemShut {NoStop}%
\bibitem [{\citenamefont {Fetter}\ and\ \citenamefont {Walecka}(2012)}]{fetter2012quantum}%
  \BibitemOpen
  \bibfield  {author} {\bibinfo {author} {\bibfnamefont {A.~L.}\ \bibnamefont {Fetter}}\ and\ \bibinfo {author} {\bibfnamefont {J.~D.}\ \bibnamefont {Walecka}},\ }\href@noop {} {\emph {\bibinfo {title} {{Quantum theory of many-particle systems}}}}\ (\bibinfo  {publisher} {Courier Corporation},\ \bibinfo {year} {2012})\BibitemShut {NoStop}%
\bibitem [{\citenamefont {Kubo}(1957)}]{Kubo1957}%
  \BibitemOpen
  \bibfield  {author} {\bibinfo {author} {\bibfnamefont {R.}~\bibnamefont {Kubo}},\ }\bibfield  {title} {\bibinfo {title} {{Statistical-Mechanical Theory of Irreversible Processes. I. General Theory and Simple Applications to Magnetic and Conduction Problems}},\ }\href {https://doi.org/10.1143/JPSJ.12.570} {\bibfield  {journal} {\bibinfo  {journal} {Journal of the Physical Society of Japan}\ }\textbf {\bibinfo {volume} {12}},\ \bibinfo {pages} {570} (\bibinfo {year} {1957})}\BibitemShut {NoStop}%
\bibitem [{\citenamefont {Damascelli}\ \emph {et~al.}(2003)\citenamefont {Damascelli}, \citenamefont {Hussain},\ and\ \citenamefont {Shen}}]{ARPES-review}%
  \BibitemOpen
  \bibfield  {author} {\bibinfo {author} {\bibfnamefont {A.}~\bibnamefont {Damascelli}}, \bibinfo {author} {\bibfnamefont {Z.}~\bibnamefont {Hussain}},\ and\ \bibinfo {author} {\bibfnamefont {Z.-X.}\ \bibnamefont {Shen}},\ }\bibfield  {title} {\bibinfo {title} {{Angle-resolved photoemission studies of the cuprate superconductors}},\ }\href {https://doi.org/10.1103/RevModPhys.75.473} {\bibfield  {journal} {\bibinfo  {journal} {Rev. Mod. Phys.}\ }\textbf {\bibinfo {volume} {75}},\ \bibinfo {pages} {473} (\bibinfo {year} {2003})}\BibitemShut {NoStop}%
\bibitem [{\citenamefont {Fischer}\ \emph {et~al.}(2007)\citenamefont {Fischer}, \citenamefont {Kugler}, \citenamefont {Maggio-Aprile}, \citenamefont {Berthod},\ and\ \citenamefont {Renner}}]{STS-STM-review}%
  \BibitemOpen
  \bibfield  {author} {\bibinfo {author} {\bibfnamefont {O.}~\bibnamefont {Fischer}}, \bibinfo {author} {\bibfnamefont {M.}~\bibnamefont {Kugler}}, \bibinfo {author} {\bibfnamefont {I.}~\bibnamefont {Maggio-Aprile}}, \bibinfo {author} {\bibfnamefont {C.}~\bibnamefont {Berthod}},\ and\ \bibinfo {author} {\bibfnamefont {C.}~\bibnamefont {Renner}},\ }\bibfield  {title} {\bibinfo {title} {{Scanning tunneling spectroscopy of high-temperature superconductors}},\ }\href {https://doi.org/10.1103/RevModPhys.79.353} {\bibfield  {journal} {\bibinfo  {journal} {Rev. Mod. Phys.}\ }\textbf {\bibinfo {volume} {79}},\ \bibinfo {pages} {353} (\bibinfo {year} {2007})}\BibitemShut {NoStop}%
\bibitem [{\citenamefont {Peruzzo}\ \emph {et~al.}(2014)\citenamefont {Peruzzo}, \citenamefont {McClean}, \citenamefont {Shadbolt}, \citenamefont {Yung}, \citenamefont {Zhou}, \citenamefont {Love}, \citenamefont {Aspuru-Guzik},\ and\ \citenamefont {O'brien}}]{VQE_peruzzo2014variational}%
  \BibitemOpen
  \bibfield  {author} {\bibinfo {author} {\bibfnamefont {A.}~\bibnamefont {Peruzzo}}, \bibinfo {author} {\bibfnamefont {J.}~\bibnamefont {McClean}}, \bibinfo {author} {\bibfnamefont {P.}~\bibnamefont {Shadbolt}}, \bibinfo {author} {\bibfnamefont {M.-H.}\ \bibnamefont {Yung}}, \bibinfo {author} {\bibfnamefont {X.-Q.}\ \bibnamefont {Zhou}}, \bibinfo {author} {\bibfnamefont {P.~J.}\ \bibnamefont {Love}}, \bibinfo {author} {\bibfnamefont {A.}~\bibnamefont {Aspuru-Guzik}},\ and\ \bibinfo {author} {\bibfnamefont {J.~L.}\ \bibnamefont {O'brien}},\ }\bibfield  {title} {\bibinfo {title} {{A variational eigenvalue solver on a photonic quantum processor}},\ }\href {https://www.nature.com/articles/ncomms5213} {\bibfield  {journal} {\bibinfo  {journal} {Nat. Commun.}\ }\textbf {\bibinfo {volume} {5}},\ \bibinfo {pages} {4213} (\bibinfo {year} {2014})}\BibitemShut {NoStop}%
\bibitem [{\citenamefont {Zhao}\ \emph {et~al.}(2022)\citenamefont {Zhao}, \citenamefont {Zhou}, \citenamefont {Shaw}, \citenamefont {Li},\ and\ \citenamefont {Childs}}]{TrotErrorAve_PhysRevLett.129.270502}%
  \BibitemOpen
  \bibfield  {author} {\bibinfo {author} {\bibfnamefont {Q.}~\bibnamefont {Zhao}}, \bibinfo {author} {\bibfnamefont {Y.}~\bibnamefont {Zhou}}, \bibinfo {author} {\bibfnamefont {A.~F.}\ \bibnamefont {Shaw}}, \bibinfo {author} {\bibfnamefont {T.}~\bibnamefont {Li}},\ and\ \bibinfo {author} {\bibfnamefont {A.~M.}\ \bibnamefont {Childs}},\ }\bibfield  {title} {\bibinfo {title} {{Hamiltonian Simulation with Random Inputs}},\ }\href {https://doi.org/10.1103/PhysRevLett.129.270502} {\bibfield  {journal} {\bibinfo  {journal} {Phys. Rev. Lett.}\ }\textbf {\bibinfo {volume} {129}},\ \bibinfo {pages} {270502} (\bibinfo {year} {2022})}\BibitemShut {NoStop}%
\bibitem [{\citenamefont {{\c{S}}ahino{\u{g}}lu}\ and\ \citenamefont {Somma}(2021{\natexlab{b}})}]{TrotTight1_csahinouglu2021hamiltonian}%
  \BibitemOpen
  \bibfield  {author} {\bibinfo {author} {\bibfnamefont {B.}~\bibnamefont {{\c{S}}ahino{\u{g}}lu}}\ and\ \bibinfo {author} {\bibfnamefont {R.~D.}\ \bibnamefont {Somma}},\ }\bibfield  {title} {\bibinfo {title} {Hamiltonian simulation in the low-energy subspace},\ }\href {https://www.nature.com/articles/s41534-021-00451-w} {\bibfield  {journal} {\bibinfo  {journal} {npj Quantum Information}\ }\textbf {\bibinfo {volume} {7}},\ \bibinfo {pages} {119} (\bibinfo {year} {2021}{\natexlab{b}})}\BibitemShut {NoStop}%
\bibitem [{\citenamefont {Su}\ \emph {et~al.}(2021)\citenamefont {Su}, \citenamefont {Huang},\ and\ \citenamefont {Campbell}}]{TrotTight2_su2021nearly}%
  \BibitemOpen
  \bibfield  {author} {\bibinfo {author} {\bibfnamefont {Y.}~\bibnamefont {Su}}, \bibinfo {author} {\bibfnamefont {H.-Y.}\ \bibnamefont {Huang}},\ and\ \bibinfo {author} {\bibfnamefont {E.~T.}\ \bibnamefont {Campbell}},\ }\bibfield  {title} {\bibinfo {title} {{Nearly tight Trotterization of interacting electrons}},\ }\href {https://quantum-journal.org/papers/q-2021-07-05-495/} {\bibfield  {journal} {\bibinfo  {journal} {Quantum}\ }\textbf {\bibinfo {volume} {5}},\ \bibinfo {pages} {495} (\bibinfo {year} {2021})}\BibitemShut {NoStop}%
\bibitem [{\citenamefont {Heyl}\ \emph {et~al.}(2019)\citenamefont {Heyl}, \citenamefont {Hauke},\ and\ \citenamefont {Zoller}}]{TrotTight3_heyl2019quantum}%
  \BibitemOpen
  \bibfield  {author} {\bibinfo {author} {\bibfnamefont {M.}~\bibnamefont {Heyl}}, \bibinfo {author} {\bibfnamefont {P.}~\bibnamefont {Hauke}},\ and\ \bibinfo {author} {\bibfnamefont {P.}~\bibnamefont {Zoller}},\ }\bibfield  {title} {\bibinfo {title} {{Quantum localization bounds Trotter errors in digital quantum simulation}},\ }\href {https://www.science.org/doi/10.1126/sciadv.aau8342} {\bibfield  {journal} {\bibinfo  {journal} {Science advances}\ }\textbf {\bibinfo {volume} {5}},\ \bibinfo {pages} {eaau8342} (\bibinfo {year} {2019})}\BibitemShut {NoStop}%
\bibitem [{\citenamefont {Chen}\ \emph {et~al.}(2021)\citenamefont {Chen}, \citenamefont {Huang}, \citenamefont {Kueng},\ and\ \citenamefont {Tropp}}]{TrotTight4_PRXQuantum.2.040305}%
  \BibitemOpen
  \bibfield  {author} {\bibinfo {author} {\bibfnamefont {C.-F.}\ \bibnamefont {Chen}}, \bibinfo {author} {\bibfnamefont {H.-Y.}\ \bibnamefont {Huang}}, \bibinfo {author} {\bibfnamefont {R.}~\bibnamefont {Kueng}},\ and\ \bibinfo {author} {\bibfnamefont {J.~A.}\ \bibnamefont {Tropp}},\ }\bibfield  {title} {\bibinfo {title} {{Concentration for Random Product Formulas}},\ }\href {https://doi.org/10.1103/PRXQuantum.2.040305} {\bibfield  {journal} {\bibinfo  {journal} {PRX Quantum}\ }\textbf {\bibinfo {volume} {2}},\ \bibinfo {pages} {040305} (\bibinfo {year} {2021})}\BibitemShut {NoStop}%
\bibitem [{\citenamefont {An}\ \emph {et~al.}(2021)\citenamefont {An}, \citenamefont {Fang},\ and\ \citenamefont {Lin}}]{TrotTight5_an2021time}%
  \BibitemOpen
  \bibfield  {author} {\bibinfo {author} {\bibfnamefont {D.}~\bibnamefont {An}}, \bibinfo {author} {\bibfnamefont {D.}~\bibnamefont {Fang}},\ and\ \bibinfo {author} {\bibfnamefont {L.}~\bibnamefont {Lin}},\ }\bibfield  {title} {\bibinfo {title} {{Time-dependent unbounded Hamiltonian simulation with vector norm scaling}},\ }\href {https://quantum-journal.org/papers/q-2021-05-26-459/} {\bibfield  {journal} {\bibinfo  {journal} {Quantum}\ }\textbf {\bibinfo {volume} {5}},\ \bibinfo {pages} {459} (\bibinfo {year} {2021})}\BibitemShut {NoStop}%
\bibitem [{\citenamefont {Endo}\ \emph {et~al.}(2018)\citenamefont {Endo}, \citenamefont {Benjamin},\ and\ \citenamefont {Li}}]{PBE_PhysRevX.8.031027}%
  \BibitemOpen
  \bibfield  {author} {\bibinfo {author} {\bibfnamefont {S.}~\bibnamefont {Endo}}, \bibinfo {author} {\bibfnamefont {S.~C.}\ \bibnamefont {Benjamin}},\ and\ \bibinfo {author} {\bibfnamefont {Y.}~\bibnamefont {Li}},\ }\bibfield  {title} {\bibinfo {title} {{Practical Quantum Error Mitigation for Near-Future Applications}},\ }\href {https://doi.org/10.1103/PhysRevX.8.031027} {\bibfield  {journal} {\bibinfo  {journal} {Phys. Rev. X}\ }\textbf {\bibinfo {volume} {8}},\ \bibinfo {pages} {031027} (\bibinfo {year} {2018})}\BibitemShut {NoStop}%
\bibitem [{\citenamefont {Litinski}(2019)}]{litinski2019game}%
  \BibitemOpen
  \bibfield  {author} {\bibinfo {author} {\bibfnamefont {D.}~\bibnamefont {Litinski}},\ }\bibfield  {title} {\bibinfo {title} {{A game of surface codes: Large-scale quantum computing with lattice surgery}},\ }\href {https://quantum-journal.org/papers/q-2019-03-05-128/} {\bibfield  {journal} {\bibinfo  {journal} {Quantum}\ }\textbf {\bibinfo {volume} {3}},\ \bibinfo {pages} {128} (\bibinfo {year} {2019})}\BibitemShut {NoStop}%
\bibitem [{\citenamefont {Cai}(2020)}]{FH_resource_PhysRevApplied.14.014059}%
  \BibitemOpen
  \bibfield  {author} {\bibinfo {author} {\bibfnamefont {Z.}~\bibnamefont {Cai}},\ }\bibfield  {title} {\bibinfo {title} {{Resource Estimation for Quantum Variational Simulations of the Hubbard Model}},\ }\href {https://doi.org/10.1103/PhysRevApplied.14.014059} {\bibfield  {journal} {\bibinfo  {journal} {Phys. Rev. Appl.}\ }\textbf {\bibinfo {volume} {14}},\ \bibinfo {pages} {014059} (\bibinfo {year} {2020})}\BibitemShut {NoStop}%
\bibitem [{\citenamefont {Daley}\ \emph {et~al.}(2022)\citenamefont {Daley}, \citenamefont {Bloch}, \citenamefont {Kokail}, \citenamefont {Flannigan}, \citenamefont {Pearson}, \citenamefont {Troyer},\ and\ \citenamefont {Zoller}}]{qsim1_daley2022practical}%
  \BibitemOpen
  \bibfield  {author} {\bibinfo {author} {\bibfnamefont {A.~J.}\ \bibnamefont {Daley}}, \bibinfo {author} {\bibfnamefont {I.}~\bibnamefont {Bloch}}, \bibinfo {author} {\bibfnamefont {C.}~\bibnamefont {Kokail}}, \bibinfo {author} {\bibfnamefont {S.}~\bibnamefont {Flannigan}}, \bibinfo {author} {\bibfnamefont {N.}~\bibnamefont {Pearson}}, \bibinfo {author} {\bibfnamefont {M.}~\bibnamefont {Troyer}},\ and\ \bibinfo {author} {\bibfnamefont {P.}~\bibnamefont {Zoller}},\ }\bibfield  {title} {\bibinfo {title} {Practical quantum advantage in quantum simulation},\ }\href {https://www.nature.com/articles/s41586-022-04940-6} {\bibfield  {journal} {\bibinfo  {journal} {Nature}\ }\textbf {\bibinfo {volume} {607}},\ \bibinfo {pages} {667} (\bibinfo {year} {2022})}\BibitemShut {NoStop}%
\bibitem [{\citenamefont {Fauseweh}(2024)}]{qsim2_fauseweh2024quantum}%
  \BibitemOpen
  \bibfield  {author} {\bibinfo {author} {\bibfnamefont {B.}~\bibnamefont {Fauseweh}},\ }\bibfield  {title} {\bibinfo {title} {Quantum many-body simulations on digital quantum computers: State-of-the-art and future challenges},\ }\href {https://www.nature.com/articles/s41467-024-46402-9} {\bibfield  {journal} {\bibinfo  {journal} {Nature Communications}\ }\textbf {\bibinfo {volume} {15}},\ \bibinfo {pages} {2123} (\bibinfo {year} {2024})}\BibitemShut {NoStop}%
\bibitem [{\citenamefont {Kanno}\ \emph {et~al.}(2022)\citenamefont {Kanno}, \citenamefont {Endo}, \citenamefont {Utsumi},\ and\ \citenamefont {Tada}}]{DFresource_PhysRevA.106.012612}%
  \BibitemOpen
  \bibfield  {author} {\bibinfo {author} {\bibfnamefont {S.}~\bibnamefont {Kanno}}, \bibinfo {author} {\bibfnamefont {S.}~\bibnamefont {Endo}}, \bibinfo {author} {\bibfnamefont {T.}~\bibnamefont {Utsumi}},\ and\ \bibinfo {author} {\bibfnamefont {T.}~\bibnamefont {Tada}},\ }\bibfield  {title} {\bibinfo {title} {Resource estimations for the hamiltonian simulation in correlated electron materials},\ }\href {https://doi.org/10.1103/PhysRevA.106.012612} {\bibfield  {journal} {\bibinfo  {journal} {Phys. Rev. A}\ }\textbf {\bibinfo {volume} {106}},\ \bibinfo {pages} {012612} (\bibinfo {year} {2022})}\BibitemShut {NoStop}%
\bibitem [{\citenamefont {Clinton}\ \emph {et~al.}(2024)\citenamefont {Clinton}, \citenamefont {Cubitt}, \citenamefont {Flynn}, \citenamefont {Gambetta}, \citenamefont {Klassen}, \citenamefont {Montanaro}, \citenamefont {Piddock}, \citenamefont {Santos},\ and\ \citenamefont {Sheridan}}]{clinton2024towards}%
  \BibitemOpen
  \bibfield  {author} {\bibinfo {author} {\bibfnamefont {L.}~\bibnamefont {Clinton}}, \bibinfo {author} {\bibfnamefont {T.}~\bibnamefont {Cubitt}}, \bibinfo {author} {\bibfnamefont {B.}~\bibnamefont {Flynn}}, \bibinfo {author} {\bibfnamefont {F.~M.}\ \bibnamefont {Gambetta}}, \bibinfo {author} {\bibfnamefont {J.}~\bibnamefont {Klassen}}, \bibinfo {author} {\bibfnamefont {A.}~\bibnamefont {Montanaro}}, \bibinfo {author} {\bibfnamefont {S.}~\bibnamefont {Piddock}}, \bibinfo {author} {\bibfnamefont {R.~A.}\ \bibnamefont {Santos}},\ and\ \bibinfo {author} {\bibfnamefont {E.}~\bibnamefont {Sheridan}},\ }\bibfield  {title} {\bibinfo {title} {Towards near-term quantum simulation of materials},\ }\href {https://arxiv.org/abs/2205.15256} {\bibfield  {journal} {\bibinfo  {journal} {Nature Communications}\ }\textbf {\bibinfo {volume} {15}},\ \bibinfo {pages} {211} (\bibinfo {year} {2024})}\BibitemShut {NoStop}%
\bibitem [{\citenamefont {Childs}\ \emph {et~al.}(2021)\citenamefont {Childs}, \citenamefont {Su}, \citenamefont {Tran}, \citenamefont {Wiebe},\ and\ \citenamefont {Zhu}}]{TrotError_PhysRevX.11.011020}%
  \BibitemOpen
  \bibfield  {author} {\bibinfo {author} {\bibfnamefont {A.~M.}\ \bibnamefont {Childs}}, \bibinfo {author} {\bibfnamefont {Y.}~\bibnamefont {Su}}, \bibinfo {author} {\bibfnamefont {M.~C.}\ \bibnamefont {Tran}}, \bibinfo {author} {\bibfnamefont {N.}~\bibnamefont {Wiebe}},\ and\ \bibinfo {author} {\bibfnamefont {S.}~\bibnamefont {Zhu}},\ }\bibfield  {title} {\bibinfo {title} {{Theory of Trotter Error with Commutator Scaling}},\ }\href {https://doi.org/10.1103/PhysRevX.11.011020} {\bibfield  {journal} {\bibinfo  {journal} {Phys. Rev. X}\ }\textbf {\bibinfo {volume} {11}},\ \bibinfo {pages} {011020} (\bibinfo {year} {2021})}\BibitemShut {NoStop}%
\bibitem [{\citenamefont {Schubert}\ and\ \citenamefont {Mendl}(2023)}]{TrotFH_PhysRevB.108.195105}%
  \BibitemOpen
  \bibfield  {author} {\bibinfo {author} {\bibfnamefont {A.}~\bibnamefont {Schubert}}\ and\ \bibinfo {author} {\bibfnamefont {C.~B.}\ \bibnamefont {Mendl}},\ }\bibfield  {title} {\bibinfo {title} {{Trotter error with commutator scaling for the Fermi-Hubbard model}},\ }\href {https://doi.org/10.1103/PhysRevB.108.195105} {\bibfield  {journal} {\bibinfo  {journal} {Phys. Rev. B}\ }\textbf {\bibinfo {volume} {108}},\ \bibinfo {pages} {195105} (\bibinfo {year} {2023})}\BibitemShut {NoStop}%
\bibitem [{\citenamefont {Verstraete}\ \emph {et~al.}(2009)\citenamefont {Verstraete}, \citenamefont {Cirac},\ and\ \citenamefont {Latorre}}]{fswap_PhysRevA.79.032316}%
  \BibitemOpen
  \bibfield  {author} {\bibinfo {author} {\bibfnamefont {F.}~\bibnamefont {Verstraete}}, \bibinfo {author} {\bibfnamefont {J.~I.}\ \bibnamefont {Cirac}},\ and\ \bibinfo {author} {\bibfnamefont {J.~I.}\ \bibnamefont {Latorre}},\ }\bibfield  {title} {\bibinfo {title} {Quantum circuits for strongly correlated quantum systems},\ }\href {https://doi.org/10.1103/PhysRevA.79.032316} {\bibfield  {journal} {\bibinfo  {journal} {Phys. Rev. A}\ }\textbf {\bibinfo {volume} {79}},\ \bibinfo {pages} {032316} (\bibinfo {year} {2009})}\BibitemShut {NoStop}%
\end{thebibliography}%

\end{document}